\documentclass[10pt]{amsart}
\usepackage{graphicx}
\usepackage{epsf}
\usepackage{amsfonts,amssymb}
\usepackage{amsmath}
\usepackage{fancyhdr}
\usepackage[all,cmtip]{xy}
\usepackage[colorlinks,linkcolor=blue,anchorcolor=blue,citecolor=blue
]{hyperref}
\usepackage{tikz}
\newtheorem{theorem}{Theorem}
\newtheorem*{theorem*}{Theorem}

\newtheorem{Lemma}{Lemma}
\newtheorem{corollary}{Corollary}

\newtheorem*{corollary*}{Corollary}
\newtheorem{proposition}{Proposition}
\newtheorem*{proposition*}{Proposition}

\newtheorem*{claim*}{Claim}

\newtheorem{definition}{Definition}

\theoremstyle{definition}

\theoremstyle{remark}
\newtheorem{remark}{Remark}
\newtheorem*{remark*}{Remark}
\newtheorem*{notation}{Notation}

\theoremstyle{plain}

\newcommand{\Z}{{\mathbb Z}}

\newcommand{\R}{{\mathbb R}}

\newcommand{\N}{{\mathbb N}}

\newcommand{\pa}{\partial}

\newcommand{\beq}{\begin{equation}}
\newcommand{\eeq}{\end{equation}}

\begin{document}

\title{The Dry Ten Martini Problem for $C^2$ cosine-type quasiperiodic Schr\"odinger operators} \maketitle
\begin{center}
\author{
Lingrui Ge \footnote{Beijing International Center for Mathematical Research, Peking University, Beijing, China,  gelingrui@bicmr.pku.edu.cn},\ \  Yiqian Wang \footnote{Department of Mathematics,
Nanjing University, Nanjing, China, yiqianw@nju.edu.cn},\ \  Jiahao Xu \footnote{The corresponding author.}\ \footnote{Department of Mathematics,
Nanjing University, Nanjing, China,   xjhfmmxzy@outlook.com}
  }
\end{center}

\fancyhead{}

\begin{abstract}
This paper solves ``The Dry Ten Martini Problem'' for $C^2$ cosine-type quasiperiodic Schr\"odinger operators with large coupling constants and Diophantine frequencies, a model originally introduced by Sinai in 1987 \cite{sinai}. This shows that the analyticity assumption on the potential is not essential for obtaining a dry Cantor spectrum and can be replaced by a certain geometric condition in the low regularity case.
In addition, we prove the homogeneity of the spectrum and the absolute continuity of the integrated density of states (IDS).
\end{abstract}
\section{Introduction}

We consider the discrete quasiperiodic Schr\"odinger operators on \( \ell^2(\mathbb{Z}) \):
\begin{equation}\label{1.1}
(H_{\alpha,\lambda v,x}u)_n = u_{n+1} + u_{n-1} + \lambda v(x + n\alpha)u_n,
\end{equation}
where \( v \in C^r(\mathbb{R}/\mathbb{Z}, \mathbb{R}) \) with \( r \in \mathbb{N} \cup \{\infty, \omega\} \) is the potential, \( \lambda \in \mathbb{R} \) is the coupling constant, \( x \in \mathbb{T} = \mathbb{R}/\mathbb{Z} \) is the phase, and \( \alpha \in \mathbb{R} \setminus \mathbb{Q} \) is the frequency. The spectral properties of the operator \eqref{1.1} are closely related to the {\it Schr\"odinger cocycle} \( (\alpha, A^{(E - \lambda v)}) \in \mathbb{R}/\mathbb{Z} \times C^r(\mathbb{R}/\mathbb{Z}, SL(2, \mathbb{R})) \), defined by

\begin{equation}\label{sm}
A^{(E - \lambda v)}(x) = \begin{pmatrix}
E - \lambda v(x) & -1 \\
1 & 0
\end{pmatrix}.
\end{equation}
This cocycle defines a family of dynamical systems on \( \mathbb{R}/\mathbb{Z} \times \mathbb{R}^2 \), given by

$$
(x, w) \mapsto (x + \alpha, A^{(E - \lambda v)}(x)w).
$$

It is well-known that the spectrum of the bounded self-adjoint operator \( H_{\alpha, \lambda v, x} \) is a compact set, independent of \( x \) \cite{as}, and is denoted by \( \Sigma_{\alpha, \lambda v} \). Each finite interval of \( \mathbb{R} \setminus \Sigma_{\alpha, \lambda v} \) is called a {\it spectral gap}. These spectral gaps are labeled by integers according to the famous Gap-Labelling Theorem (GLT) \cite{JM}.

In 1981, during an AMS annual meeting, Kac  \cite{Kac} posed the question for the almost Mathieu operator \( H_{\alpha, 2\lambda \cos, x} \):``Are all gaps present?", and offered ten martinis for the solution. Simon \cite{Barry} later popularized this as ``The Ten Martini Problem'', asking whether the spectrum is a Cantor set, and Kac's original problem was subsequently termed ``The Dry Ten Martini Problem'', inquiring whether all gaps allowed by the GLT are open. Both problems have attracted significant attention, with important progress made in \cite{AJ1, AJ2, ak, bs, cey, hs, last, Puig, sinai}. The Ten Martini Problem for almost Mathieu operator was completely resolved by Avila and Jitomirskaya \cite{AJ1}. In comparison, The Dry Ten Martini Problem is more challenging. Recent progress includes results by Avila and Jitomirskaya \cite{AJ2} for non-critical almost Mathieu operators with non-exponentially approximated frequencies, and by Avila, You, and Zhou \cite{AYZ} for any {\it irrational} frequency. For Fibonacci Hamiltonian, Damanik,  Gorodetski and Yessen \cite{DGY} showed that The Dry Ten Martini Problem
 holds true for all values of the coupling constant, which was generalized by Band, Beckus and Loewy to Sturmian Hamiltonian \cite{BBL}.

For a long time, Cantor spectrum results were available for operators \eqref{1.1} with analytic \( v \), but results for non-almost Mathieu families typically required implicit conditions on frequencies or potentials \cite{e, gs, Puig06}. A significant recent advancement is the work by Ge, Jitomirskaya and You \cite{gjy, gjyz}, which provides a solution to ``The Ten Martini Problem'' for type I operators (including several popular explicit models and their analytic neighborhoods) by developing a method independent of arithmetic constraints and almost Mathieu specifics. This method is based on the quantitative global theory established in \cite{gjyz}, with additional applications detailed in \cite{g, gj}.

Lowering the regularity of potentials, the existing results on ``The Ten Martini Problem'' pertain to \( C^2 \) cosine-type quasiperiodic Schr\"odinger operators \cite{sinai, wz2} with large coupling and Diophantine frequencies, where \( v \) satisfies:
\begin{itemize}
\item \( \frac{dv}{dx} = 0 \) at exactly two points, denoted \( x_1 \) and \( x_2 \), which are the minimal and maximal, respectively.
\item These extremals are non-degenerate, i.e., \( \frac{d^2v}{dx^2}(x_j) \neq 0 \) for \( j = 1, 2 \).
\end{itemize}

In this paper, we provide, for the first time, an explicit non-almost Mathieu example with a dry Cantor spectrum under extremely low regularity assumptions on the potential.

We say \( \alpha \in DC_{\tau, \gamma} \) with constants \( \tau \) and \( \gamma\) if
\begin{equation}\label{dioph}
\inf_{j \in \mathbb{Z}} \left| \alpha - \frac{p}{q} - j \right| \geq \frac{\gamma}{|q|^\tau}
\end{equation}
for all \( p, q \in \mathbb{Z} \) with \( q \neq 0 \). It is a standard result that
$$
DC = \bigcup_{\tau > 2, \gamma > 0} DC_{\tau, \gamma}
$$
is of full Lebesgue measure.

\begin{theorem}\label{Th1}
Let \( \alpha \in DC \) and \( v \) be a \( C^2 \) cosine-type potential. There exists a constant \( \lambda_0(\alpha, v) \) such that if \( \lambda \geq \lambda_0 \), then The Dry Ten Martini Problem holds for \( H_{\alpha, \lambda v, x} \) for any \( x \in \mathbb{T} \).
\end{theorem}
\begin{remark}
The Diophantine condition on \( \alpha \) is not essential; we expect the above result holds for all irrational \( \alpha \).
\end{remark}

In addition to resolving ``The Dry Ten Martini Problem'' for cosine-type quasiperiodic operators, our method also provides its quantitative version. By controlling resonances precisely, we achieve detailed quantitative estimates on spectral gaps (see Theorem \ref{15} for specifics). As a corollary, we prove the homogeneity of the spectrum, which is defined as follows:

\begin{definition}
A closed set \( S \subset \mathbb{R} \) is called {\it homogeneous} if there exists \( \mu > 0 \) such that for any \( \text{diam } S > \epsilon > 0 \) and any \( E \in S \),
$$
|S \cap (E - \epsilon, E + \epsilon)| > \mu \epsilon.
$$
\end{definition}

Homogeneity of the spectrum is crucial in inverse spectral theory. For reflectionless Schr\"odinger operators with finite total gap length and homogeneous spectrum, Sodin and Yuditskii \cite{42} proved that the corresponding potential is almost periodic, and Gesztesy and Yuditskii \cite{27} showed that the spectral measure is purely absolutely continuous.

Recent results on spectral homogeneity include Damanik, Goldstein and Lukic \cite{20}, who proved homogeneous spectrum for continuous Schr\"odinger operators with a Diophantine frequency and a sufficiently small analytic potential. For discrete operators in the positive Lyapunov exponent regime, Damanik, Goldstein, Schlag and Voda \cite{22} demonstrated homogeneous spectrum for strong Diophantine \( \alpha \). More recently, Leguil, You, Zhao and Zhou \cite{LYZZ} proved that the spectrum is homogeneous for a (measure-theoretically) typical analytic potential with strong Diophantine \( \alpha \). Moreover,  Ge, You and Zhou \cite{gyzq} provides the exact exponential decay rate for the width of gaps of Almost Mathieu operator for all Diophantine frequency and $|\lambda|\not =1$.

\begin{theorem}\label{homogeneous}
Let \( \alpha \in DC \) and \( v \) be a \( C^2 \) cosine-type potential. There exists a constant \( \lambda_0(\alpha, v) \) such that if \( \lambda \geq \lambda_0 \), the spectrum of \( H_{\alpha, \lambda v, x} \)  is homogeneous for any \( x \in \mathbb{T} \).
\end{theorem}

A key discovery of this paper is the establishment of an equivalence between the distance of the critical points introduced by Wang and Zhang in \cite{wz1} and the fibered rotation number introduced by Johnson and Moser in \cite{JM}. This equivalence also enables us to establish the absolute continuity of the IDS. The IDS is uniformly defined for the Schr\"odinger operators \( (H_{\alpha, \lambda v, x})_{x \in \mathbb{T}} \) by
$$
N(E) = \int_{\mathbb{T}} \mu_x(-\infty, E] \, dx,
$$
where \( \mu_x \) is the spectral measure associated with \( H_{\alpha, \lambda v, x} \) and \( \delta_0 \). Roughly speaking, the density of states measure \( N([E_1, E_2]) \) represents the "number of states per unit volume" with energy between \( E_1 \) and \( E_2 \).

\begin{theorem}\label{idsac}
Let \( \alpha \in DC \) and \( v \) be a \( C^2 \) cosine-type potential. There exists a constant \( \lambda_0(\alpha, v) \) such that if \( \lambda \geq \lambda_0 \), the IDS of \( H_{\alpha, \lambda v, x} \) is absolutely continuous.
\end{theorem}

The regularity of the IDS is a key topic in the spectral theory of quasiperiodic operators, particularly regarding absolute continuity \cite{aviladamanik, avila1, avila2, AJ2, gs2} and H\"older regularity \cite{amor, AJ2, gs1, gs2, gyzh}. It is also closely related to various other topics, such as the relation between absolute continuity of the IDS and purely absolutely continuous spectrum in the zero Lyapunov exponent regime \cite{kotani, damanik}, and the connection between H\"older continuity of the IDS and homogeneity of the spectrum \cite{20, 22, LYZZ}.

One significant challenge in proving the absolute continuity of the IDS is the need to eliminate frequencies \( \alpha \) in a highly implicit manner, due to the problem of ``double resonances." To address this issue, Ge, Jitomirskaya, and Zhao \cite{gjzh} developed a new method to obtain the absolute continuity of IDS for type I operators with fixed Diophantine frequencies, based on Avila's global theory \cite{avila0}. Recently, Ge and Jitomirskaya \cite{gj} removed the Diophantine restriction on the frequency, which plays a crucial role in achieving sharp phase transitions. Xu, Wang, You, and Zhou \cite{xwyz} demonstrated that if \( v \) is a trigonometric polynomial, the IDS of \( H_{\alpha, \lambda v, x} \) is absolutely continuous in the perturbative regime.

\

\section{Main Strategy}

Our results build upon and enhance the work presented in \cite{wz2, sinai}. From \cite{wz2, sinai}, we know that under the assumptions of Theorem \ref{Th1}, the spectrum of \( H_{\alpha, \lambda v, x} \) is a Cantor set. This paper's primary goal is to demonstrate that all gaps are open, as predicted by the Gap-Labelling Theorem \cite{JM}. Note that Cantor spectrum only requires certain properties to hold densely in the spectrum, while the dry version of the Cantor spectrum necessitates handling all energies with rational rotation numbers, leaving no room for energy exclusion.

Given this strict requirement, one might ask whether the \( C^2 \) smoothness condition in the main theorem is too weak compared to the analytic condition (noting that even for analytic potentials, the conclusion might not hold). We will show that the cosine-type geometric condition is sufficiently robust to compensate for the regularity condition's deficiencies. This geometric condition is particularly delicate, and outside the cosine-type case, the dry version of the Cantor spectrum might not hold even for analytic cases. For instance, in Avila's global theory \cite{avila0}, the acceleration of the Lyapunov exponent (LE) for a cosine-type potential is one, whereas for other potentials, it may be greater than one. This indicates that the geometric condition is as critical as the regularity condition and helps to explain why the almost Mathieu operator is so unique among analytic operators. Our paper demonstrates that Young's method can provide insights into both the growth of the transfer matrix's norm (the Lyapunov exponent) and the growth of the transfer matrix's angle (the rotation number).

To prove the dry version of the Cantor spectrum, we need to establish that for each \( k \in \mathbb{Z} \), there exists a spectral gap, denoted \( G_k \), such that for each \( E \in G_k \), the IDS \( N(E) = -k\alpha \pmod{1} \). Typically, computing the IDS is challenging. In fact, \cite{wz2} does not provide detailed information about it. However, \cite{wz2} (along with \cite{sinai} and \cite{wz1}) points out a deep connection between the special dynamical properties of cosine-type cocycles (with large coupling) and the presence of spectral gaps. Specifically, the appearance of spectral gaps corresponds to resonances between the so-called critical points of cocycles, which provides a route to describe the rotation number related to the IDS via \cite{JM}:
$$
N(E) = 1 - 2\rho(E).
$$

Our proof is structured around four key steps:

\textit{\textbf{Step 1:}} Identifying Resonances and Labeling Spectral Gaps.
   As observed by Sinai \cite{sinai} and Wang-Zhang \cite{wz2}, there exists some large $\lambda_0(\alpha,v)>0$ such that for each $\lambda\geq \lambda_0$ and for each (opening) spectral gap \( G \), there exists a unique integer \( k \in \mathbb{Z} \) that characterizes it. Specifically, for every \( E \in \partial G \) (i.e., the boundary of \( G \)),
   $$
   c^n_2(E) - c^n_1(E) - k\alpha \rightarrow 0,
   $$
   where \( c^n_j \) (\( j=1,2 \)) denote the $n$-th critical points. This behavior allows us to assign the integer \( k \) as a unique "label" to each gap \( G \), denoted by \( G_k = G_k^\lambda \) \(\lambda\geq \lambda_0\). Now in the following paper, we always assume that $\lambda_0$ satisfies that for any $\lambda\geq \lambda_0,$ spectral gaps are density in $\R,$ which has been obtained by Wang-Zhang \cite{wz2}.

\textit{\textbf{Step 2:}} Independence of labels on  \( \lambda \).
   Given the definitions in Theorem \ref{Th1}, Theorem \ref{15} later define \( G_k^\lambda \) as the gap labeled \( k \) for the operator \( H_{\alpha, \lambda v, x} \) and \( K(\lambda) \) as the set of all labels for gaps. The following theorem ensures that the set of labels \( K(\lambda) \) remains unaffected by changes in \( \lambda \):

   \begin{theorem}\label{Th2} For any $\lambda\geq \lambda_0,$ it holds that
   \(K(\lambda)=K(\lambda_0)=\Z.\)
   \end{theorem}

\textit{\textbf{Step 3:}} Independence of the Rotation Number on  \( \lambda \).

   \begin{theorem}\label{Th4}
   For any \( k \in \mathbb{Z} \) and $\lambda \geq \lambda_0$, the rotation number \( \rho(G_k^\lambda) \) remains constant as \( \lambda \) varies.
   \end{theorem}

\textit{\textbf{Step 4:}} Calculating the Rotation Number for Large \( \lambda \).
   The primary result of our proof involves finding an explicit expression for the rotation number within the \( k \)-th spectral gap as \( \lambda \) grows large:

   \begin{theorem}\label{Th5}
   For any \( k \in \mathbb{Z} \), we have
   $$
   \lim_{\lambda \to \infty} \rho(G_k^\lambda) = \frac{k\alpha(\rm{mod}~1)}{2}.
   $$
   \end{theorem}

\textbf{The proof of Theorem \ref{Th1}}:
Combining Theorems \ref{Th4} and \ref{Th5}, we conclude that for any \( k \in \mathbb{Z} \) and $\lambda>\lambda_0$, if \( G_{k}^\lambda \), which is opening, denotes the spectral gap labeled by resonance, then
  $G_{k}^\lambda$ is the $k$-th spectral gap labelled by GLT. This completes the proof of Theorem \ref{Th1}. $\square$

\

\textbf{The structure of the proof of Theorem \ref{homogeneous} and Theorem \ref{idsac}}:
The proof of Theorem \ref{homogeneous} relies on our precise characterization of the distances between spectral gaps and the lengths of the spectral gaps. The proof of Theorem \ref{idsac} depends on the equivalence between  the IDS and the distance of the limit-critical points (i.e. $\lim\limits_{n\rightarrow +\infty}c_1^n(E)-c_2^n(E)=N(E)$) provided by Theorem \ref{Th1}.

Throughout this paper, $c, C$ are universal constants satisfying $0<c<1<C$.

\section{Transformation Between SL(2,\(\mathbb{R}\)) Cocycles and Schr\"odinger Cocycles}

For \(\theta \in \mathbb{T}\), let
\[
R_{\theta} = \begin{pmatrix}
\cos{\theta} & -\sin{\theta} \\
\sin{\theta} & \cos{\theta}
\end{pmatrix}
\in SO(2, \mathbb{R}).
\]

In the following, each \(SL(2, \mathbb{R})\)-matrix is assumed to have a norm strictly greater than 1. Define the map
\[
s: SL(2, \mathbb{R}) \rightarrow \mathbb{R}\mathbb{P}^1 = \mathbb{R} / \pi \mathbb{Z}
\]
such that \(s(A)\) represents the most contracted direction of \(A \in SL(2, \mathbb{R})\). That is, for a unit vector \(\hat{s}(A) \in s(A)\), we have \(\|A \cdot \hat{s}(A)\| = \|A\|^{-1}\). Abusing notation slightly, let
\[
u: SL(2, \mathbb{R}) \rightarrow \mathbb{R}\mathbb{P}^1
\]
be defined by \(u(A) = s(A^{-1})\) with \(\hat{u}(A) \in u(A)\). Then for \(A \in SL(2, \mathbb{R})\), it follows that
\[
A = R_{u} \cdot \begin{pmatrix}
\|A\| & 0 \\
0 & \|A\|^{-1}
\end{pmatrix}
\cdot R_{\frac{\pi}{2} - s},
\]
where \(s, u \in [0, \pi)\) are angles corresponding to the directions \(s(A)\) and \(u(A) \in \mathbb{R} / (\pi \mathbb{Z})\).

For a preliminary understanding, we present some basic knowledge of the rotation number. The following proposition provides an equivalent form of the polar decomposition of the Schr\"odinger cocycle map \eqref{sm}, which is more convenient for studying the Lyapunov exponent (note that the Lyapunov exponent is invariant under conjugation).

\begin{proposition}[\cite{z1}]\label{lemma4}
Assume the potential \(v\) is a \(C^2\) \(\cos\)-type function and \(\alpha\) is an irrational number. Then the Schr\"odinger cocycle \((\alpha, A^{(E - \lambda v)})\) is conjugate to the cocycle \((\alpha, A)\) with
\[
A(x, t, \lambda) = \left(
\begin{array}{cc}
\|A(x, t, \lambda)\| & 0 \\
0 & \|A(x, t, \lambda)\|^{-1}
\end{array}
\right) R_{\frac{\pi}{2} - \phi(x, t, \lambda)} \triangleq \Lambda(x, t, \lambda) \cdot R_{\frac{\pi}{2} - \phi(x, t, \lambda)},
\]
where
$$
t = \frac{E}{\lambda}, \quad C_1 \lambda \leq \|A(x, t, \lambda)\| \leq C_2 \lambda, \quad \left|\partial_x^j \|A(x, t, \lambda)\|\right| \leq C_3 \lambda, \quad j = 1, 2,
$$
and \(\tan \phi(x, t, \lambda)\rightarrow t - v(x - \alpha)\) in \(C^2\)-topology as \(\lambda \to \infty\). Thus \(\phi\) is also a \(\cos\)-type function in \(x\) for large \(\lambda\).
\end{proposition}

Fix \(\lambda\) be sufficiently large and define
$$
A(x, t) = \Lambda(x + \alpha, t, \lambda) \cdot R_{\frac{\pi}{2} -\phi(x, t, \lambda)}.
$$
In this paper, we restrict \(t\) to the following interval:
\[
\mathcal{I} := \left[\inf v - \frac{2}{\lambda}, \sup v + \frac{2}{\lambda}\right].
\]

It is well known that $$\frac{1}{\lambda}\Sigma_{\alpha, \lambda v}\subset \mathcal{I}.$$ Therefore, \(t_0 \notin \mathcal{I}\) implies \((\alpha, A(\cdot, t_0))\) is uniformly hyperbolic, indicating that \(t_0\) lies outside the spectrum.

For any $n$, we define
$$
A_n(x,t)=\left\{
                         \begin{array}{ll}
                           A(x+(n-1)\alpha,t)\cdots A(x,t), & \hbox{ $n\geq1$;} \\
                           I_2, & \hbox{ $n=0$;} \\
                           {\left(A_{-n}(x+n\alpha,t)\right)}^{-1}, & \hbox{ $n\leq-1$.}
                         \end{array}
                       \right.
$$

The Lyapunov exponent $L(t)$ of the cocycle is defined as
$$
L(t)= \lim_{n\rightarrow\infty}\frac 1n \int_{\R/\Z} \log\|A_n(x,t)\|dx:=\lim_{n\rightarrow\infty}L_n(t),
$$
where $\|\cdot\|$ denotes the matrix norm in $\text{SL}(2,\mathbb R)$. By Kingman's subadditive Ergodic Theorem, the limit always exists and for irrational $\alpha$
$$
L(t)=\lim_{n\rightarrow\infty}\frac 1n \log\|A_n(x,t)\|:=\lim_{n\rightarrow\infty}L_n(x,t),
\quad\ \ \text{a.e.}\ \  x\in \R/\Z.
$$


\section{New labels of the spectral gaps}\label{new-label}

We know that Johnson-Moser's Gap Labeling Theorem uses rotation numbers to label all spectral gaps (possibly degenerate). In this subsection,  with Sinai's mechanism-``resonances creates gaps'' under the \(C^2\) cosine-type condition, we label each resonance with an integer number $k$ such that $k\alpha$ approximately equal to the distance between two critical points. Then we prove that this labeling is equivalent to Johnson-Moser's labeling. For this purpose, we need a more refined version of Wang-Zhang induction Theorem \cite{wz1}. For the convenience of the readers and for the sake of completeness, we provide the full details of the proof in later sections.

\subsection{Induction Theorem for \(C^2\) Cosine Type}

Let $\{p_n/q_n\}$ be the fraction approximant of $\alpha$. Note we have by \eqref{dioph}
\begin{equation}\label{Calpha}q_{s+1}<Cq_s^{\tau-1},~s\in \Z_+.\end{equation}
Suppose that $N$ is sufficiently large such that
$$\sum_{n\geq 1}q_{N+n-1}^{-\frac1{100}} \leq\frac 1{100}.$$

 We denote
 $$
 s_1(x,t)=s[A_1(x,t)],~u_1(x,t)=s[A_{-1}(x,t)],\ \ t=E/\lambda.
 $$
 From Proposition \ref{lemma4}, for the initial angle function $g_1\in\mathrm{C}^2(\mathbb \R/\Z,\R)$ defined as  $g_1(x,t):=s_1(x,t)-u_1(x,t)$, we have in $C^2$-sense that
\beq\label{g_N}
g_1(x,t)=\phi(x, t, \lambda)+o(\lambda^{-1})=\arctan[t-v(x-\alpha)]+o(\lambda^{-1}),\quad \lambda\rightarrow \infty.
\eeq
In the following, we assume  $\lambda$ is sufficiently large such that
$$\lambda> \max\Big\{e^{100q_{N+1}},C^{10}\Big\}.$$


Let $c_{1,j}(t),\ j=1, 2, \cdots, J,$ be all points on $\mathbb \R/\Z$ minimizing $\{|g_{1,j}(x,t)|(\text{mod}~\pi)\}.$ From (\ref{g_N}) and the cosine-type condition on $v$, we have
$J=1$ or $2$ and if $J=2$ $c_{1,j}(t),\ j=1,\ 2$ is roughly equal to zeroes of $t-v(x-\alpha)$.
For simplicity, we only consider the case $J=2$ and regard the case $J=1$ as a special case of $J=2$ by assuming $c_{1,1}(t)=c_{1,2}(t)$. Denote
$$
C^{(1)}(t)=\{c_{1,1}(t),c_{1,2}(t)\}.
$$
\begin{remark}\label{rmk2} By appropriately translating \( x \) and cosine-type condition, we can assume that $$c_{1,1}(\inf v)=0,c_{1,2}(\inf v)=1;~c_{1,1}(\sup v)=c_{1,2}(\sup v).$$
\end{remark}

 For $j=1,2$, we  define the followings:
\begin{itemize}
   \item The critical interval $I_{1,j}\subset \mathbb T$ centers at the critical point $c_{1,j}$ with a radius of $q_N^{-2000\tau}$ and  $I_1$ is the union of all these $I_{1,j}$. In other words,
      $$I_{1,j}=(c_{1,j}-q_N^{-2000\tau},c_{1,j}+q_N^{-2000\tau}),\quad I_{1}=\bigcup_{1\leq j\leq 2} I_{1,j}.$$
  \item For $t'\in \R,$ we care about the following region called a parameterized critical interval
      $x\in I_{1}(t)$ with $t\in Q_{1}(t')$,
      where $$Q_n(t')=(t'-\lambda^{-q^{\frac{1}{800}}_{N+n-1}},t'+\lambda^{-q^{\frac{1}{800}}_{N+n-1}}),\ n\ge 1.$$
  \item Given $t\in Q_{1}(t')$ and $t'\in \mathcal{I},$  $r^{\pm}_{1}(x,t):I_{1}(t)\rightarrow \mathbb Z_{+}$ is the first return time to $I_1$ {\it after $q_N-1$} under the action
  $x\rightarrow x+\alpha$ and  $x\rightarrow x-\alpha$, respectively.
      In other words,
      $$
      r^{\pm}_{1}(x,t)=\min\{l\in \mathbb Z_+:x\pm l\alpha\in I_{1}\},\quad x\in I_{1}.
      $$
      Let $r^\pm_{1}(t)=\min\limits_{x\in I_{1}(t)}r^\pm_{1}(x,t)$ and $r_{1}(t)=\min\{r^+_{1}(t),r^-_{1}(t)\}$. Let $r_{0}(t)=1$.
\end{itemize}

In the following theorem, given $t\in Q_{1}(t')$ and $t'\in \mathcal{I},$ we will inductively define
$$
C^{(n)} = \{c_{n,1}, c_{n,2}\},\ \ I_{n,j}=(c_{n,j}-q_{N+n-1}^{-2000\tau},c_{n,j}+q_{N+n-1}^{-2000\tau}),\quad I_{n}=\cup_{1\leq j\leq 2} I_{n,j},
$$
and  $r^{\pm}_{n}(x,t):I_{n}(t)\rightarrow \mathbb Z_{+}$ is the first return time to $I_n(t)$ {\it after $q_{N+n-1}-1$}, under the action
  $x\rightarrow x+\alpha$ and  $x\rightarrow x-\alpha$, respectively. Let $r^\pm_{n}(t)=\min\limits_{x\in I_{n}(t)}r^\pm_{n}(x,t)$ and $r_{n}(t)=\min\{r^+_{n}(t),r^-_{n}(t)\}$. Note  that $
r_{n}(t)\geq q_{N+n-1}-1.$

Moreover, we denote
$$s_n(x,t)=s[A_n(x,t)],~u_n(x,t)=s[A_{-n}(x,t)]$$
 and
the angle function $g_{n+1}: I_{n}\rightarrow \mathbb{RP}^1$ is defined by $$g_{n+1}(x,t):=s_{r^+_n}(x,t)-u_{r^-_n}(x,t).$$

Denote $\tilde{I}_{n,j}(t)=\{x\in I_{n,j}||g_{n+1}(x,t)\ {\rm mod\ }\pi|\le \lambda^{-r^{\frac{1}{700}}_{n}}\} ,~j=1,2.$

For the convenience, we sometimes omit the dependence of $C^2(I_n(t))$ on $I_n(t)$, that of $r_n(t)$, $c_{n,j}(t),$ $I_{n,j}(t)$ and $I_n(t)$ on $t$,  and that of $g_n(x,t)$ on $(x,t)$, respectively.
\begin{notation} We say \( \text{Case } {\bf X}_n \) occurs if Case \( {\bf X} \) (defined as below) occurs at step \( n \), where \( {\bf X} \in \{1, 2,\ 3\} \).\end{notation}
\begin{theorem}\label{theorem12} ({\bf The Inductive Theorem})
Let $t' \in \mathcal{I}$ and  $\alpha$ be a Diophantine number. There exist constants $N=N(v, \alpha)$, $\lambda_0=\lambda_0(v, \alpha, N)$, $0<\epsilon=\epsilon(\lambda_0)\ll 1$, $0<c=c(v, \alpha)<1$
 and $C = C(v, \alpha)>1$ such that if
$
\lambda > \lambda_0,
$
then the following statements hold for $n \geq 1$, $t\in Q_n(t')$ and $x \in I_n(t)$:



If $X(t):=\{x\in I_n \vert g_{n+1}(x,t)=0\}$ is not empty, then
\begin{equation}\label{maxmax}|g_{n+1}(x,t)|\geq c\cdot \text{dist}(x,X(t))^3\end{equation}
and
\begin{equation}\label{lm17-main}
\frac{1}{10} <  \frac{\partial g_{n+1}(x,t)}{\partial t}  < \lambda^{10q_{N+n-1}}.
\end{equation}

\vskip 0.2cm

Moreover,  by induction, $g_{n+1}$ lies in one of the following three cases:

\noindent\textbf{Case 1:} $T^k{I_{n,1}} \bigcap I_{n,2}=\emptyset$ for each $0 \le |k| <  q_{N+n-1}$. Then the angle function $g_{n+1}$ satisfies (see the third picture in Figure 1):
\begin{equation}\label{I-zero-derivative}
g_{n+1}(x,t) \in \left[ -\frac{\pi}{200}, \frac{\pi}{200} \right].
\end{equation}
Moreover
$g_{n+1}$ has exactly one zero $c_{n+1,j}$ in each $I_{n,j}$ ($j=1,2$), which satisfies
\beq\label{ga1}
|c_{n+1,j} - c_{n,j}| < \lambda^{-\frac{1}{2} r_{n-1}},
\eeq
 \begin{equation}\label{I-first-derivative}\small
q_{N+n-1}^{-2} \leq |\partial_x g_{n+1}(x,t)| \leq q_{N+n-1}^{2}, \ \ \partial_x g_{n+1}(x_1,t) \cdot \partial_x g_{n+1}(x_2,t)\le 0,\ x_j\in I_{n,j},
\end{equation}
\begin{equation}\label{I-second-derivative}
|\partial^2_x g_{n+1}(x,t)| \leq \lambda^{10q_{N+n-1}}.
\end{equation}

\vskip 0.4cm

\noindent\textbf{Case {\bf 2}:} There exists (a unique) $0 \le |k| < q_{N+n-1}$ such that $T^k{I_{n,1}} \bigcap I_{n,2}\not=\emptyset.$ Then
\begin{equation}\label{ga110} r_{n}(t)\geq q_{N+n-1}^{2000}.
\end{equation}
Moreover, $g_{n+1}$ has the following properties (see the first picture in Figure 1):
\begin{enumerate}
\item $\tilde{I}_{n,j}$ composes of one or two intervals and the minimum point set of $|g_{n+1}{\rm\ mod}\ \pi|$ on $\tilde{I}_{n,j}$ composes of one or two elements.  We denote this set by
$$
C^{(n+1,j)}=\{c_{n+1,j},\ \ c'_{n+1,j}\}, \footnote{For the one-element case, we assume $c_{n+1,j}=c'_{n+1,j}$,}
$$
where $c_{n+1,j}$ is determined by the condition
$$
|c_{n+1,i}-c_{n,i}|\le |c'_{n+1,i}-c_{n,i}|.
$$
It holds that \begin{equation}\label{cc}
\ \inf\limits_{j\in \Z}\left|c_{n+1,1} + k\alpha - c'_{n+1,1}-j\right|, \quad \inf\limits_{j\in \Z}\left|c_{n+1,2} - k\alpha - c'_{n+1,2}-j\right| < \lambda^{-\frac{1}{30} r_{n}}.
\end{equation}
\beq\label{ga2}
|c_{n+1,j} - c_{n,j}| < \lambda^{-\frac{1}{2} r_{n-1}}.
\eeq
\item $\pa_xg_{n+1}\vert_{I_{n+1,j}}$ has one or two zeros, which has to be the maximum or the minimum point, denoted by $\{\tilde{c}'_{n+1,j},\tilde{c}_{n+1,j}\}$. \footnote{For the one-element case, we assume $\tilde{c}_{n+1,j}=\tilde{c}'_{n+1,j}$.} Moreover,\begin{equation}\label{either12}\text{either}~ \tilde{c}'_{n+1,1}\leq c'_{n+1,2}\leq \tilde{c}_{n+1,1}\leq c_{n+1,1}; c_{n+1,2}\leq \tilde{c}_{n+1,2}\leq c'_{n+1,1}\leq \tilde{c}'_{n+1,2};\end{equation}
 $$\text{or}~c_{n+1,1}\leq \tilde{c}'_{n+1,1}\leq c'_{n+1,1}\leq \tilde{c}_{n+1,1}; \tilde{c}_{n+1,2}\leq c'_{n+1,1}\leq \tilde{c}'_{n+1,1}\leq c_{n+1,2}.$$
\begin{equation}\label{tic}|\tilde{c}_{n+1,j}-\tilde{c}'_{n+1,j}|<\lambda^{-\frac{1}{2}|k|}.\end{equation}
\begin{equation}\label{maxmax1}\max_{\tilde{I}_{n,j}}|\partial^2_x g_{n+1}| \le  C\lambda^{10|k|}.\end{equation}
\item We have that
\begin{equation}\label{range-g-n}
 \max_{x \in I_{n,j}} g_{n+1}(x,t) - \min_{x \in I_{n,j}} g_{n+1}(x,t) \leq \pi - c\lambda^{-100|k|}, \quad j = 1, 2.
\end{equation}
If $\tilde{c}_{n+1,j}\not=\tilde{c}'_{n+1,j},$ then
\begin{equation}\label{range-g-n-lower-bound}
\pi - C\lambda^{-\frac{1}{100}|k|} \le\max_{x \in I_{n,j}} g_{n+1}(x,t) - \min_{x \in I_{n,j}} g_{n+1}(x,t).
\end{equation}
\item It holds that
\begin{equation}\label{lm17-main1}
\frac{1}{10} <  \frac{\partial g_{n+1}(x,t)}{\partial t}  < C\lambda^{5|k|}.
\end{equation}
\end{enumerate}

\vskip 0.4cm

\noindent\textbf{Case {\bf 3}:}(a subcase of Case \textbf{2}, see the third picture in Figure 2) $\tilde{I}_{n,j}$ is empty, i.e.,$$
\min\limits_{x \in I_{n,j}} |g_{n+1,j}(x,t)(\text{mod}~\pi)| > 2\lambda^{-r^{\frac{1}{50}}_n}.
$$
In this case, we redefine $I_{m,j}=I_{n,j}$ for any $m\ge n$.
Then (with $k$ defined as in Case {\bf 2}), it holds that
\begin{equation}\label{cn1cn2} \inf\limits_{j\in \Z}\vert c_{n,1}+k\alpha-c_{n,2}-j\vert\leq \lambda^{-r_{n-1}^{\frac{1}{55}}}\end{equation} and
\begin{equation}\label{Case3lem}t\notin \Sigma^{\lambda}. \end{equation} In particular, for $|k|>0,$
\begin{equation}\label{g_nmin1}~c_{n,1}+k\alpha-c_{n,2}=0~\text{implies}~t\notin \Sigma^{\lambda}.\end{equation}



\vskip 0.3cm
\noindent Finally, if both step \( n \) and step \( n+1 \) belong to Case \textbf{1}~or belong to Case \textbf{2} with the same $k$, then
\begin{equation}\label{gn-gn+1}\|g_n-g_{n+1}\|_{C^2}\leq \lambda^{-\frac{3}{2}r_{n-1}}.\end{equation}
\end{theorem}

\begin{remark}
We remark that if $k=0$, then ${I_{n,1}} \bigcap I_{n,2}\not=\emptyset$ implies $c'_{n+1,1}=c_{n+1,2}$ and $c'_{n+1,2}=c_{n+1,1}$. Hence (\ref{cc}) is obvious. In this case,  $I_n$ becomes a  {consecutive} interval.
Moreover, the angle function $g_{n+1}$ satisfies $\|g_{n+1}-g_N\|_{C^2}\leq \lambda^{-1}$.
\end{remark}

\begin{remark}\label{appearance-gap}
If Case {\bf 3} occurs for  $t$, then   the cocycle corresponding to this $t$ is uniformly hyperbolic by \cite{yoc} and \cite{z1}. In other words, $t$ is in some spectral gap. In particular,  in Case {\bf 3}, $k=0$ corresponds to $(-\infty, \inf \frac{1}{\lambda}\Sigma_{\alpha, \lambda v})\bigcup (\sup \frac{1}{\lambda}\Sigma_{\alpha, \lambda v}, +\infty)$, respectively.
\end{remark}

\begin{remark}\label{rmk4} It is worth noting that in Theorem \ref{theorem12}, when \(\min |g_n(x) \mod\pi| > 0\), \(c_n\) is defined as the \(x\) that minimizes \(|g_n(x)\mod\pi |\). In particular, when \(g_n(\tilde{c}_{n,j}) = \pi - g_n(\tilde{c}'_{n,j})\), both \(\tilde{c}_{n,j}\) and \(\tilde{c}'_{n,j}\) minimize \(|g_n(x)| \mod \pi\),~where $\tilde{c}_{n,j}$ and $\tilde{c}'_{n,j}$ are zeros of $\partial_x g_n.$ In this case, \(c_{n,j}(t)\) is not well-defined. Case \textbf{3} guarantees that  \( c_{n,j},~j=1,2 \) is always well-defined.
\end{remark}

\subsection{Some remarks on the Induction Theorem}

The Induction Theorem has been proved in \cite{wz1}, \cite{wz2}, \cite{LWY} or \cite{XGW}. For the sake of readers, we put the proof in the Section \ref{proofofinduction}.

Now, we provide some intuitive explanations to help readers better understand the Induction Theorem.

Obviously, the Inductive Theorem holds for $n=1$  with the cosine-type assumption {if we assume that $N$ and $\lambda$ are sufficiently large.}

For simplicity we assume that Case {\bf 1} in the Inductive Theorem holds for the $n$-th step ($n\geq 1$) (note that for step $1$, Case {\bf 2}
occurs only if $t$ is close to the extreme values of $v$), since the main difficulty lies in this situation.
We aim to study the $(n+1)$-th step.
Let $x\in I_{n+1,1}$.
We can rewrite $A_{r^+_{n+1}}(x,t)$ as
\beq\label{product31}
A_{r^+_{n+1}}(x,t)=A_{t_{s}-t_{s-1}}(T^{ t_{s-1}}x,t)\cdots A_{ t_{j}-t_{j-1}}(T^{ t_{j-1}}x,t)\cdots A_{t_{1}-t_0}(x,t),
\eeq
where $T^{ t_i}x\in I_{n}$ and $t_{j+1}(x,t)-t_{j}(x,t)=r^+_{n}(T^{t_j}x,t)$ ($0\leq j\leq s-1$), $t_0=0$ and
$r^+_{n+1}=t_{s}$  ($T^{ t_s}\in I_{n+1}$).

By the Diophantine condition, there exists at most one $t_l<t_s$ with $t_l\ll q_{n+1}$ such that $T^{ t_l}x\in I_{n+1}$. If there exist no such  $l$, we will go to Case {\bf 1} for the $(n+1)$-th step.
If such $l$ exists, by the Diophantine condition, we have $t_s\gg t_l$ and we will go to the Case {\bf 2} for the $(n+1)$-th step with $k=t_l$, see Figure 1.

\subsubsection{A baby model}
To study (\ref{product31}), we have to examine the following baby model (which essentially encompasses all scenarios in the iterative process)
\beq\label{product2} A_{n_1+n_2}(x,t)=A_{n_2}(T^{n_1}x,t)\cdot A_{n_1}(x,t),\eeq
where $x\in I_{n+1,1}$ and $T^{n_1}x,\ T^{n_1+n_2}x=T^{n_2}(T^{n_1}x)\in I_n$.

Let $g_{n}(x,t)=s_{n_1}(x,t)$, $g_{n+1}(x,t)=s_{n_1+n_2}(x,t)$, $\theta(x,t)=s_{n_2}(T^{n_1}x,t)-u_{n_1}(T^{n_1}x,t)$, $l_{n_1+n_2}=\|A_{n_1+n_2}(x,t)\|$, $l_{n_1}=\|A_{n_1}(x,t)\|$ and $l_{n_2}=\|A_{n_2}(T^{n_1}x,t)\|$. One can see that $\theta$ plays a key role in the growth of the norm of product of matrices. In fact, if $\min\{\|A_{n_1}\|^{-1}, \|A_{n_2}\|^{-1}\}\ll|\theta|\ll 1$, then \beq\label{norm-non-resonance}\|A_{n_1+n_2}(x,t)\|\approx \|A_{n_1}(x,t)\|\cdot
\|A_{n_2}(x,t)\|\cdot |\theta|.\eeq
Hence $\|A_{n_1+n_2}(x,t)\|\gg \max\{\|A_{n_1}(x,t),
\|A_{n_2}(x,t)\|\}.$

 However, if $|\theta|\ll\min\{\|A_{n_1}\|^{-1}, \|A_{n_2}\|^{-1}\}$,
then
\beq\label{norm-resonance}
\|A_{n_1+n_2}(x,t)\|\approx  \|A_{n_2}(x,t)\|/
\|A_{n_1}(x,t)\|,
\eeq
which is much smaller than $\max\{\|A_{n_2}(x,t)\|,
\|A_{n_1}(x,t)\|\}$.

Thus we need a sharp estimate on $\theta,$ in particular we need an upper bound for the measure of $x$ such that $\theta$ is close to $0$ in some sense.  With the assumption that $\min\{\|A_{n_1}\|^{-1}, \|A_{n_2}\|^{-1}\}\ll|\theta|\ll 1$ or $n_2\gg n_1$, roughly we have that
\begin{equation}\label{g_{n+1}}
 g_{n+1}(x, t)\approx\arctan(l_{n_1}^2(x,t) \tan(\theta(x, t)))+\frac\pi2+g_{n}(x,t)\ {\rm (mod\ \pi)}.
\end{equation}

First we study the function
\begin{equation}\label{phi-theta}
 \phi_{n+1}(\theta)=\arctan(l_{n_1}^2 \tan\theta)+\frac\pi2\ {\rm (mod\ \pi)},\quad \theta\in (-\frac\pi2,\ \frac\pi2).
\end{equation}

The shape of $\phi_{n+1}(\theta)$ satisfies the following properties (see Figure 3):
\begin{enumerate}
\item $\phi_{n+1}(\theta)$ monotonically increases from $0$ to $\pi$ as $\theta$ changes from $-\frac\pi2$ to $\frac\pi2$.
\item \ $\phi_{n+1}(0)=\frac{\pi}{2}$ and the curve $\phi=\phi_{n+1}(\theta)$ is symmetric with respect to the point $(0, \phi_{n+1}(0))$.
\item $|\theta|\gtrsim l_{n_1}^{-2}$ and $0\le i\le 2$, it holds that $|\partial^i_x\phi_{n+1}(\theta)|\approx l_{n_1}^{-2}|\theta|^{-(i+1)}$.
\item (2) and (3) imply that for any fixed $2c>0$, it holds that on $|\theta|<l_{n_1}^{-2+c}$, $\phi_{n+1}(\theta)$ increases from
  $c_1l_{n_1}^{-c}$ to $\pi-c_2l_{n_1}^{-c}$ for some universal constants $c_1, c_2$. For simplicity, in the following we  omit $c_1, c_2$ and just say $\phi_{n+1}(\theta)$ increases from
  $l_{n_1}^{-c}$ to $\pi-l_{n_1}^{-c}$.  On the other hand, for $0\le \theta\le -l_{n_1}^{-2+c}$ or $l_{n_1}^{-2+c}\le \theta\le \frac\pi2$, it holds that $0\le \phi_{n+1}(\theta)<l_{n_1}^{-c}$ or $\pi-l_{n_1}^{-c}\le \phi_{n+1}(\theta)\le\pi$.
\item $\phi'_{n+1}(\theta)\ge l_{n_1}^{4\mu}\gg 1$ on $l_{n_1}^{-2}\lesssim|\theta|<l_{n_1}^{-1-2\mu}$, while $0<\phi'_{n+1}(\theta)\le l_{n_1}^{-4\mu}\ll 1$ on $|\theta|\ge l_{n_1}^{-1+2\mu}$.
Moreover, $|\phi''_{n+1}(\theta)|\ge l_{n_1}^{6\mu}\gg 1$ on $l_{n_1}^{-2}\lesssim|\theta|<l_{n_1}^{-2/3-2\mu}$, while $|\phi''_{n+1}(\theta)|\le l_{n_1}^{-6\mu}\ll 1$ on $|\theta|\ge l_{n_1}^{-2/3+2\mu}$.
\end{enumerate}

Next we analyze the shape of $\tilde\phi(x,t)=\arctan(l_{n_1}^2(x,t) \tan(\theta(x,t)))+\frac\pi2$, where $x\in I_{n+1,1}$ centered at $c_n$
with $l_{n_1}^{-1}\ll |I_{n+1}|$.

From the inductive assumption, we can assume $|\pa_x^j\theta(x, t)|\leq l_{n_1}^{j\eta}$ and $|\pa_x^j l_{n_1}|\leq |l_{n_1}|^{1+j\eta}$ ($j=0,1,2$) from (\ref{norm-derivative}) and (\ref{I-second-derivative}). Moreover from (\ref{I-first-derivative}), we have $|\pa_xg|$ is far from $0$ and in particular $\pa_xg(x,t)$ and $\pa_x\theta(x,t)$ possess opposite signs.   
 Without loss of generality, we assume $\pa_xg(x,t)<0$. Then  $\pa_x\theta(x,t)>0$, which implies the shape of $\tilde\phi(x,t)$ (with a fixed $t$) satisfy similar properties as $\phi(\theta)$ (see (1)-(5)).




Now we go back to (\ref{g_{n+1}}), that is $g_{n+1}(x, t)=\tilde{\phi}(x,t)+g_{n}(x,t)$.
 Assume the only two zeros of $g_n$ on $I_{n+1}$ are $c_n\in I_{n+1,1}$ and $c'_n\in I_{n+1,2}$.
Note that $\theta(x,t)\thickapprox g_n(x+n_1\alpha,t).$ From the above argument, we know that the shape of $g_{n+1}$ depends on whether $g_{n}$ possesses one zero on $I_{n+1,1}+n_1\alpha$,
or equivalently, depends on whether $c'_n\in n_1\alpha+I_{n+1,1} \ (mod\ 1)$.

If $c_n+n_1\alpha$ is close to $c'_{n}$ in the sense of $\mod 1$ (in this case $n_1=t_l$ with $t_l<t_s$), the jump part of $\tilde \phi$ is contained in $I_{n+1,1}$.
If $c_n+n_1\alpha$ is far from $c'_{n}$ in the sense of $\mod 1$ (in this case there exists no such a $k$), the jump part of $\tilde \phi$ moves outside $I_{n+1,1}$, and $\tilde \phi$ in $I_{n+1,1}$ nearly vanishes in $C^2$ sense, see Figure 3.
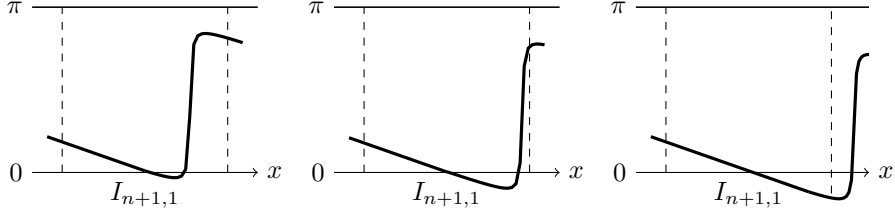
\begin{figure}
\begin{tikzpicture}[yscale=0.7]
\draw [->] (-1.5,0) -- (1.5,0);
\draw [very thick,domain=-0:1.3] plot (\x, {1/180*pi*atan(50*tan((\x -0.6) r))+0.5*pi-0.5*(\x)});
\draw [very thick,domain=-1.3:0] plot (\x, {0.03-0.5*(\x)});
\draw [semithick,domain=-1.5:1.5] plot (\x, {0*\x+pi});
\draw [dashed] (-1.1,0) -- (-1.1,pi);
\draw [dashed] (1.1,0) -- (1.1,pi);
\node [below] at (0,0) {$I_{n+1,1}$};
\node [right] at (1.5,0) {$x$};
\node [left] at (-1.5,0) {$0$};
\node [left] at (-1.5,3.14) {$\pi$};
\end{tikzpicture}
\begin{tikzpicture}[yscale=0.7]
\draw [->] (-1.5,0) -- (1.5,0);
\draw [very thick,domain=-0:1.3] plot (\x, {1/180*pi*atan(50*tan((\x -1) r))+0.5*pi-0.5*(\x)});
\draw [very thick,domain=-1.3:0] plot (\x, {0.01-0.5*(\x)});
\draw [semithick,domain=-1.5:1.5] plot (\x, {0*\x+pi});
\draw [dashed] (-1.1,0) -- (-1.1,pi);
\draw [dashed] (1.1,0) -- (1.1,pi);
\node [below] at (0,0) {$I_{n+1,1}$};
\node [right] at (1.5,0) {$x$};
\node [left] at (-1.5,0) {$0$};
\node [left] at (-1.5,3.14) {$\pi$};
\end{tikzpicture}
\begin{tikzpicture}[yscale=0.7]
\draw [->] (-1.5,0) -- (1.6,0);
\draw [very thick,domain=0.8:1.6] plot (\x, {1/180*pi*atan(50*tan((\x -1.4) r))+0.5*pi-0.5*(\x)});
\draw [very thick,domain=-1.3:0.8] plot (\x, {{0.03-0.5*(\x)}});
\draw [semithick,domain=-1.5:1.6] plot (\x, {0*\x+pi});
\draw [dashed] (-1.1,0) -- (-1.1,pi);
\draw [dashed] (1.1,-0.4) -- (1.1,pi);
\node [below] at (0,0) {$I_{n+1,1}$};
\node [right] at (1.6,0) {$x$};
\node [left] at (-1.5,0) {$0$};
\node [left] at (-1.5,3.14) {$\pi$};
\end{tikzpicture}
\caption{Graphs of $g_{n+1}$ in $I_{n+1,1}$: the jump part may move outside $I_{n+1,1}$ if $|d_{n}|$ increases as $t$ changes. The first and the third pictures correspond to Case {\bf 2} and {\bf 1}, respectively.}
\end{figure}

\vskip 0.2cm

Thus a key variable is $$d_n=c_n+n_1\alpha-c'_n (\mod\ 1).$$ For this purpose, we  give the following definitions.
 \begin{definition}
Let $d_{n}=c_n+n_1\alpha-c'_n (\mod1)$.
We say that the step $n$ is {\bf non-resonant} (corresponding to Case \textbf{1}), if $|d_{n}|\ge 3q_{N+n-1}^{-2000\tau}$ for each $n_1$ satisfying $0\leq|n_1|<q_{N+n-1}$.
Otherwise, if there exists a (unique) $n_1<q_{N+n-1}$ such that  we say that the step $n$ is {\bf resonant} (corresponding to Case \textbf{2}). From the condition $n_1+n_2\thicksim q_n$, we have $n_2\gg n_1$.
 \end{definition}
 \vskip -0.2cm

 For the non-resonance case, the jump part of $\tilde\phi$ is not contained in $n_1\alpha+I_{n+1,1}$.
 Hence $\tilde{\phi}$ is small and the shape of $g_{n+1}$ is similar as the one of $g_n$, i.e., Case {\bf 1} occurs for $g_{n+1}$, see the third one in Figure 1.

For the resonance case, the jump part is really contained in $n_1\alpha+I_{n+1,1}$. Note that $g'_n$ possess opposite signs on $I_{n+1,1}$ and $n_1\alpha+I_{n+1,1}$, respectively.
Thus the shape of $g_{n+1}=\tilde\phi(x,t)+g_{n}$ looks like the third one in Figure 1 or the first one in Figure 2.
It can be seen that $\pa_xg_{n+1}$ has exactly 2 zeros in $n_1\alpha+I_{n+1,1}$ (it still preserves the cosine-type characteristics). Moreover, we have the following important observation:
\beq\label{range-of-distance1}
\pi-l_{n}^{-1+3\eta}\leq \max_{x\in I_{n+1,1}}\{g_{n+1}(x,t)\}-\min_{x\in I_{n+1,1}}\{g_{n+1}(x,t)\} \leq \pi-l_{n}^{-1-3\eta}.
\eeq
(\ref{range-of-distance1}) implies that $g_{n+1}$ may has at most two zeros in $I_{n+1,1}$ denoted by $\{c_{n+1,1},c'_{n+1,2}\}.$ The above property also holds for $I_{n+1,2}.$ Furthermore, we have
$$\left|c_{n+1,1} + n_1\alpha - c'_{n+1,2}\right|, \quad \left|c_{n+1,2} - n_1\alpha - c'_{n+1,1}\right| < \lambda^{-\frac{1}{30} r_{n}}$$  for any further $n.$ This means that $c'_{n+1,j}$ and $c_{n+1,j}$ almost lie in the same orbit of the base dynamics. This also means that, essentially, we still have only two critical points $c_{n+1,1}$ and $c_{n+1,2}.$

\begin{figure}

\begin{tikzpicture}[yscale=0.7]
\draw [->] (-1.5,0) -- (1.5,0);
\draw [very thick,domain=-0:1.3] plot (\x, {1/180*pi*atan(50*tan((\x -0.6) r))+0.5*pi-0.5*(\x)-0.2});
\draw [very thick,domain=-1.3:0] plot (\x, {0.03-0.5*(\x)-0.2});
\draw [semithick,domain=-1.5:1.5] plot (\x, {0*\x+pi});
\draw [dashed] (-1.1,0) -- (-1.1,pi);
\draw [dashed] (1.1,0) -- (1.1,pi);
\node [right] at (1.5,0) {$x$};
\node [left] at (-1.5,0) {$0$};
\node [left] at (-1.5,3.14) {$\pi$};
\end{tikzpicture}
\begin{tikzpicture}[yscale=0.7]
\draw [->] (-1.5,0) -- (1.5,0);
\draw [very thick,domain=-0:1.3] plot (\x, {1/180*pi*atan(50*tan((\x -0.6) r))+0.5*pi-0.5*(\x)+0.1});
\draw [very thick,domain=-1.3:0] plot (\x, {0.03-0.5*(\x)+0.1});
\draw [semithick,domain=-1.5:1.5] plot (\x, {0*\x+pi+0.});
\draw [dashed] (-1.1,0) -- (-1.1,pi);
\draw [dashed] (1.1,0) -- (1.1,pi);
\node [right] at (1.5,0) {$x$};
\node [left] at (-1.5,0) {$0$};
\node [left] at (-1.5,3.14) {$\pi$};
\end{tikzpicture}
\begin{tikzpicture}[yscale=0.7]
\draw [->] (-1.5,0) -- (1.5,0);
\draw [very thick,domain=-0:1.3] plot (\x, {1/180*pi*atan(50*tan((\x -0.6) r))+0.5*pi-0.5*(\x)+0.3});
\draw [very thick,domain=-1.3:0] plot (\x, {0.03-0.5*(\x)+0.3});
\draw [semithick,domain=-1.5:1.5] plot (\x, {0*\x+pi+0.});
\draw [dashed] (-1.1,0) -- (-1.1,pi);
\draw [dashed] (1.1,0) -- (1.1,pi);
\node [right] at (1.5,0) {$x$};
\node [left] at (-1.5,0) {$0$};
\node [left] at (-1.5,3.14) {$\pi$};
\end{tikzpicture}
\begin{tikzpicture}[yscale=0.7]
\draw [->] (-1.5,0) -- (1.5,0);
\draw [very thick,domain=-0:1.3] plot (\x, {1/180*pi*atan(50*tan((\x -0.6) r))+0.5*pi-0.5*(\x)+0.45});
\draw [very thick,domain=-1.3:0] plot (\x, {0.03-0.5*(\x)+0.45});
\draw [semithick,domain=-1.5:1.5] plot (\x, {0*\x+pi+0.});
\draw [dashed] (-1.1,0) -- (-1.1,pi);
\draw [dashed] (1.1,0) -- (1.1,pi);
\node [right] at (1.5,0) {$x$};
\node [left] at (-1.5,0) {$0$};
\node [left] at (-1.5,3.14) {$\pi$};
\end{tikzpicture}
\begin{tikzpicture}[yscale=0.7]
\draw [->] (-1.5,0) -- (1.5,0);
\draw [very thick,domain=-0:1.5] plot (\x, {1/180*pi*atan(50*tan((\x -0.25) r))+0.5*pi-0.7*(\x+0.25)+0.75});
\draw [very thick,domain=-1.3:0] plot (\x, {0.03-0.5*(\x+0.25)+0.75});
\draw [semithick,domain=-1.5:1.5] plot (\x, {0*\x+pi+0.});
\draw [dashed] (-1.1,0) -- (-1.1,pi);
\draw [dashed] (1.1,0) -- (1.1,pi);
\node [right] at (1.5,0) {$x$};
\node [left] at (-1.5,0) {$0$};
\node [left] at (-1.5,3.14) {$\pi$};
\end{tikzpicture}

\caption{The Case {\bf 2}: $g_{n+1}\  {\rm\ (mod\ \pi)}$ may have two zeroes (the first and fifth pictures), one zero (the second  and fourth pictures) or no zeroes (the third picture, corresponds to Case {\bf 3}) in $x\in I_{n+1,1}$  if $d_{n}$ changes its sign as $t$ changes. }
\end{figure}
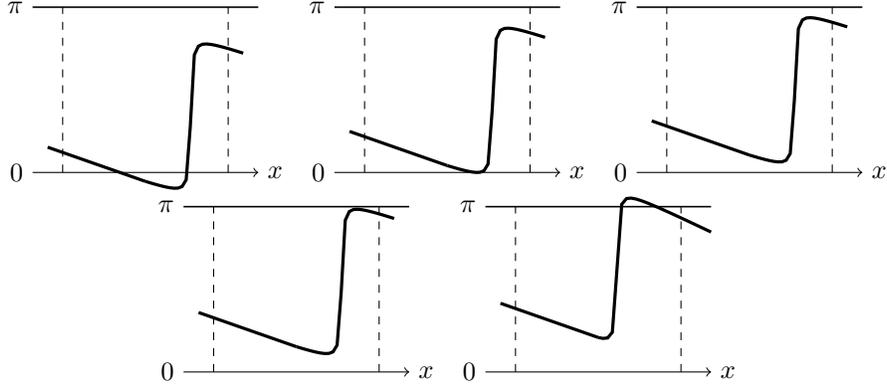
\begin{figure}
\begin{tikzpicture}[yscale=0.7]
\draw [->] (-1.4,0) -- (1.8,0);
\draw [very thick,domain=-1.3:1.5] plot (\x, {1/180*pi*atan(30*tan(\x r))+0.5*pi});
\draw [very thick,domain=1.5:1.7] plot (\x, {0.001*\x+pi});
\draw [semithick,domain=-1.4:1.8] plot (\x, {0*\x+pi});
\draw [dashed] (-0.8,0) -- (-0.8,pi);
\draw [dashed] (0.8,0) -- (0.8,pi);
\node [below] at (0,0) {$c^{*}_n\in n_1\alpha+I_{n+1,1}$ corresponds to Case {\bf 2}};
\node [right] at (1.8,0) {$x$};
\node [left] at (-1.4,0) {$0$};
\node [left] at (-1.4,3.14) {$\pi$};
\end{tikzpicture}
\begin{tikzpicture}[yscale=0.7]
\draw [->] (-1.4,0) -- (1.8,0);
\draw [very thick,domain=-1.3:1.5] plot (\x, {1/180*pi*atan(30*tan(\x r))+0.5*pi});
\draw [very thick,domain=1.5:1.7] plot (\x, {0.001*\x+pi});
\draw [semithick,domain=-1.4:1.8] plot (\x, {0*\x+pi});
\draw [dashed] (1.3,0) -- (1.3,pi);
\draw [dashed] (-0,0) -- (-0,pi);
\node [below] at (0.6,0) {$c^{*}_n {\rm\ is\ near\ the \ edge\ of\ }n_1\alpha+I_{n+1,1}$};
\node [right] at (1.8,0) {$x$};
\node [left] at (-1.4,0) {$0$};
\node [left] at (-1.4,3.14) {$\pi$};
\end{tikzpicture}
\begin{tikzpicture}[yscale=0.7]
\draw [->] (-1.4,0) -- (1.8,0);
\draw [very thick,domain=-1.3:1.5] plot (\x, {1/180*pi*atan(30*tan(\x r))+0.5*pi});
\draw [very thick,domain=1.5:1.7] plot (\x, {0.001*\x+pi});
\draw [semithick,domain=-1.4:1.8] plot (\x, {0*\x+pi});
\draw [dashed] (0.5,0) -- (0.5,pi);
\draw [dashed] (1.7,0) -- (1.7,pi);
\node [below] at (1.2,0) {$c^{*}_n\not\in n_1\alpha+I_{n+1,1}$ corresponds to Case {\bf 1}};
\node [right] at (1.8,0) {$x$};
\node [left] at (-1.4,0) {$0$};
\node [left] at (-1.4,3.14) {$\pi$};
\end{tikzpicture}
\caption{The graph of $g_{n+1}-g_{n}=\tilde \phi(x,t)$, the restriction of the complete graph of approximate arc tangent function in $n_1\alpha+I_{n+1,1}$, depends on the position of $c^{*}_n$, where $c^{*}_n$ is the middle point of the jump part.
  The Case that $g_{n+1}-g_{n} {\rm \ mod}\ \pi$ is  small corresponds to $c^{*}_n\not\in n_1\alpha+I_{n+1,1}$; and the Case that $g_{n+1}-g_{n}$ is a pulse function corresponds to $c^{*}_n\in n_1\alpha+I_{n+1,1}$.}
\end{figure}
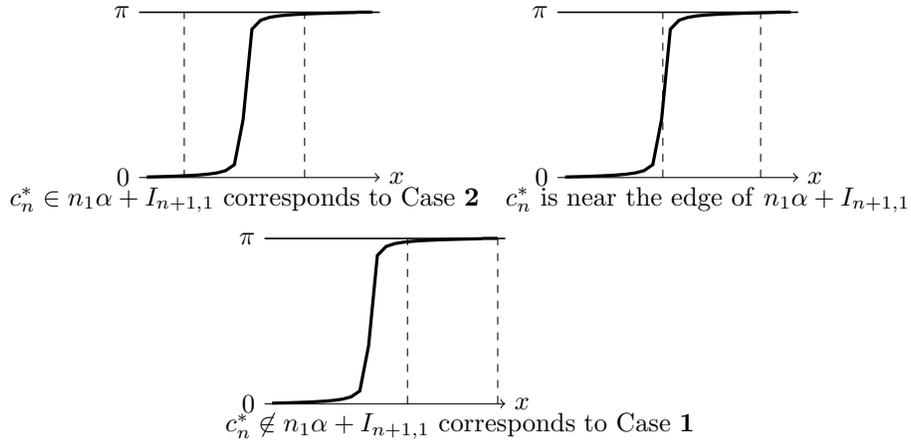

\begin{figure}\label{fig6'}
\begin{tikzpicture}[yscale=0.7]
\draw [->] (-6,0) -- (6,0);
\node [right] at (6,0) {$x$};
\node [left] at (-6,0) {$0$};
\node [left] at (-6,3.14) {$\pi$};
\node [left] at (-6,-3.14) {$-\pi$};
\node [left] at (-6,6.28) {$2\pi$};
\draw [semithick,domain=-6:6] plot (\x, {0*\x+2*pi+0});
\draw [semithick,domain=-6:6] plot (\x, {0*\x+pi+0});
\draw [semithick,domain=-6:6] plot (\x, {0*\x-pi+0});


\draw [very thick, orange, domain=0:1.5] plot ({\x - 0.3-3}, {1/180*pi*atan(50*tan((\x -0.25) r))+0.5*pi-0.7*(\x+0.25)+0.75 + 0.3});
\draw [very thick, orange, domain=-1.3:0] plot ({\x - 0.3-3}, {0.03-0.5*(\x+0.25)+0.75 + 0.3});
\draw [very thick, dashed, red, domain=-0.5:1.5] plot ({\x - 1-2.6}, {1/180*pi*atan(50*tan((\x -0.25) r))+0.5*pi});
\draw [very thick, dashed, yellow, domain=-2:0.4] plot ({\x-1.2}, {0.03-0.5*(\x+0.25)+0.75 + 2.3-pi});
\draw [very thick, dashed, yellow, domain=-0.5:1.5] plot ({-\x + 1+2.6}, {1/180*pi*atan(50*tan((\x -0.25) r))+0.5*pi});
\draw [very thick, dashed, red, domain=-2:0.4] plot ({-\x+1.2}, {0.03-0.5*(\x+0.25)+0.75 + 2.3-pi});
\draw [very thick, orange, domain=0:1.5] plot ({-\x + 0.3+3}, {1/180*pi*atan(50*tan((\x -0.25) r))+0.5*pi-0.7*(\x+0.25)+0.75 + 0.3});
\draw [very thick, orange, domain=-1.3:0] plot ({-\x + 0.3+3}, {0.03-0.5*(\x+0.25)+0.75 + 0.3});
\node [above] at (-3,-pi) {$I_{n,1}\bigcup T^{-k}I_{n,2}$};
\node [above] at (3,-pi) {$I_{n,2}\bigcup T^{k}I_{n,1}$};

\node [above] at (-2.6,2*pi-1.7) {$I_{n+1,1}$};
\node [above] at (2.6,2*pi-1.7) {$I_{n+1,2}$};

\fill (2.06,pi)circle(3pt);\fill (-2.06,pi)circle(3pt);
\fill (3,pi)circle(3pt);
\fill (2.92,pi+0.44)circle(3pt);
\fill (-2.92,pi+0.44)circle(3pt);
\fill (3.35,0)circle(3pt);\node [below] at (3.35,0) {$c_{n,1}+k\alpha$};
\fill (-3.35,0)circle(3pt);\node [below] at (-3.35,0) {$c_{n,2}-k\alpha$};
\fill (1.6,0)circle(3pt);\node [below] at (1.6,0) {$c_{n,2}$};
\fill (-1.6,0)circle(3pt);\node [below] at (-1.6,0) {$c_{n,1}$};

\fill (-3.2,pi-2.17)circle(3pt);\node [below] at (-3.2,pi-2.17) {$\tilde{c}'_{n+1,1}$};
\fill (+3.2,pi-2.17)circle(3pt);\node [below] at (+3.2,pi-2.17) {$\tilde{c}'_{n+1,2}$};
\fill (-3,pi)circle(3pt);
\node [above] at (2.92,pi+0.44) {$\tilde{c}_{n+1,2}$};\node [above] at (-2.92,pi+0.44) {$\tilde{c}_{n+1,1}$};
\node [below] at (1.85,pi) {$c_{n+1,2}$};
\node [below] at (3,pi) {$c'_{n+1,1}$};
\node [below] at (-1.85,pi) {$c_{n+1,1}$};
\node [below] at (-3.2,pi) {$c'_{n+1,2}$};

\draw [dashed] (2.06,pi) -- (2.06,0);
\draw [dashed] (3,0) -- (3,pi);
\draw [dashed] (-2.06,0) -- (-2.06,pi);
\draw [dashed] (-3,pi) -- (-3,0);
\draw [dashed] (-2.92,pi+0.44) -- (-5.9,pi+0.44);
\draw [dashed] (-3.2,pi-2.17) -- (-5.9,pi-2.17);
\draw[<-] [dashed,very thick] (-5.7,pi+0.44) -- (-5.7,2.3);
\draw[->] [dashed,very thick] (-5.7,1.9) -- (-5.7,pi-2.17);
\node [below] at (-5.5,2.6) {$\approx \pi-\lambda^{-c|k|}(<\pi)$};

\draw[<->] [dashed,very thick, domain=-2.06:3] plot (\x, 1);

\draw[<->] [dashed,very thick, domain=-3:2.06] plot (\x, 2.2);
\node [below] at (0,2) {$\approx k\alpha (\text{mod}~1)$};

\draw[black,dashed] (-1.3-3,-pi) rectangle (2.3-3,2*pi);
\draw[black,dashed] (1.3+3,-pi) rectangle (-2.3+3,2*pi);

\draw[black,dashed] (-3.4,0.6) rectangle (-1.3,2*pi-0.9);
\draw[black,dashed] (3.4,0.6) rectangle (1.3,2*pi-0.9);
\draw[black,dashed] (-3.2,2.5) rectangle (-1.5,2*pi-2.5);
\draw[black,dashed] (3.2,2.5) rectangle (1.5,2*pi-2.5);
\draw[->] [dashed,very thick] (-1.5,2*pi-2.5) -- (-1.5+0.6,2*pi-2.5+0.6);\node [right] at (-1.5+0.6,2*pi-2.5+0.6) {$\tilde{I}_{n,1}$};
\draw[->] [dashed,very thick] (1.5,2*pi-2.5) -- (1.5-0.6,2*pi-2.5+0.6);\node [left] at (1.5-0.6,2*pi-2.5+0.6) {$\tilde{I}_{n,2}$};
\end{tikzpicture}
\caption{In the domain $I_{n,1}\bigcup T^{-k}I_{n,2},$ red dashed line is $\arctan[\|A_k(x,t)\|^2 \tan(g_{n,2}(x+k\alpha,t))]+\frac\pi 2$, the yellow dashed line is $g_{n,1}(x,t)$;~In the domain $I_{n,2}\bigcup T^{k}I_{n,1},$ yellow dashed line is $\arctan[\|A_k(x,t)\|^2 \tan(g_{n,1}(x-k\alpha,t))]+\frac\pi 2$ and the red dashed line is $g_{n,2}(x,t).$ The orange part is yellow part $+$ red part, which represents $g_{n+1}.$ }
\end{figure}
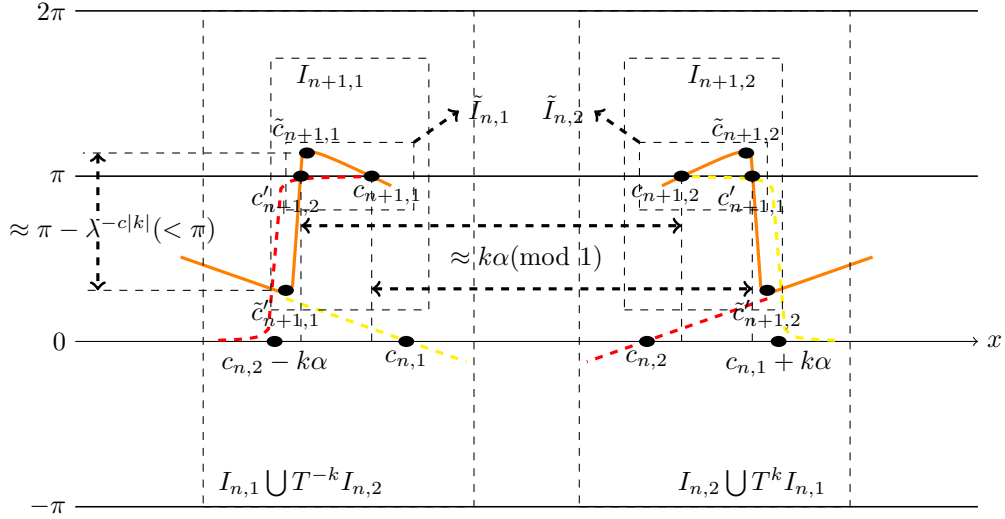

For the non-resonance case, property (ii) of $\tilde{\phi}$ shows that $g_{n+1}$ possesses similar properties as $g_n$, see the third picture of Figure 1.

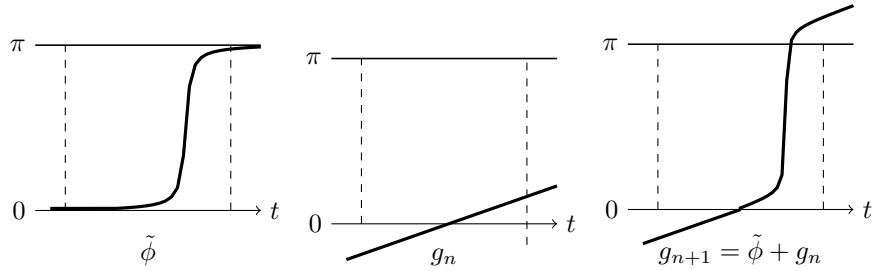
\begin{figure}
\begin{tikzpicture}[yscale=0.7]
\draw [->] (-1.5,0) -- (1.5,0);
\draw [very thick,domain=-0.4:1.5] plot (\x, {1/180*pi*atan(20*tan((\x -0.5) r))+0.5*pi});
\draw [very thick,domain=-1.3:-0.4] plot (\x, 0.04+0.001*\x);
\draw [semithick,domain=-1.5:1.5] plot (\x, {0*\x+pi});
\draw [dashed] (-1.1,0) -- (-1.1,pi);
\draw [dashed] (1.1,0) -- (1.1,pi);
\node [below] at (0,-0.3) {$ \tilde{\phi}$};
\node [right] at (1.5,0) {$t$};
\node [left] at (-1.5,0) {$0$};
\node [left] at (-1.5,3.14) {$\pi$};
\end{tikzpicture}
\begin{tikzpicture}[yscale=0.7]
\draw [->] (-1.5,0) -- (1.5,0);
\draw [very thick,domain=-1.3:1.5] plot (\x, {-0.03+0.5*(\x)});
\draw [semithick,domain=-1.5:1.5] plot (\x, {0*\x+pi});
\draw [dashed] (-1.1,0) -- (-1.1,pi);
\draw [dashed] (1.1,-0.4) -- (1.1,pi);
\node [below] at (0,-0.3) {$g_{n}$};
\node [right] at (1.5,0) {$t$};
\node [left] at (-1.5,0) {$0$};
\node [left] at (-1.5,3.14) {$\pi$};
\end{tikzpicture}
\begin{tikzpicture}[yscale=0.7]
\draw [->] (-1.5,0) -- (1.5,0);
\draw [very thick,domain=-0.02:1.5] plot (\x, {pi+1/180*pi*atan(50*tan((\x -0.6) r))-0.5*pi+0.5*(\x)});
\draw [very thick,domain=-1.3:0] plot (\x, {-0.0+0.5*(\x)});
\draw [semithick,domain=-1.5:1.5] plot (\x, {0*\x+pi});
\draw [dashed] (-1.1,0) -- (-1.1,pi);
\draw [dashed] (1.1,0) -- (1.1,pi);
\node [below] at (0,-0.3) {$g_{n+1}= \tilde{\phi}+g_{n}$};
\node [right] at (1.5,0) {$t$};
\node [left] at (-1.5,0) {$0$};
\node [left] at (-1.5,3.14) {$\pi$};
\end{tikzpicture}
\caption{Graphs of $g_{n+1}$ as a function of $t$ in $I_{n+1}$: }
\end{figure}

Now we go back to (\ref{product31}). Assume there is no $l<s$ with $0<|k|=|t_l|< 10^{-3}\varepsilon q_{n+1}$ such that  $T^{k}x\in I_{n+1}$. Applying the argument above for non-resonant case on $A_{ t_2-t_1}(T^{ t_1}x,t)\cdot A_{t_{1}}(x,t)$, one shows that  Case {\bf 1} occurs for
$A_{t_2}(x,t)=A_{ t_2-t_1}(T^{ t_1}x,t)\cdot A_{t_{1}}(x,t)$. Moreover, from (\ref{norm-non-resonance}), we obtain (\ref{norm-lower-bound}) and (\ref{norm-derivative}) for $A_{t_2}(x,t)$.
     Inductively, we can obtain similar estimates on $A_{t_j}(x,t)$ for any $j\le s.$ Repeat this process forward and we obtain Case {\bf 1} for $A_{t_s}(x,t)$.  Otherwise, if  such an $l\in \mathbb Z$ does exist, it is unique by the properties of continued fraction approximation and the Diophantine condition on $T$.
Then, we rewrite (\ref{product31}) as
$$A_{r^+_{n+1}}(x,t)=A_{r^+_{n+1}-k}(T^kx,t)\cdot A_{k}(x,t),\quad x\in I_{n+1},\quad k=t_l,$$
where $A_{r^+_{n+1}-k}(T^kx,t)$ and $A_{k}(x,t)$ can be estimated by the above discussion for non-resonance Case. Recall $r^+_{n+1}\thicksim q_{n+1}$,
we have $n_2=r^+_{n+1}-k\gg n_1=k$.
Then the estimate on $g_{n+1}$ is reduced to the resonant Case in Figure 2 (the first picture) and we obtain Case {\bf 2} for the $(n+1)$-th step. Moreover, by (\ref{norm-resonance}) we have
$$
\|A_{r^+_{n+1}}(x,t)\|\ge \|A_{r^+_{n+1}-k}(T^kx,t)\|/\|A_{k}(x,t)\|\gtrsim \lambda^{r^+_{n+1}-2k}\gtrsim \lambda^{(1-\epsilon_n)r^+_{n+1}}\gg \lambda^{r^+_{n}}
$$
with $\epsilon_n\rightarrow 0$ as $n\rightarrow \infty$, which is just (\ref{norm-lower-bound}). It is not difficult to obtain
(\ref{norm-derivative}).

If it is Case {\bf 2} for the $n$-th step, then critical points of $g_n$ on $I_{n+1}$ are close to each other such that $$\inf\limits_{j\in \Z}|(c_{n+1, 1}-c_{n+1,2}+k\alpha)-j|\ll |I_{n+1}|.$$ Thus only non-resonant Case occurs by the Diophantine condition in $(n+1)$-th steps. The situation keeps unchanged until some $m$ large enough such that
$$\inf\limits_{j\in \Z}|(c_{n+m, 1}-c_{n+m,2}+k\alpha)-j|\gtrsim |I_{n+m}|.$$ At that time, we reobtain a lower bound for $g'_{n+m}$ similar as in (\ref{I-first-derivative}) and thus we go back Case {\bf 1} again.

The mechanism for the absence of zeros for $g_{n+1}$, or equivalently, the appearance of spectral gaps in Case {\bf 3}, is shown in Figures 2 and 4.
Note that $g_n\thickapprox \arctan(t-v(x))$, which implies that the sign of $\pa_t g_n$ is unchanged (equals to $1$) everywhere (we have seen that due to the cosine-type condition, the sign of $\pa_x g_n$ cannot keep unchanged). Hence $\pa_t\tilde{\phi}$ possesses the same sign as that of $\pa_t g_n$, which leads that $\min|\pa_t g_{n+1}|\ge \min |\pa_t g_n|$. Thus by induction we obtain the first inequality in  (\ref{lm17-main}) (See Figure 5). The estimate for the second one is similar as the one with respect to $x$.

\subsection{The properties $c_{n,j}$ with respect to $t$}\label{class}
Corollary \ref{lmg'} is a direct result of \eqref{lm17-main} and \eqref{I-first-derivative} in Theorem \ref{theorem12}.

\begin{corollary} \label{lmg'} For any fixed $t_0\in [\inf v-\frac{2}{\lambda},\sup v+\frac{2}{\lambda}]$ and $n^*,~i\in\Z_+,$ if each step $i\geq n^*$ belongs to Case \textbf{1}, then for$\ any\ ~t\in B(t_0,\lambda^{-q^{\frac{1}{800}}_{N+i-1}}),$ the following inequality holds:
\begin{equation}\label{lmg'-2} c_0(\alpha,v)q^{-20}_{N+{n^*}-2}<\left\vert\frac{d(c_{i,1}(t)-c_{i,2}(t))}{dt}\right\vert< C(\alpha,v)\lambda^{20q_{N+{n^*}-2}}.\end{equation}

\end{corollary}
\begin{definition}\label{def5} For $X\in \Z,$ let $s(X)$ as $$q_{N+s(X)-2}< |X|\leq q_{N+s(X)-1}.$$ \end{definition}
\begin{corollary}\label{gwx1} If there exist $k\in \Z$ and $n>s(k)+1$~\text{such\ that}~$n-$th~\text{step belongs to Case} \textbf{3} with $T^kI_{n,1}\cap I_{n,2}\neq \emptyset$,  \text{then for all}~$s(k)\leq j\leq n-1,$
\begin{equation}\label{case32}~j-th~\text{step belongs to Case}~\textbf{2}.\end{equation}
Moreover, if $\inf\limits_{j\in \Z}\vert c_{s(k),1}(t)+k\alpha-c_{s(k),2}(t)-j\vert\leq C\lambda^{-r_{s(k)-1}^{\frac{1}{56}}},$ then
\begin{equation}\label{partialx}\tilde{c}_{s(k)+1,j}\neq \tilde{c}'_{s(k)+1,j},~j=1,2.\end{equation}
\end{corollary}
\begin{proof} Suppose, for the sake of contradiction,  that for some $s(k)+1\leq j\leq n-1,$ $j-th$ step belongs to Case \textbf{1}, then we have $T^kI_{j,1}\bigcap I_{j,2}=\emptyset.$ This implies
$$\inf\limits_{l\in \Z}\vert c_{j,1}+k\alpha-c_{j,2}-l\vert>|I_{j,1}|=q^{-2000\tau}_{N+j-1}.$$
Then \eqref{ga1} and \eqref{ga2} guarantee that
$$\inf\limits_{l\in \Z}\vert c_{n,1}+k\alpha-c_{n,2}-l\vert>\frac{1}{2}q^{-2000\tau}_{N+j-1}.$$

On the other hand, in Case \textbf{3}, \eqref{cn1cn2} implies
$$\inf\limits_{l\in \Z}\vert c_{n,1}+k\alpha-c_{n,2}-l\vert\leq \lambda^{-\frac{1}{55}r_{n-1}}\ll \frac{1}{2}q^{-2000\tau}_{N+n-1}<\frac{1}{2}q^{-2000\tau}_{N+j-1},$$  which leads to a contradiction. Hence \eqref{case32} holds true.

By \eqref{cc}, \eqref{either12} and \eqref{tic}, we have
$$
|\tilde{c}_{s(k),j}-c_{s(k),j}|+|\tilde{c}'_{s(k),j}-c_{s(k),j}|\leq C\lambda^{-r_{s(k)-1}^{\frac{1}{56}}}+C\lambda^{-\frac{1}{30}r_{s(k)-1}}+C\lambda^{-\frac{1}{2}|k|}\ll \frac{1}{2}|I_{s(k),1}|,
$$
which implies \eqref{partialx} by (2) of Case {\bf 2} in Theorem \ref{theorem12}.
\end{proof}

\subsection{The definition of the new labels of gaps}

The following theorem is the core part of this paper. From the induction theorem, we can label all the spectral gaps  by the number $k$ describing the distance of two critical points when a strong resonance occurs between them by (\ref{cn1cn2}). Moreover, it provides an accurate estimate for both the size of each spectral gap and the distance between any two spectral gaps.

\begin{definition} $~dist(A, B)\ {\rm denotes\ the \ distance\ of\ two\ sets}\ A \ and\ B.$
\end{definition}
Let $\Sigma^{\lambda}:=\lambda^{-1}\Sigma_{\lambda v, \alpha}.$ Note that $\mathbb{R} \backslash\Sigma^{\lambda}$ is open, it composes of a series of open intervals, that is, the spectral gaps.
\begin{theorem}\label{15}
For each $\lambda \geq \lambda_0$, there exists a map $\Phi$ between
$K(\lambda)$ and the set of gaps such that $\Phi(K(\lambda))=\mathbb{R}\backslash\Sigma^{\lambda} = \bigcup_{k \in K(\lambda)} G_k^{\lambda}$ and for each $k\in K(\lambda)$ the following properties are satisfied for each $ n \geq s(k)$:

\begin{itemize}
    \item[]{\rm (1)} For $G^{\lambda}_k = (t_k^-(\lambda), t_k^+(\lambda))$ and $X \in \{+, -\}$, we have
    \begin{equation}\label{var-ep}
    \liminf\limits_{j\in \Z}\left| c_{n,1}(t^{X}_k(\lambda)) + k\alpha - c_{n,2}(t^{X}_k(\lambda)) -j\right| \leq \lambda^{-r^{\frac{1}{100}}_{n-1}}.
    \end{equation}
    Furthermore, there exist sequences $\{t^{k}_{X,n}(\lambda)\}_{n \ge l_k} \subset \mathbb{R}$ such that $g_{n,1}(t^{k}_{X,n})$ is tangent to $y=0\ {\rm(}\!\mod \pi{\rm)}$ and
    \begin{equation}\label{tangent}
    | t^{k}_{X,n} - t_{k}^{X} | \leq \lambda^{-r^{\frac{1}{100}}_{n-1}}.
    \end{equation}

    \item[]{\rm (2)} It holds that
    \begin{equation}\label{pugapguj}
    c\lambda^{-10000|k|} \leq |G_{k}^{\lambda}| \leq C\lambda^{-\frac{1}{10000}|k|}.
    \end{equation}

\item[]{\rm (3)} $~k\in K(\lambda)$ if and only if~\text{there~exists~some~t}\text{~such~that}~$$\inf\limits_{j\in \Z}|c_{s(k),1}(t)+k\alpha-c_{s(k),2}(t)-j|\leq C\lambda^{-r^{\frac{1}{100}}_{s(k)-1}},\ \ \min\limits_{x\in I_{s(k)}}|g_{s(k)+1}(x,t)(\rm{mod}~\pi)|\geq c\lambda^{-|k|^3}.$$
    Particularly,
\begin{equation}\label{gap-mid}\text{if}~\quad~T^kI_{s(k),1}(t,\lambda)= I_{s(k),2}(t,\lambda),~\text{then}~k\in K(\lambda).\end{equation}
\item[]{\rm (4)} For any fixed $k\in \Z,$ there exists $\lambda^*(k)>\lambda_0$ such that for any $\lambda\geq \lambda^*,$ $k\in K(\lambda)$ and
\begin{equation}\label{xzengd}{\rm Leb}\{x\vert v(x-\alpha)<t^{X}_k(\lambda)\}=-k\alpha(\text{\rm mod}~1)+O(\lambda^{-c})~\text{for}~X\in \{+,-\}.\end{equation}
\end{itemize}
\end{theorem}
\vskip 0.2cm
Thus we have the following definition of the (new) labels of spectral gaps,
\begin{definition}\noindent Each spectral gap $G$ of $\lambda v$ can be identified by a unique $k=k(G)\in \mathbb{Z}$ such that \eqref{var-ep}--\eqref{pugapguj} holds true, which is called the new label  of the gap (different from the one obtained from GLT).
\end{definition}

From the Induction Theorem, we know that Case \textbf{3} is the only mechanism that creates an ``open'' gap.

\section{Proof of Theorems \ref{Th2}, \ref{Th4}, and \ref{Th5}}

\subsection{Large Coupling Asymptotics of the IDS}

In this subsection, we provide a basic lemma for estimating the IDS under a fairly mild setting. Let \( v \in C^0(\R/\Z) \) be such that for any \( a \in \mathbb{R} \) and \( \epsilon>0\),
\begin{equation}\label{C1}
\text{Leb}(\{x \mid |v(x) - a| < \epsilon \}) < \epsilon^c
\end{equation}
for some \( c(v) > 0 \). This condition are satisfied by all 1-periodic real analytic functions and functions \( v \in C^k(\R/\Z) \) satisfying the following non-degeneracy condition:
$$
\inf_{x \in \R/\Z} \sum_{m=1}^k |v^{(m)}(x)| \geq c>0.
$$

Let \( H^L_{\alpha, \lambda v, x} \) denote the restriction of \( H_{\alpha, \lambda v, x} \) to the interval \([-L, L]\), and let \(\{E_j(x)\}_{j=-L}^{L}\) be the eigenvalues of \( H^L_{\alpha, \lambda v, x} \). Recall that the IDS of \( H_{\alpha, \lambda v, x} \) is defined by
\begin{equation}\label{vL}
N(E, \lambda v) = \lim_{L \to \infty} \frac{1}{2L+1} \#\{j : E_j(x) < E, |j| \leq L\},
\end{equation}
where the limit exists for almost every \( x \in \R/\Z \) and is independent of \( x \).

The estimate of the IDS can be reduced to the properties of potential functions.

\begin{theorem}\label{keyl}
Let \( v \in C^0(\R/\Z) \) satisfy \eqref{C1}, \( E \in \mathbb{R} \), and \( \alpha \in \mathbb{R} \setminus \mathbb{Q} \). Then for \( \lambda_0(v,\alpha) > 0 \) and \( c(v, \alpha) > 0 \) given in the Induction Theorem, for each \( \lambda > \lambda_0 \) it holds that
\begin{equation}\label{tarobtain1}
N(E, \lambda v) = \text{Leb}\{x \mid \lambda v(x) < E\} + O(\lambda^{-c}).
\end{equation}
\end{theorem}

\begin{proof}
Define a diagonal matrix as
$
H_{\lambda v}^L = \text{diag}\{\lambda v(x + L \alpha), \ldots, \lambda v(x), \ldots, \lambda v(x - L \alpha)\}.
$
Note that both \( H_{\alpha, \lambda v, x}^L \) and \( H_{\lambda v}^L \) are symmetric matrices. It was proved by Hoffman and Wielandt \cite{hw} that there exists a permutation of \(\{-L, \ldots, 0, \ldots, L\}\) such that
\begin{equation}\label{pematrix}
\sum_{j=-L}^{L} |E_{\pi(j)} - \lambda v(x + j \alpha)|^2 \leq 4L.
\end{equation}
Since \( \alpha \in \mathbb{R} \setminus \mathbb{Q} \), by  Birkhoff ergodic theorem, we have
\begin{equation}\label{finalL}
\lim_{L \to \infty} \frac{1}{2L+1} \#\{j : \lambda v(x + j \alpha) < E, |j| \leq L\} = \text{Leb}\{x \mid \lambda v(x) < E\}
\end{equation}
for almost every \( x \in \R/\Z \).

Fix \( x \in \R/\Z \) such that \eqref{vL} and \eqref{finalL} hold. Then there exists \( L_0(\lambda v, x) \) such that if \( L > L_0 \), we have
\begin{equation}\label{weneed3}
\left|N(E, \lambda v) - \frac{1}{2L+1} \#\{j : E_j(x) < E, |j| \leq L\}\right| \leq \lambda^{-c}
\end{equation}
and
\begin{equation}\label{weneed4}
\left|\text{Leb}\{x \mid \lambda v(x) < E\} - \frac{1}{2L+1} \#\{j : \lambda v(x + j \alpha) < E, |j| \leq L\}\right| \leq \lambda^{-c}.
\end{equation}
Define
$$
B_L = \{j : |E_{\pi(j)}(x) - \lambda v(x + j \alpha)| > \lambda^{1/2}\},
$$
then by \eqref{pematrix}, we have
\begin{equation}\label{weneed1}
\# B_L \leq 4 \lambda^{-1/2} L.
\end{equation}
On the other hand, for \( j \in [-L, L] \setminus B_L \), we have
\begin{equation}\label{weneed2}
\{j : \lambda v(x + j \alpha) \leq E - \lambda^{1/2}\} \subset \{j : E_{\pi(j)}(x) < E\} \subset \{j : \lambda v(x + j \alpha) \leq E + \lambda^{1/2}\}.
\end{equation}

Denote \( D_1 = \{j : E_j(x) < E, |j| \leq L\} \) and \( D_2 = \{j : \lambda v(x + j \alpha) < E, |j| \leq L\} \). Let \( D_{i1} = D_i \cap B_L \) and \( D_{i2} = D_i \setminus B_L \) for \( i = 1, 2 \). By \eqref{weneed1}, \eqref{weneed2}, \eqref{weneed4}, and \eqref{C1}, we have
\begin{equation}\label{weneed5}
\begin{array}{ll}
 &|\#D_1 - \#D_2| \le |\#D_{11} - \#D_{21}| + |\#D_{12} - \#D_{22}| \\
 \leq & \#B_L + \#\{j : E - \lambda^{1/2} \leq \lambda v(x + j \alpha) \leq E + \lambda^{1/2}, |j| \leq L\} \\
 \leq & 4 \lambda^{-1/2} L + |\#\{j : \lambda v(x + j \alpha) \leq E + \lambda^{1/2}, |j| \leq L\} \\ &- \#\{j : \lambda v(x + j \alpha) < E - \lambda^{1/2}, |j| \leq L\}| \\
 \leq & 4 \lambda^{-1/2} L + (2L + 1) \left(\text{Leb}\{|\lambda v(x) - E| \leq \lambda^{1/2}\} + 2 \lambda^{-c}\right) \\
 \leq & 4 \lambda^{-1/2} L + (2L + 1) \left(\text{Leb}\left\{\left|v(x) - \frac{E}{\lambda}\right| \leq \lambda^{-1/2}\right\} + 2 \lambda^{-c}\right) \\
 \leq & 8 (2L + 1) \lambda^{-c} \quad \text{if } c < \frac{1}{2}.
\end{array}
\end{equation}
Finally, \eqref{weneed3}, \eqref{weneed4}, and \eqref{weneed5} imply \eqref{tarobtain1}.
\end{proof}

\subsection{Continuity of the Spectral Gaps}
First we list some classic results on \textbf{uniformly hyperbolic} systems ($\mathcal{UH}$ for short) without proof, one can see \cite{z1} for details.
\begin{proposition}[\cite{johnson}]\label{johnson} For irrational $\alpha,$ it holds that
$$\Sigma^{\lambda}=\{t \vert (\alpha,A^{\lambda(t-v)})\notin \mathcal{UH}\}.$$
\end{proposition}

\begin{proposition}[\cite{yoc}]\label{yocc} $(\alpha,A^{\lambda(t-v)})\in \mathcal{UH}$ if and only if there exists $c>0$ and $\rho>1$ such that
$$\|A_{n}(x,t)\|\geq c\rho^{|n|}$$ for all $n\in \Z$ and for all $x\in \R/\Z.$
\end{proposition}

\begin{proposition}[\cite{z1}, Lemma 11]\label{uh} Let $\{B^{(k)}\}_{k\in \Z}\subset SL(2;\R)$ be a bounded sequence,  $\beta=\inf\limits_{k\in \Z}\|B^{(k)}\|$ and $\gamma=\inf\limits_{k\in \Z}\left\vert \tan[s(B^{(k)})-u(B^{(k-1)})] \right\vert$. Assume
$$\beta\gg\frac{1}{\gamma}\gg 1>\frac{2}{\beta-\beta^{-1}}.$$ Then for each $k\in \Z$ and each $n\geq 1$, it holds that
$$\|B^{(k+n-1)}\cdots B^{(k)}\|\geq (c\beta\gamma)^n.$$
\end{proposition}

By the help of above proposition, the following holds true.
\begin{Lemma}\label{lem20} For $t\in [-\frac{2}{\lambda}+\inf v,\frac{2}{\lambda}+\sup v],$
if $t\notin \Sigma^{\lambda},$ then $~\liminf\limits_{j\rightarrow +\infty}\min\limits_{x\in I_{j}}\left\vert g_{j}(x,t) (\rm{mod}~\pi)\right\vert>0.$

\end{Lemma}

\begin{proof}

It is directly obtained from \cite[Lemma 2]{z2}.

\end{proof}

\begin{Lemma}\label{conti-spec}
For \( \lambda_1 > \lambda_0 \) and \( k \in K(\lambda_1) \), there exists $\delta_k>0$ such that for any \( \lambda \in (\lambda_1-\delta_k,\lambda_1+\delta_k) \), it holds that \( k \in K(\lambda) \), and \( t_k^{\pm}(\lambda) \) is continuous on \((\lambda_1-\delta_k,\lambda_1+\delta_k)\).
\end{Lemma}

\begin{proof}
By Theorem \ref{theorem12}, for any fixed \( t \), \(n\in \Z\) and \( \lambda > \lambda_0 \), \( c_n(t, \lambda) \) must satisfy one of the following two conditions:
\begin{enumerate}
    \item[a:] \( g_n(c_{n}(t, \lambda), t, \lambda) = 0 \) with \( \partial_x g_n(c_{n}(t, \lambda), t, \lambda) \neq 0 \).
    \item[b:] \( \partial_x g_n(c_{n}(t, \lambda), t, \lambda) = 0 \) with \( \left|\frac{\partial^2 g_n(c_{n}(t, \lambda), t, \lambda)}{\partial x^2}\right| \neq 0 \).
\end{enumerate}
Since \( g_n(x, t, \lambda) \in C^2(I_n \times [\lambda_0, +\infty)) \), it follows from the Implicit Function Theorem that \( c_n(t,\lambda)\) is continuously differentiable with respect to \( (t,\lambda) \).

Notice that for \( k \in K(\lambda_1) \), by (3) of Lemma \ref{15}, there exist some \( t_1 \) and $s(k)$ for which
\[
\begin{cases}
    \inf\limits_{j\in \Z}|c_{s(k),1}(t_1, \lambda_1) + k \alpha - c_{s(k),2}(t_1, \lambda_1)-j| < C\lambda_1^{-r^{\frac{1}{100}}_{s(k)-1}}, \\
    \min\limits_{x} |g_{s(k)+1}(x,t_1, \lambda_1)(\rm{mod}~\pi)| > c \lambda_1^{-|k|^3}.
\end{cases}
\]
Since both \( c_{s(k)}(t_1,\lambda) \) and \( g_{s(k)+1}(x, t_1, \lambda) \) are continuous in $\lambda$, thus there exists \( \delta_k(\lambda_1) > 0 \) such that for any \( \lambda \in (\lambda_1 - \delta_k, \lambda_1 + \delta_k) \), the following conditions hold:
\[
\begin{cases}
    \inf\limits_{j\in \Z}|c_{s(k),1}(t_1, \lambda) + k \alpha - c_{s(k),2}(t_1, \lambda)-j| < 2C\lambda^{-r^{\frac{1}{100}}_{s(k)-1}}\ll |I_{s(k),1}(t_1)|, \\
    \min\limits_{x} |g_{s(k)+1}(x,t_1, \lambda)(\rm{mod}~\pi)| >\frac{1}{2}c \lambda^{-|k|^3},
\end{cases}
\]
which immediately implies that for $t_1$ and any \( \lambda \in (\lambda_1 - \delta_k, \lambda_1 + \delta_k) \), step $n$ belongs to \( \text{Case } {\bf 3} \) with $T^kI_{n,1}\bigcap I_{n,2}\neq \emptyset$. Thus,  \eqref{Case3lem} implies $t_1\notin \Sigma^{\lambda}$ for any \( \lambda \in (\lambda_1 - \delta_k, \lambda_1 + \delta_k) \). Therefore, \( k \in K(\lambda) \) for any \( \lambda \in (\lambda_1 - \delta_k, \lambda_1 + \delta_k) \).

On the other hand, by (1) of Theorem \ref{15}, there exist two sequence $\{{t}^{k}_{\pm,n}(\lambda)\}_{n\ge 1}\subset \R$ such that ${\rm \ for\ any\ } n\geq l(k)$, $g_{n,1}({t}^{k}_{\pm,n}(\lambda))$ is tangent to $y=0 \ (mod\  \pi)$ and
\begin{equation}\label{tangent111}\vert {t}^{k}_{\pm,n}(\lambda)-{t}_{k}^{\pm}(\lambda) \vert \leq \lambda^{-r^{\frac{1}{100}}_{n-1}}.
\end{equation}

Finally, one notes that
$$
g_{n,1}(c_{n}({t}^{k}_{\pm,n}, \lambda), {t}^{k}_{\pm,n}, \lambda)=\frac{\partial g_{n,1}(x,t,\lambda)}{\partial x}(c_{n}({t}^{k}_{\pm,n}, \lambda), {t}^{k}_{\pm,n}, \lambda)= 0.
$$
By \eqref{lm17-main} and the Implicit Function Theorem, we obtain the continuity of \( t^k_{\pm, n}(\lambda) \) with respect to \( \lambda \). Finally, from \eqref{tangent111}, we conclude the continuity of \( t_k^{\pm}(\lambda) \) on \((\lambda_1-\delta_k,\lambda_1+\delta_k)\).
\end{proof}

\subsection{Proof of Theorem \ref{Th2}}

\

\textbf{Proof of $K(\lambda_0)\subset K(\lambda)$:}
Suppose, for the sake of contradiction, that there exists some \( k \in K(\lambda_0) \) such that
\begin{equation}\label{LAM1}
\Lambda(k) := \{\lambda \geq \lambda_0 \mid k \notin K(\lambda)\}\neq \emptyset.
\end{equation}
Then we can define
$$
\lambda_* := \inf \{\lambda\geq \lambda_0 \mid \lambda\in \Lambda(k)\}~({\rm The~infimum~of~\lambda(\geq \lambda_0)~such~that~|G^{\lambda}_k|=0}).
$$
Note the continuity of $|G_k^{\lambda}|,$ which follows from Lemma \ref{conti-spec}, implies \begin{equation}\label{contra} |G_k^{\lambda_*}| = 0 .\end{equation}


By the definition of $\lambda_*$, for any \( \epsilon > 0 \), we have $k\in K(\lambda_*-\epsilon)$, thus by \eqref{pugapguj} of Theorem 8,
\beq\label{gd1}
|(\lambda_*-\epsilon)^{-1} G_k^{\lambda_* - \epsilon}| > c (\lambda_*-\epsilon)^{-10000|k|} > 0.
\eeq
By Lemma \ref{conti-spec}, there exists $\delta_k>0$ such that \( |\lambda^{-1} G_k^{\lambda}| \) is continuous on \[ [\lambda_* - \frac{1}{2}\delta_k, \lambda_*]\subset [\lambda_*- \frac{1}{2}\delta_k, \lambda_* + \frac{1}{2}\delta_k] \subset (\lambda_*-\epsilon-\delta_k, \lambda_*-\epsilon+\delta_k)\](~by taking \(0<\epsilon\leq \frac{\delta_k}{10000}\)). It together with \eqref{gd1} implies that
$$
|\lambda_*^{-1} G_k^{\lambda_*}| = \lim_{\lambda \to \lambda_*^-} |\lambda^{-1} G_k^{\lambda}| \geq \frac{1}{2}c \lambda_*^{-10000|k|} > 0.
$$
This contradicts \eqref{contra}. Hence, \eqref{LAM1} is false and we have $k\in \bigcap\limits_{\lambda\geq \lambda_0}K(\lambda).$ This implies \begin{equation}\label{K111}K(\lambda_0)\subset K(\lambda),~\text{for any}~\lambda\geq\lambda_0.\end{equation}

\

\textbf{Proof of $K(\lambda)\subset K(\lambda_0)$:}
Suppose,  for the sake of contradiction, there exists some $\lambda'\geq \lambda_0$ such that \begin{equation}\label{contra1}K(\lambda')\backslash K(\lambda_0)\neq \emptyset.\end{equation} Then there exists $k\in K(\lambda')\backslash K(\lambda_0)$ such that $|G_k^{\lambda_0}|=0$  and \( \left\vert \lambda'^{-1} G_k^{\lambda'} \right\vert > 0 \) with $\lambda'>\lambda_0$. Then $
\lambda^*(k) := \sup \{\lambda\geq \lambda_0 \mid \left\vert G_k^{\lambda} \right\vert = 0 \text{ and } \lambda < \lambda'\}
$ is well-defined.


And Lemma \ref{conti-spec} implies \begin{equation}\label{LAM3} |G_k^{\lambda^*}|=0 \end{equation}
and \( \left\vert \lambda^{-1} G_k^{\lambda} \right\vert > 0 \) on \( [\lambda' - \delta, \lambda'](\subset [\lambda_0+\delta, \lambda'] )\) for some suitable $\delta>0.$ Therefore \( \lambda_0 \leq \lambda^*(k) < \lambda' - \delta \). Note that the definition of $\lambda^*$ implies that for any sufficiently small \( \epsilon > 0 \), \( |(\lambda^*+\epsilon)^{-1} G_k^{\lambda^* + \epsilon}| > 0 \). As in the previous argument, this contradicts \eqref{LAM3}. Hence, \eqref{contra1} is invalid. Therefore
\begin{equation}\label{K222}K(\lambda)\subset K(\lambda_0),~\text{for any}~\lambda\geq\lambda_0.\end{equation}

\

\eqref{K111} and \eqref{K222} complete the proof of $K(\lambda)=K(\lambda_0)$.

\

\textbf{Proof of $K(\lambda_0)=\Z$:} This directly follows from (4) of Theorem \ref{15} and the  fact $K(\lambda)=K(\lambda_0)$ for any $\lambda\geq \lambda_0.$
\qed

\subsection{Proof of Theorem \ref{Th4}}

We need the following two conclusions.
\begin{Lemma}\label{CE-cont}
For each \( k \in \mathbb{Z} \), the functions \( t^\pm_k(\lambda) \) are continuous on \( [\lambda_0, +\infty) \).
\end{Lemma}
\begin{proof} This is a direct corollary of Lemma \ref{conti-spec} and Theorem \ref{Th2}. \end{proof}
Johnson and Moser proved the following result in \cite{JM}:

\begin{proposition}[\cite{JM}]
\label{rhocont}
For any \( \alpha \in \mathbb{R} \setminus \mathbb{Q} \) and \( v \in C^0(\mathbb{T}) \), the function \( N(E, \lambda v) \), as a function of \( (E, \lambda v) \), is jointly continuous in \( \mathbb{R} \times (0, +\infty) \).
\end{proposition}

\begin{proof}[Proof of Theorem \ref{Th4}:]
By Proposition \ref{rhocont} and Lemma \ref{CE-cont}, the set
\[
Q_{\lambda_0, \lambda'} := \left\{ N\left( \lambda t^{-}_{k}(\lambda), \lambda v \right) \mid \lambda_0 \leq \lambda \leq \lambda' \right\}
\]
is a closed interval for any \( \lambda' \geq \lambda_0 \). On the other hand, GLT implies
\[
Q_{\lambda_0, \lambda'} \subset \left\{ m\alpha \, (\mathrm{mod} \, 1) \mid m \in \mathbb{Z} \right\}.
\]
Then by the connectedness of a closed interval, we conclude the existence of  some~ $k' \in \mathbb{Z}$ such that
\[
Q_{\lambda_0, \lambda'} \equiv k'\alpha ({\rm mod~1}), \quad \forall~\lambda' > \lambda_0.
\]
Therefore, for any \( \lambda_1, \lambda_2 \geq \lambda_0 \) and \(k\in \Z\) it holds that
\[
1-2\rho(G_k^{\lambda_1}) = N\left( \lambda_1 t^{-}_{k}(\lambda_1), \lambda_1 v \right) = N\left( \lambda_2 t^{-}_{k}(\lambda_2), \lambda_2 v \right) = 1-2\rho(G_k^{\lambda_2}).
\]
This completes the proof of Theorem \ref{Th4}.
\end{proof}

\subsection{Proof of Theorem \ref{Th5}}
Combining Theorem \ref{Th4}, Theorem \ref{keyl} and \eqref{xzengd} of Theorem \ref{15}, we have
\[
\lim\limits_{\lambda \to +\infty} N(\lambda t_k^+(\lambda), \lambda v) =\lim\limits_{\lambda\rightarrow +\infty}\mathrm{Leb}\{x\in \R/\Z \mid v(x) < t_k(\lambda)\}=1-k\alpha \ (\mathrm{mod} \ 1).
\]

Then $
   \lim_{\lambda \to \infty} \rho(G_k^\lambda) = \frac{1-\lim\limits_{\lambda \to +\infty} N(\lambda t_k^+(\lambda), \lambda v)}{2}= \frac{k\alpha(\rm{mod}~1)}{2}.
   $ $\square$


\section{\textbf{Proof of Theorem \ref{idsac}}}

In this section, we will demonstrate that the IDS for $C^2$-$\cos$ type quasiperiodic Schr\"odinger operators is absolutely continuous.

\subsection{\textbf{Representation Theorem of IDS}} By \eqref{var-ep} and Theorem \ref{Th5}, we obtain
\begin{equation}\label{duand}
c_{\infty,2}(t_k^{X}) - c_{\infty,1}(t_k^{X}) = k\alpha=-[1 - 2\rho(\lambda t_k^{X})]=-N(\lambda t_k^+(\lambda), \lambda v)
\end{equation}
for any \( k \in \mathbb{Z} \) and \( X \in \{+,-\} \). This implies that the function \( c_{\infty,1}(\frac{E}{\lambda}) - c_{\infty,2}(\frac{E}{\lambda}) \) equals to \( N(E,\lambda v) \) at the endpoints of the gaps \footnote{Here, we define a special gap labeled by \( 0 \) and denoted by \( G^{\lambda}_0 \),  is \( (-\infty, \inf \Sigma^{\lambda}) \cup (\sup \Sigma, +\infty) \). For this gap, \( t^0_- = \inf \Sigma^{\lambda} \) and \( t^0_+ = \sup \Sigma^{\lambda} \).}.

By \eqref{ga1}, we define (recall that $E=\lambda t$)
\beq\label{gaa1}
\mathcal{C}(E) :=\lim\limits_{n\rightarrow\infty}(c_{n,1}(t) - c_{n,2}(t)){\rm\ (mod}\ 1{\rm )}:=c_{\infty,1}(t) - c_{\infty,2}(t){\rm\ (mod}\ 1{\rm )},\ \ \forall t\in \Sigma^\lambda.
\eeq

We now state the following Representation Theorem, which connects the limit-critical point with the IDS.

\begin{theorem}\label{represent}
It holds that
$\mathcal{C}(E)\in C^0(\lambda \Sigma^{\lambda})$ and
$\mathcal{C}(E) = N(E)~{\rm for~any}~E\in\lambda\Sigma^{\lambda}.$
\end{theorem}
\begin{proof}

For the first conclusion, it is enough to show $\mathcal{C}(t)\in C^0(\Sigma^{\lambda}).$ Thus the continuity directly follows from the continuity of $c_{n,j}$ with respect to $t$ in its well-defined domain and that the convergence in \eqref{gaa1} is uniform in $t$.

For the second one, one notes that $N,\mathcal{C}\in C^{0}(\lambda\Sigma^\lambda)$. By \eqref{duand}, we have
\[ N(E)=1-2\rho(E)=\mathcal{C}(E)
\]
on $\bigcup_{k\in \Z}\{\lambda t_k^{+},\lambda t_k^-\}$.

Note that it was proved in \cite{wz2} that $\bigcup\limits_{k\in \Z}\{t_k^{+},t_k^-\}$ is dense in $\Sigma^{\lambda}.$ Therefore, \( N(E) =\mathcal{C}(E) \) on $\overline{\bigcup_{k\in \Z}\{\lambda t_k^{+},\lambda t_k^-\}}=\lambda \Sigma^{\lambda}.$
\end{proof}

\noindent \textbf{Proof of Theorem \ref{idsac}:} For any $\gamma>0$ and $\tau>1$, we define the following sets
\[
\quad \mathcal{G}_{\gamma,\tau} := \{ E \in \lambda\Sigma^{\lambda} \mid \inf_{j\in\Z}\vert \mathcal{C}(E) - k\alpha-j\vert \geq \frac{\gamma}{|k|^{\tau}},~\forall~k\in~\Z-\{0\}\},
\]
\[
\Sigma_2 := \lambda\Sigma^{\lambda} -\cup_{\gamma>0}\mathcal{G}_{\gamma,\tau}.
\]

\begin{Lemma}\label{ga3}
\( N(\cdot) \) is Lipschitz on each \(\mathcal{G}_{\gamma,\tau}\). More precisely, there is $C(\gamma)>0$ such that
$$
|N(E)-N(E')|\leq C(\gamma)|E-E'|, \ \  \forall E,E'\in \mathcal{G}_{\gamma,\tau}.
$$
\end{Lemma}

\begin{proof}

For some fixed $E\in \mathcal{G}_{\gamma,\tau},$ by the definition there exists some $k^*(\gamma)\in \Z_+$ such that
\begin{equation}\label{ce-k}\inf_{j\in\Z}\vert \mathcal{C}(E) + k\alpha-j\vert \geq \frac{1}{|k|^{\frac{3}{2}\tau}},~\forall~|k|\geq k^*.\end{equation}

By \eqref{ga1} and \eqref{gaa1},  we have \begin{equation}\label{ce-k1}\inf\limits_{j\in \Z}|\mathcal{C}(E)-(c_{n,1}(E)-c_{n,2}(E))-j|\leq C\lambda^{-\frac{1}{2}r_{n-1}},~\forall~n\in \Z_+.\end{equation}

By \eqref{ce-k} and \eqref{ce-k1}, for $k^*<|k|\leq q_{N+n-1}$, we have
$$\inf\limits_{j\in \Z}|c_{n,1}(t)-c_{n,2}(t)+ k\alpha-j|\geq \frac{1}{|q_{N+n-1}|^{\frac{3}{2}\tau}}-C\lambda^{-\frac{1}{2}r_{n-1}}>|I_n|,~\forall~n\in \Z_+,$$
which implies for $k^*<|k|\leq q_{N+n-1}$,
\begin{equation}\label{uuu}T^kI_{n,1}\bigcap I_{n,2}=\emptyset.
\end{equation}

We claim that there exists some $n^*(\gamma)$ such that each step $n>n^*(\gamma)$ belongs to Case \textbf{1}. In fact, if there exist infinity many $n_j$ and $0<|k_j|<q_{N+n_j-1}$ such that step $n_j$ belongs to Case \textbf{2} with
$T^{k_j}I_{n_j-1,1}\bigcap I_{n_j-1,2}\neq \emptyset,$  then it will contradict \eqref{uuu}.

Then by \eqref{lmg'-2}, for any $n>n^*(\gamma)$,
\[
\left\vert \frac{d (c_{n,1} - c_{n,2})}{d E} \right\vert \leq C\lambda^{q_{N+n^*-2}}:=C^*(\gamma)
\]
holds uniformly for \( E' \in B(E, \lambda^{-q^{\frac{1}{800}}_{N+n-2}}) \) and \( n \geq n^* \). Thus by \eqref{ga1} for any $E'\in\mathcal{G}_{\gamma,\tau}\cap  B(E, \lambda^{-q^{\frac{1}{800}}_{N+n-2}})$, we have
\[
\left\vert \mathcal{C}(E) - \mathcal{C}(E') \right\vert \leq \left\vert c_n(E) - c_n(E') \right\vert + \lambda^{-\frac{1}{2}r_{n-1}} \leq C^*\left\vert E - E' \right\vert + \lambda^{-\frac{1}{2}q_{N+n-2}},
\]
which implies
\[
\left\vert \mathcal{C}(E) - \mathcal{C}(E') \right\vert \leq 2C^*(\gamma)\left\vert E - E' \right\vert
\]
for large $n\gg n^*$ and any \(E'\in\mathcal{G}_{\gamma,\tau}\) satisfying \( \lambda^{-q^{\frac{1}{800}}_{N+n-1}} \leq \left\vert E - E' \right\vert \leq C \lambda^{-q^{{\frac{1}{800}}}_{N+n-2}}\). Thus there exists a  small neighbor of $E$ denoted by $B_{E}$ such that
\begin{equation}\label{bebe}
\left\vert \mathcal{C}(E) - \mathcal{C}(E') \right\vert \leq 2C^*(\gamma)\left\vert E - E' \right\vert,~\forall~E'\in~B_{E}\bigcap \mathcal{G}_{\gamma,\tau}.
\end{equation}

Since $\bigcup\limits_{E\in \mathcal{G}_{\gamma,\tau}}\left(B_{E}\bigcap \mathcal{G}_{\gamma,\tau}\right)=\mathcal{G}_{\gamma,\tau},$ \eqref{bebe} implies what we desire.
\end{proof}

\begin{Lemma}\label{ga4}
We have that
$\text{Leb} \{ N(\Sigma_2) \} = 0.
$
\end{Lemma}

\begin{proof}
It follows from Theorem \ref{represent} and a simple Borel-Catelli argument.
\end{proof}

Finally for a zero Lebesgue measure set $e$, let
$$
e_n=e\cap\mathcal{G}_{\frac{1}{n},\tau}, \ \ e_c=e\cap\Sigma_2.
$$
By Lemma \ref{ga3} and Kirszbtann Theorem \cite{Kir},  one can find some Lipschitz function $C_n:\R\rightarrow\R$  with
$$
C_n\big|_{\mathcal{G}_{\frac{1}{n},\tau}}=\mathcal{C}\big|_{\mathcal{G}_{\frac{1}{n},\tau}}.
$$
Thus, $\text{Leb} \{N(e_n)\}=\text{Leb} \{\mathcal{C}(e_n)\}=\text{Leb} \{C_n(e_n)\}=0$. By Lemma \ref{ga4}, we have
$$
\text{Leb} \{N(e)\}\leq \sum_{n}\text{Leb} \{N(e_n)\}+\text{Leb} \{N(e^c)\}=0,
$$
which implies the absolute continuity of the IDS.

\section{Proof of Theorem \ref{homogeneous}}

First we need the following result.
\begin{theorem}For $k, k' \in K(\lambda)$ with $k \neq k'$, it holds that
    \begin{equation}\label{dgkij}
    \text{dist}(G^{\lambda}_{k}, G^{\lambda}_{k'}) \geq c|k-k'|^{-C}.
    \end{equation}
\end{theorem}

\begin{proof} For $k,k'\in K(\lambda),$ assume that $t_k^+<t_{k'}^-.$ By \eqref{duand}, $N(t_k^+)=k\alpha$ and $N(t_{k'}^-)=k'\alpha.$ By \cite{LWY}, $N(E)$ is H\"older continuous. Thus
$$\frac{c_2}{|k-k'|^{c_1}}<|k\alpha -k'\alpha|=|N(t_k^+)-N(t_{k'}^-)|\leq C|t_k^+-t_{k'}^-|^{\alpha}\leq C(\text{dist}(G^{\lambda}_{k}, G^{\lambda}_{k'}))^{\alpha}.$$
Hence $\text{dist}(G^{\lambda}_{k}, G^{\lambda}_{k'})\geq c|k-k'|^{-C}.$
\end{proof}
Next, we proceed to prove the homogeneity of the spectrum. Let
\begin{equation}\label{dyghhh}
0 < \epsilon \leq \epsilon_0(\lambda_0, v, \alpha) \ll \min\{c, C^{-1}\},
\end{equation}
with \( C > 1 > c > 0 \) defined in Theorem \ref{15}, and  $E$ be a spectral point and
\[
k_{\epsilon} := \min \{|k| \in \mathbb{Z}_+ : \Sigma_{\alpha,\lambda v} \cap (E-\epsilon, E+\epsilon) \neq \emptyset\}.
\]

\textbf{Case I:} If \( k_{\epsilon} \) satisfies
$
2C\lambda^{-\frac{1}{2}ck_{\epsilon}} \leq \frac{1}{10} \epsilon,
$
then by the definition of \( k_{\epsilon} \), our assumption and the upper bound of \( |G^{\lambda}_{k}| \) from Theorem \ref{15}, we have
\begin{equation}\label{tjy}
\left\vert (E-\epsilon, E+\epsilon) \cap (\mathbb{R} \setminus \Sigma_{\alpha,\lambda v}) \right\vert \leq \sum_{j = k_{\epsilon}}^{+\infty} |\lambda G^{\lambda}_{j}| \leq \sum_{j = k_{\epsilon}}^{+\infty} C \lambda^{-c|j|} \leq 2C\lambda^{-\frac{1}{2}ck_{\epsilon}} \leq \frac{1}{10} \epsilon.
\end{equation}
Clearly, \eqref{tjy} implies that
\begin{equation}\label{tjydjl}
\left\vert (E-\epsilon, E+\epsilon) \cap \Sigma_{\alpha,\lambda v} \right\vert \geq 2\epsilon - \frac{1}{10} \epsilon > \frac{1}{100} \epsilon.
\end{equation}

\textbf{Case II:} If \( k_{\epsilon} \) satisfies
\[
2C\lambda^{-\frac{1}{2}ck_{\epsilon}} \geq \frac{1}{10} \epsilon,
\]
then for large \( \lambda \geq \lambda_0(\alpha, v) \) and small \( \epsilon \leq \epsilon_0(\lambda_0, c, C) \) defined in \eqref{dyghhh}, we have
$$
k_{\epsilon} \leq 4c^{-1} \frac{\ln \epsilon^{-1}}{\ln \lambda}.
$$
Note that \( E \notin \lambda G^{\lambda}_{k_{\epsilon}} \). Without loss of generality, assume that
\begin{equation}\label{wolg0}
(E, E+\epsilon) \cap \lambda G^{\lambda}_{k_{\epsilon}} \neq \emptyset,
\end{equation}
which implies
\begin{equation}\label{wolg}
\lambda t^-_{k_{\epsilon}}(\lambda) > E,
\end{equation}
where \( \lambda t_{k_{\epsilon}}^-(\lambda) \) is the left endpoint of \( \lambda G^{\lambda}_{k_{\epsilon}}\).

We denote \( G = (E-\epsilon, \lambda t_{k_{\epsilon}}^-(\lambda)) \). By assumptions \eqref{wolg0} and \eqref{wolg}, we have
\[
E-\epsilon \leq E < \lambda t_{k_{\epsilon}}^- < E+\epsilon,
\]
hence
\[
\epsilon \leq |E - (E-\epsilon)| \leq |\lambda t_{k_{\epsilon}}^- - (E-\epsilon)| = |\lambda G| \leq |E+\epsilon - (E-\epsilon)| \leq 2\epsilon.
\]
We define
\[
k^*_{\epsilon} = \min \{|k| > k_{\epsilon} : G^{\lambda}_{k} \cap G \neq \emptyset\}.
\]
By \eqref{dgkij},
we have
\begin{equation}\label{day}
k^*_{\epsilon} \geq \epsilon^{-c}.
\end{equation}

By \eqref{day} and the upper bound of \( |G^{\lambda}_{k}| \) from Theorem \ref{15}, we have
\[
\left\vert G \cap (\mathbb{R} \setminus \Sigma_{\alpha,\lambda v}) \right\vert \leq C\lambda^{-\frac{1}{10000}\epsilon^{-c}} < C e^{-\frac{1}{10000}\epsilon^{-c}} \leq \frac{\epsilon}{10}.
\]
Therefore, from the fact that \( (E, E+\epsilon) \cup G = (E-\epsilon, E+\epsilon) \), we have
$$
\left\vert (E-\epsilon, E+\epsilon) \cap (\mathbb{R} \setminus \Sigma_{\alpha,\lambda v}) \right\vert < |(E, E+\epsilon)| + \left\vert G \cap (\mathbb{R} \setminus \Sigma_{\alpha,\lambda v}) \right\vert < \epsilon + \frac{\epsilon}{10}.
$$
Then,
\begin{equation}\label{jl2}
\left\vert (E-\epsilon, E+\epsilon) \cap \Sigma_{\alpha,\lambda v} \right\vert \geq 2\epsilon - \epsilon - \frac{\epsilon}{10} > \frac{\epsilon}{100}.
\end{equation}
Hence, by \eqref{tjydjl} and \eqref{jl2}, for any \( 0 < \epsilon \leq \min\{{\rm diam}~\Sigma_{\alpha,\lambda v},\  \epsilon_0\} \) (${\rm diam}~(S):=\sup\limits_{p,q\in S}|p-q|$), we have
\[
\left\vert (E-\epsilon, E+\epsilon) \cap \Sigma_{\alpha,\lambda v} \right\vert \geq \min\left\{\frac{\epsilon_0}{200 \cdot \mathrm{diam}(\Sigma_{\alpha,\lambda v})}, \frac{1}{100}\right\} \cdot \epsilon,
\]
which shows that \( \Sigma_{\alpha,\lambda v} \) is
\[
\min\left\{\frac{\epsilon_0}{200 \cdot \mathrm{diam}(\Sigma_{\alpha,\lambda v})}, \frac{1}{100}\right\}\text{-homogeneous.}
\]

\section{Proof of Theorem \ref{15}}
This section is devoted to the proof of Theorem \ref{15}.
\begin{Lemma}\label{gapkehua} $t\notin \Sigma^{\lambda}$ if and only if there exists some step which belongs to Case \textbf{3}.
\end{Lemma}
\begin{proof}
The sufficiency follows from \eqref{Case3lem}. For the necessity, if $t\notin \Sigma^{\lambda}$ and the iteration never comes to Case \textbf{3}, then from the definition of Case \textbf{3}, we have
 $
 \min\limits_{x\in I_n}|g_{n+1}(x,t)(\rm{mod}~\pi)|\leq \lambda^{-r_n^{\frac{1}{50}}}.
 $
 It follows that
 $
 \liminf\limits_{n\rightarrow +\infty}\min\limits_{x\in I_n}|g_{n+1}(x,t)(\rm{mod}~\pi)|=0.
 $
 By Lemma \ref{lem20}, $t\in \Sigma^{\lambda}.$ This leads to a contradiction.
\end{proof}

Let $G^*=(a,b)$ be a spectral gap and $t\in (a,b)$, by Lemma \ref{gapkehua}, there is $n\in\N$ such that in the induction  Case $\textbf{3}_n$ occurs.

\begin{Lemma}\label{lemma9} Let $t\in G^*$ be as above, then the following holds true.
\begin{enumerate}
\item[i:] There is $k\in \Z$ such that $n\geq s(k)$ with $T^kI_{n,1}(t)\bigcap I_{n,2}(t)\neq \emptyset$ and  \begin{equation}\label{csk}\inf\limits_{j\in \Z}\vert c_{s(k),1}(t)+k\alpha-c_{s(k),2}(t)-j\vert\leq C\lambda^{-r^{\frac{1}{100}}_{s(k)-1}}.\end{equation}
\item[ii:]
We have that
\begin{equation}\label{gapest}c\lambda^{-10000|k|}\leq |a-b|\leq C\lambda^{-\frac{1}{10000}|k|}.\end{equation}
Moreover for $Q\in \{a,b\},$ we have
\begin{equation}\label{def-K}\inf\limits_{j\in \Z}|c_{s(k)+p,1}(Q)+k\alpha-c_{s(k)+p,2}(Q)-j|\leq C\lambda^{-r^{\frac{1}{100}}_{s(k)+p-1}}, \ p\geq 0.\end{equation}
Furthermore there exist two sequences denoted by $\{a_p\}$ and $\{b_p\}$ such that for all~$p\geq 0$ it holds that $$|a_p-a|+|b_p-b|\leq C\lambda^{-r^{\frac{1}{100}}_{s(k)+p-1}}$$ and $g_{s(k)+p+1,1}(Q,x)$ is tangent to $y = 0~(\text{\rm mod}~\pi).$
\end{enumerate}
\end{Lemma}
\begin{proof}
By \eqref{cn1cn2}, at step $n$, there is $k\in\Z$ with $0<|k|<q_{N+n-1}$ and $T^kI_{n,1}(t)\bigcap I_{n,2}(t)\neq \emptyset$, such that \begin{equation}\label{i11}\inf\limits_{j\in \Z}\vert c_{n,1}(t)+k\alpha-c_{n,2}(t)-j\vert\leq C\lambda^{-r_{n-1}^{\frac{1}{55}}}.\end{equation}
By the definition of $s(k)$, we have
$
q_{N+s(k)-2}\leq |k|\leq q_{N+s(k)-1}.
$
Hence $n\geq s(k).$

Note that by \eqref{ga1} and \eqref{ga2}, we always have
$$
|c_{j,1}(t)-c_{j+1,1}(t)|\leq C\lambda^{-\frac{1}{2}r_{j-1}},~j\leq n-1.
$$
Hence the invalidity of \eqref{csk} will lead to
$$
\inf\limits_{j\in \Z}\vert c_{n,1}(t)+k\alpha-c_{n,2}(t)-j\vert>\lambda^{-r_{s(k)-1}^{\frac{1}{100}}}-C\lambda^{-\frac{1}{2}r_{s(k)-1}}>\frac{1}{2}\lambda^{-r_{s(k)-1}^{\frac{1}{100}}}\gg C\lambda^{-r_{n-1}^{\frac{1}{55}}},
$$ which contradicts \eqref{i11}. Thus we completes the proof of (i).

We turn to the proof of (ii). \eqref{i11} implies that Case $\textbf{2}_{s(k)}$ occurs and $$\inf\limits_{j\in \Z}\vert c_{s(k),1}(t)+k\alpha-c_{s(k),2}(t)-j\vert\leq C\lambda^{-r_{n-1}^{\frac{1}{55}}}+\sum\limits_{l=s(k)+1}^{n-1}\sum\limits_{j=1}^2|c_{l,j}-c_{l-1,j}|<C\lambda^{-r_{s(k)-1}^{\frac{1}{56}}}.$$ Then, by \eqref{partialx} of Corollary \ref{gwx1}, we have
$$
\{\tilde{c}_{s(k)+1,1}(t),\tilde{c}'_{s(k)+1,1}(t)\}=\{x\in I_{s(k)+1,1} \vert \partial_x g_{s(k)+1,1}(x,t)=0\},
$$
where $\tilde{c}_{s(k)+1,1}(t)$ and $\tilde{c}'_{s(k)+1,1}(t)$ correspond to the maximum and minimum value of $g_{s(k)+1,1}(\cdot,t)$, respectively. It follows by \eqref{range-g-n} and \eqref{range-g-n-lower-bound} that
\beq\label{gwx2}
\pi - C\lambda^{-\frac{1}{100}|k|} \le g_{s(k)+1}(\tilde{c}_{s(k)+1,1}(t),t) - g_{s(k)+1}(\tilde{c}'_{s(k)+1,1}(t),t)\le \pi - c\lambda^{-100|k|}.
\eeq
Recall that \eqref{lm17-main1} shows
\begin{equation}\label{lm17-main1'}
\frac{1}{10} <  \frac{\partial g_{s(k)+1}(x,\bar{t})}{\partial t}  < C\lambda^{5|k|}, \ \ \forall \bar{t}\in Q_{s(k)}(t)\footnote{One can refer to Figure 4 to understand how \( g_{s(k)+1}(x,\bar{t}) \) moves on $Q_{s(k)}(t)$ with respect to \( \bar{t}\).},\ \ x\in I_{s(k)}(\bar{t}).
\end{equation}

For $1\leq \xi\leq 2,$ we define
$$\mathcal{H}_{\pi; \xi}:=\left\{\bar{t}\in Q_{s(k)}(t) \vert~\vert g_{s(k)+1}(c_{s(k)+1,1}(\bar{t}),\bar{t})-\pi\vert<\lambda^{-r^{\frac{1}{50\xi}}_{s(k)}}\right\},$$
$$\mathcal{H}_{0; \xi}:=\left\{\bar{t}\in Q_{s(k)}(t) \vert~\vert g_{s(k)+1}(c_{s(k)+1,1}(\bar{t}),\bar{t})\vert<\lambda^{-r^{\frac{1}{50\xi}}_{s(k)}}\right\}.$$
Note that by  \eqref{gwx2} and \eqref{ga110}, for $\bar{t}\in \mathcal{H}_{\pi;\xi},$ we have
$$
g_{s(k)+1}(\tilde{c}'_{s(k)+1,1}(\bar{t}),\bar{t})>c\lambda^{-100|k|}-\lambda^{-r^{\frac{1}{50\xi}}_{s(k)}}>c\lambda^{-200|k|}.
$$
Thus
$$\mathcal{H}_{\pi;\xi}\subset Q_{s(k)}(t)\backslash\mathcal{H}_{0;\xi},~\xi\in [1,2].
$$
Moreover, by the monotonicity of $g_{s(k)+1}(x,\bar{t})$ with respect to $\bar{t}$ on $Q_{s(k)}(t)$ (see \eqref{lm17-main1'}), we have
$
\inf\mathcal{H}_{\pi;\xi}>\sup\mathcal{H}_{0;\xi},~\xi\in [1,2].
$
This implies $ [\sup\mathcal{H}_{0;\xi},\inf\mathcal{H}_{\pi;\xi}]\neq \emptyset.$

On the other hand, for $\bar{t}\in [\sup\mathcal{H}_{0;\xi},\inf\mathcal{H}_{\pi;\xi}],$ by the definition of $\mathcal{H}_{\pi;\xi},$ $\mathcal{H}_{0;\xi}$ and the fact that $\xi\geq 1$, we have
$$
\min\limits_{x\in I_{s(k)}(\bar{t})}|g_{s(k)+1,1}(x,\bar{t})(\rm{mod}~\pi)|\geq c\lambda^{-r^{\frac{1}{50b}}_{s(k)}}>c\lambda^{-r_{s(k)}^{\frac{1}{50}}},
$$ which implies for $\bar{t},$ Case $\textbf{3}_{s(k)}$ occurs. Hence \eqref{Case3lem} implies \begin{equation}\label{bbb}[t'_{-;\xi},t'_{+;\xi}]:=[\sup\mathcal{H}_{0;\xi},\inf\mathcal{H}_{\pi;\xi}]\subset \R\backslash\Sigma^{\lambda},~\xi\in [1,2].\end{equation}
By the continuity of $g_{s(k)+1}(c_{s(k)+1,1}(\bar{t}),\bar{t})$ with respect to $\bar{t}$, we have for $1\leq \xi\leq 2,$
\begin{equation}\label{gapmidd}
g_{s(k)+1,1}(c_{s(k)+1,1}(t'_{-;\xi}),t'_{-;\xi})=\lambda^{-r^{\frac{1}{50\xi}}_{s(k)}},
\end{equation}
\begin{equation}\label{gapmidd2}
g_{s(k)+1,1}(c_{s(k)+1,1}(t'_{+\xi}),t'_{+;\xi})= \pi-\lambda^{-r^{\frac{1}{50\xi}}_{s(k)}}.
\end{equation}

The following lemma is key to the proof of (ii).

\begin{Lemma}\label{G^*} There exist $t^*_+\in \Sigma^{\lambda}\bigcap \mathcal{H}_{\pi,\xi}$ and $t^*_-\in \Sigma^{\lambda}\bigcap \mathcal{H}_{0,\xi}$ such that
$
(t^*_-,t^*_+)=G^*.
$
\end{Lemma}
\begin{proof}
  We set
\beq\label{gwx4}
t_0:=\{\bar{t}\in Q_{s(k)}(t) \vert g_{s(k)+1}(\tilde{c}_{s(k)+1,1}(\bar{t}),\bar{t})-\pi=0\}.
\eeq
This means for $t=t_0$,  $g_{s(k)+1}$is tangent to $y=\pi.$ Hence \begin{equation}\label{xiangq} \tilde{c}_{s(k)+1,1}(\bar{t})=c_{s(k)+1,1}(\bar{t}).\end{equation}
By \eqref{lm17-main1'}, \eqref{gwx2} and \eqref{gwx4}, we have
\beq\label{gwx3}
|t-t_0|<C\lambda^{-r^{\frac{1}{50\xi}}_{s(k)}}.
\eeq

By \eqref{Calpha}, \eqref{gwx3}, Corollary \ref{lmg'} and \eqref{i11}, we have
\begin{align}\label{gwx5}
\nonumber &\inf\limits_{j\in \Z} \vert c_{s(k),1}(t_0)+ k\alpha- c_{s(k),2}(t_0)-j\vert\\
\nonumber \leq& \inf\limits_{j\in \Z}\vert c_{s(k),1}(t)+k\alpha-c_{s(k),2}(t)-j\vert+ C\lambda^{Cq_{N+s(k)-2}}|t-t_0|\\
\leq &\lambda^{-r_{s(k)-1}^{\frac{1}{55}}}+\lambda^{Cq_{N+s(k)-2}}\cdot \lambda^{-r^{\frac{1}{50\xi}}_{s(k)}}\leq \lambda^{-q^{\frac{1}{2}}_{N+s(k)-2}}\leq q^{-2000\tau}_{N+s(k)}\leq \frac{1}{2}|I_{s(k)+1}|.
\end{align}

On the other hand, note that \eqref{gwx4} implies $\tilde{I}_{s(k),1}$ is non-empty, thus by \eqref{cc} and \eqref{xiangq}, we have
\beq\label{gwx6}
\inf\limits_{j\in \Z} \vert c_{s(k)+1,1}(t_0)+ k\alpha- c_{s(k)+1,2}(t_0)-j\vert<\lambda^{-\frac{1}{10}r_{s(k)}}\ll |I_{s(k)+1}|.
\eeq
\eqref{gwx5} and \eqref{gwx6} imply that, for $t_0$ both Case $\textbf{2}_{s(k)}$ and Case $\textbf{2}_{s(k)+1}$ occur. By \eqref{gn-gn+1}, we have
$$
\|g_{s(k)+1}(\cdot,t_0)-g_{s(k)+2}(\cdot,t_0)\|_{C^2(I_{s(k)}(t_0))}\leq \lambda^{-\frac{3}{2}r_{s(k)}}.
$$
Combining the above with \eqref{xiangq}, it holds that
$$
|g_{s(k)+2}(\tilde{c}_{s(k)+2,1}(t_0),t_0)-\pi|=|g_{s(k)+2}(c_{s(k)+2,1}(t_0),t_0)-\pi|\leq \lambda^{-\frac{1}{9}r_{s(k)}}.
$$

Hence by the similar argument as above, we can find some $t_1$ such that
\beq\label{gwx7}
|t_0-t_1|<C\lambda^{-\frac{1}{9}r_{s(k)}},
\eeq
\beq\label{gwx8}
|g_{s(k)+2}({c}_{s(k)+2,1}(t_1),t_1)-\pi|=|g_{s(k)+2}(\tilde{c}_{s(k)+2,1}(t_1),t_1)-\pi|=0,
\eeq
\beq\label{877}
|g_{s(k)+1}(\tilde{c}_{s(k)+1,1}(t_1),t_1)-\pi|\leq \lambda^{5|k|}|t_0-t_1|\leq \lambda^{5|k|} \lambda^{-\frac{1}{9}r_{s(k)}}\leq \lambda^{-r_{s(k)}^{\frac{99}{100}}}.
\eeq
By \eqref{gwx7}, \eqref{gwx5} and Corollary \ref{lmg'}, we have
\beq\label{gwx5'}
\inf\limits_{j\in \Z} \vert c_{s(k),1}(t_1)+ k\alpha- c_{s(k),2}(t_1)-j\vert\ll |I_{s(k)}|.
\eeq
Note that \eqref{gwx8} implies $\tilde{I}_{s(k)+1,1}$ is non-empty. Thus by \eqref{cc}, \eqref{maxmax}, \eqref{lm17-main1'} and \eqref{877},
\beq\label{gwx5''}
\inf\limits_{j\in \Z} \vert c_{s(k)+1,1}(t_1)+ k\alpha- c_{s(k)+1,2}(t_1)-j\vert<\lambda^{-r_{s(k)}^{\frac{1}{99}}}\ll |I_{s(k)+1}|.
\eeq
Moreover, by \eqref{cc} and
\eqref{gwx8}, $\tilde{I}_{s(k)+1,1}$ is non-empty. Then \eqref{cc}, \eqref{maxmax}, \eqref{lm17-main1'} and \eqref{gwx8} imply that
\beq\label{gwx5'''}
\inf\limits_{j\in \Z} \vert c_{s(k)+2,1}(t_1)+ k\alpha- c_{s(k)+2,2}(t_1)-j\vert<\lambda^{-r_{s(k)+1}^{\frac{1}{100}}}\ll |I_{s(k)+2}|.
\eeq
Thus by \eqref{gwx5'}-\eqref{gwx5'''},
for $t_1,$ Case $\textbf{2}_{s(k)}$, Case $\textbf{2}_{s(k)+1}$ and Case $\textbf{2}_{s(k)+2}$ occur.

Repeating the above process for $p$ times, we can obtain $t_{p+1}$ such that
\beq\label{contp}
|t_p-t_{p+1}|<\lambda^{-\frac{1}{9}r_{s(k)+p-1}},
\eeq
\beq\label{congp}
|g_{s(k)+p+2}(\tilde{c}_{s(k)+p+2,1}(t_{p+1}),t_{p+1})-\pi|=0,
\eeq
 \beq\label{congp+1}
 |g_{s(k)+p+1}(\tilde{c}_{s(k)+p+1,1}(t_{p+1}),t_{p+1})-\pi|\leq \lambda^{-r^{\frac{99}{100}}_{s(k)+p-1}}.
 \eeq
Again by \eqref{cc} and a similar argument as above, for $t_{p+1}$, all the $(s(k)+j)$-th  steps with $0\leq j\leq p+1$ belong to Case \textbf{2}. \footnote{Note that  Case \textbf{2} is defined by the resonant condition $|c_{n,2}-c_{n,1}-k\alpha|\le |I_n|^{-C}$, while $|t_i-t_{i+1}|\le C^{-i}\cdot\lambda^{-|I_n|^{c}}$ with $0<c<1<C$, which guarantees that for $t_{p+1}$, the $(s(k)+j)$-th step belongs to Case {\bf 2} for $j=1,2,\cdots,p$.}

Notice that \eqref{contp} allows us to define
$
t^*_+=\lim\limits_{p\rightarrow +\infty} t_p.
$
By \eqref{lm17-main1'}, Corollary \ref{lmg'}, \eqref{contp} and \eqref{congp+1}, we have
\begin{equation}\label{g*}~|g_{s(k)+p+1}(\tilde{c}_{s(k)+p+1,1}(t^*_+),t^*_+)-\pi|\leq C\lambda^{-r^{\frac{99}{100}}_{s(k)+p-1}},~p\geq 1.\end{equation}
By \eqref{g*} and Lemma \ref{lem20}, $t^*_+\in \Sigma^{\lambda}\bigcap \mathcal{H}_{\pi;b}$.

The following claim is key to the proof of Lemma \ref{G^*}.

{\bf Claim.} \quad $t^*_+$ is the right endpoint of a spectral gap.
\vskip 0.2cm
{\it Proof of Claim.}\quad In fact, for any
$
0<\delta<\lambda^{-r^{\frac{1}{75}}_{s(k)}},
$
there is  $L\in \N$ such that
 $$
 \lambda^{-r^{\frac{1}{200}}_{L+s(k)-2}}>\delta>\lambda^{-r^{\frac{1}{200}}_{L+s(k)-1}}.
 $$
Next we prove that for $t^*_+-\delta,$ Case $\textbf{3}_{K}$ occurs with some $K\leq L+s(k).$ In fact, by \eqref{g*}, we have
$$
\min\limits_{x\in I_{s(k)+L-1}(t_+^*)}|g_{s(k)+L}(x,t^*_+)-\pi|\leq C\lambda^{-r^{\frac{1}{100}}_{L+s(k)-1}}.
$$
By  the monotonicity of $g_{s(k)+L}(\tilde{c}_{s(k)+L,1}(\bar{t}),\bar{t})$ with respect to $\bar{t}$ in \eqref{lm17-main1'} and \eqref{gwx2}, we have
\begin{align*}
&\min\limits_{x\in I_{s(k)+L-1}(t_+^*-\delta)}|g_{s(k)+L}(x,t^*_+-\delta)(\rm{mod}~\pi)|\\
=&|g_{s(k)+L}(\tilde{c}_{s(k)+L,1}(t_+^*-\delta),t^*_+-\delta)-\pi|\geq c\delta-C\lambda^{-r^{\frac{1}{100}}_{L+s(k)-1}}>c\lambda^{-r^{\frac{1}{200}}_{L+s(k)-1}}.
\end{align*}
If all step $s(k)\leq j< s(k)+L$ does not belong to Case \textbf{3}, then the $(s(k)+L)$-th step belongs to Case \textbf{3}. Therefore,
\begin{equation}\label{youd}
(t^*_+-\lambda^{-r^{\frac{1}{75}}_{s(k)}},t^*_+)\in \R\backslash \Sigma^\lambda; \ \ t^*_+\in \Sigma^{\lambda}.
\end{equation}
\qed

Similarly we can define $t^*_-$ (by replacing $\pi$ in (\ref{gwx4}) by $0$) and show that
\begin{equation}\label{zuod}
(t^*_-,t^*_-+\lambda^{-r^{\frac{1}{75}}_{s(k)}})\in \R\backslash\Sigma^\lambda\quad and \quad  t^*_-\in \Sigma^{\lambda}.
\end{equation}

On the other hand, by \eqref{lm17-main1'}, we have
$$\pi-\lambda^{-r^{\frac{1}{76}}_{s(k)}}\leq g_{s(k)+1}(\tilde{c}_{s(k)+1,1}(t_+^*-\lambda^{-r^{\frac{1}{75}}_{s(k)}}),t^*_+-\lambda^{-r^{\frac{1}{75}}_{s(k)}})\leq \pi-\lambda^{-r^{\frac{1}{74}}_{s(k)}}.$$
By the above inequality, \eqref{gapmidd} and \eqref{lm17-main1'}, we have
\begin{equation}\label{youd1}t'_{+,2}\leq t_+^*-\lambda^{-r^{\frac{1}{75}}_{s(k)}}\leq t'_{+,1}.\end{equation}
Similarly,
\begin{equation}\label{zuod1}t'_{-,1}\leq t_-^*+\lambda^{-r^{\frac{1}{75}}_{s(k)}}\leq t'_{-,2}.\end{equation}

Combining \eqref{bbb},~\eqref{youd}, \eqref{zuod}, \eqref{youd1} and \eqref{zuod1}, we have
\begin{equation}\label{bbb1}
(t^*_-,t^*_-+\lambda^{-r^{\frac{1}{75}}_{s(k)}})\cup [t'_{-;1},t'_{+;1}] \cup (t^*_+-\lambda^{-r^{\frac{1}{75}}_{s(k)}},t^*_+)=(t^*_-,t^*_+).
\end{equation}
Thus $(t^*_-,t^*_+)$ must be some spectral gap.

Next we show $(t^*_-,t^*_+)=(a,b)=G^*$. It is enough to show that $t\in (a,b)$ implies $t\in [t^*_-,t^*_+].$ For $t\in (a,b),$ by Lemma \ref{gapkehua}, there exists $P^*\geq s(k)$ such that Case $\textbf{3}_{P^*+1}$ occurs  with \begin{equation}\label{maodun1}\inf\limits_{j\in \Z}|c_{P^*,1}+k\alpha-c_{P^*,2}-j|\leq \lambda^{-r^{\frac{1}{55}}_{P^*-1}}.\end{equation}

If $P^*=s(k),$ then for $t,$ Case $\textbf{3}_{s(k)+1}$ occurs. Hence
\begin{equation}\label{1212}\inf_{x\in I_{s(k),1}}|g_{s(k)+1}(x,t)|>\lambda^{-r^{\frac{1}{50}}_{s(k)}}.\end{equation} Note in this case, $\{\tilde{c}_{s(k),1},\tilde{c}'_{s(k),1}\}\subset I_{s(k),1}.$ By \eqref{gwx2}, \eqref{bbb1} and \eqref{1212}, we obtain
$$t\in [t'_{-;1},t'_{+;1}]\subset [t^*_-, t^*_+].$$

Next we consider the case $P^*\geq s(k)+1.$ By \eqref{case32}, all $l-$th $(s(k)\leq l\leq P^*)$ steps belong to Case \textbf{2} and (i) guarantees that
$$\inf\limits_{j\in \Z}|c_{s(k),1}(t)+k\alpha-c_{s(k),2}(t)-j|\leq C\lambda^{-r^{\frac{1}{100}}_{s(k)-1}}.$$
For $t^*_+,t^*_-,$ by our construction, all $l-$th $(s(k)\leq l\leq P^*)$ steps belong to Case \textbf{2} and for $X\in\{+,-\},$
$$\inf\limits_{j\in \Z}|c_{s(k),1}(t^*_{X})+k\alpha-c_{s(k),2}(t^*_{X})-j|\leq C\lambda^{-r^{\frac{1}{100}}_{s(k)-1}}.$$
Thus by \eqref{lmg'-2}, we have
$$|t-t^*_{-}|+|t-t^*_{+}|\leq C\lambda^{-r^{\frac{1}{101}}_{s(k)-1}}.$$
It follows that
$$
t\in (t^*_--\lambda^{-r^{\frac{1}{200}}_{s(k)-1}},t^*_-+\lambda^{-r^{\frac{1}{200}}_{s(k)-1}})\bigcup (t^*_+-\lambda^{-r^{\frac{1}{200}}_{s(k)-1}},t^*_++\lambda^{-r^{\frac{1}{200}}_{s(k)-1}}).
$$
Now assume that $t=t^*_++\delta'$ with $0<\delta'<\lambda^{-r^{\frac{1}{200}}_{s(k)-1}}.$ Let $L'\in \Z_+$ satisfy
$$
\lambda^{-r^{\frac{1}{200}}_{L'+s(k)-1}}<\delta'<\lambda^{-r^{\frac{1}{200}}_{L'+s(k)-2}}.
$$
By \eqref{lm17-main1'} and the construction of $t^*_+$, for $t^*_+$, all $l-$th $(s(k)+1\leq l\leq s(k)+L')$ step belong to Case \textbf{2} with $$g_{s(k)+L'}(\tilde{c}_{s(k)+L',1}(t^*_++\delta'),t^*_++\delta')>\pi+\lambda^{-r^{\frac{1}{199}}_{L'+s(k)-1}},$$
which, together with \eqref{maxmax1} and \eqref{cc}, implies
$$
\inf\limits_{j\in \Z}|c_{s(k)+L',1}+k\alpha-c_{s(k)+L',2}-j|>c\lambda^{-r^{\frac{1}{198}}_{L'+s(k)-1}}.
$$

Note that by  \eqref{ga1} and \eqref{ga2}, if the induction continues to $(L'+i+1)$-th step, we have
$$
|c_{s(k)+L'+i,j}-c_{s(k)+L'+i+1,j}|\leq C\lambda^{-\frac{1}{2}r_{s(k)+L'+i-1}},~j=1,2.
$$
Then for any subsequent step $P\geq L'+s(k)$, $P-$th step is not in Case \textbf{3}, that is, $$\inf\limits_{j\in \Z}|c_{P,1}+k\alpha-c_{P,2}-j|\geq \lambda^{-r^{\frac{1}{197}}_{L'+s(k)-1}}.$$ This leads to a contradiction with \eqref{maodun1} by taking $P=P^*$. Hence we obtain $t\leq t^*_+.$ Similarly, we have $t\geq t^*_-.$ Therefore, we have $t\in [t^*_-,t^*_+]$ as desired.
\end{proof}


We are in a position to finish the proof of (ii) of Lemma \ref{lemma9}.

For \eqref{gapest}, note that our construction guarantees that $t^*_+\in 2\mathcal{G}_1$ and $t^*_-\in 2\mathcal{G}_2.$ Then \eqref{lm17-main1'} and \eqref{gwx2} immediately imply \eqref{gapest}.

For (\ref{def-K}), by the above construction, for $t^*_X,$ step $s(k)+p,~p\in \N$ always belongs to Case \textbf{2} with
$$
\inf\limits_{j\in \Z}|c_{s(k)+p,1}(t^*_X)+k\alpha-c_{s(k)+p,2}(t^*_X)-j|\leq C\lambda^{-r^{\frac{1}{100}}_{s(k)+p-1}}.
$$
By \eqref{congp}, for $t_p,$ $g_{s(k)+p+1,1}(t_p)$ is always tangent to $y=\pi$ with
$$
|t_p-t^*_+|\leq C\lambda^{-r^{\frac{1}{100}}_{s(k)+p-1}}.
$$
Similarly we can construct another sequence \( t'_p \) that approaches the other endpoint $t^*_-$ with
$$
|t'_p-t^*_-|\leq C\lambda^{-r^{\frac{1}{100}}_{s(k)+p-1}}
$$
and $g_{s(k)+p+1,1}(t'_p)$ is tangent to $y=0.$  This completes all the proof of (ii).

\end{proof}

\textbf{Final proof of Theorem \ref{15}.}
Note that Lemma \ref{lemma9} allows us to define a set function $\Phi: \R\backslash\Sigma^{\lambda}\rightarrow \Z$ as
$\Phi(G^*)=k$ for an opening spectral gap $G^*$ and $k\in \Z$ such that \eqref{def-K} holds true. Clearly, $K(\lambda)=\text{Ran}(\Phi)$ and
$$\{{\rm all\ (openning)~gaps\ of}\  H_{\lambda v,\alpha,x}\}=\{G_k^{\lambda}=(t^k_-(\lambda), t^k_+(\lambda))| k\in K(\lambda)\subset\mathbb{Z}\}.$$

Note (1) and (2) directly follow from (i) and (ii) of Lemma \ref{lemma9}. The first part of (3) directly follows from \eqref{gwx2} and \eqref{lm17-main1'}. The second part of (3) follows from the fact \eqref{g_nmin1} in Theorem \ref{theorem12} by taking $n=s(k)$.

For (4), one notes that for any fixed \( k\in \Z\), we can choose \( q_N \)  and \( \tilde{\lambda}>\lambda_0 \) sufficiently large  such that \( |k| \ll q_N\ll \tilde{\lambda}\). Then, since $v$ is of cosine-type, for $\lambda>\tilde{\lambda},$ there exists some \( t^*_k\in [\inf v-\frac{2}{\tilde{\lambda}},\sup v+\frac{2}{\tilde{\lambda}}] \) such that \( t^*_k - v(x-\alpha) \) has two zeros with a distance \( k\alpha(\!\!\!\!\mod 1) \), implying that  $c_{1,1}(t^*_k,\lambda)+k\alpha=c_{1,2}(t^*_k,\lambda).$ Then \eqref{gap-mid} implies $k\in K(\lambda)$ and then $t^*_k\in G_k^{\lambda}=(t^k_-,t^k_+)$ for any $\lambda>\tilde{\lambda}.$ For the remaining part, note for $\lambda>\tilde{\lambda},$ $s(k)=1.$ Hence $\{c_{1,1}(t^*_k),c_{1,2}(t^*_k)\}=\{c_{s(k),1}(t^*_k),c_{s(k),2}(t^*_k)\}=\{x\vert g_{s(k)}(x,t^*_k)=0\}=\{x\vert g_{1}(x,t^*_k)=0\}=\{x \vert t^*_k=v(x-\alpha)\},~j=1,2.$ Then \eqref{var-ep}, \eqref{pugapguj} and the cosine-type condition of $v$ with Remark \ref{rmk2} complete the proof of \eqref{xzengd}.
\section {\bf Proof of Theorem \ref{theorem12}.}\label{proofofinduction}

To obtain what we desire in Theorem \ref{theorem12}, we need to prove the followings inductively.
\begin{equation}\label{norm-lower-bound}
\|A_{r^\pm_n}(x,t)\| \geq \lambda^{(1-\epsilon)r_n},
\end{equation}

\begin{equation}\label{norm-derivative}
\left| \partial_t^i \partial_x^j \|A_{r^\pm_n}(x,t)\| \right| \leq \|A_{r^\pm_n}(x,t)\|^{1 + 10^{-1}}, \quad i, j = 0, 1, 2.
\end{equation}

Now we fix $t^*\in \R$ and $t\in Q_n(t^*)$ at the $n$-th step in the following proof.

\textbf{From $0$-th Step to $1$-th Step:}
The theorem holds for $n=0$ to $n=1$ if we assume that $N$ and $\lambda$ are sufficiently large. Let
$$
Q_0=Q_1,\ \ I_{0,j}=I_{1,j},\ \ c_{0,j}=c_{1,j},\ \ r_{0}=1,\ \ g_{0,j}=g_{1,j},\ \ j=1,2.
$$
We define the sequence $\{\lambda_n\}_{n\geq 0}^\infty$ by
$$
\log \lambda_n:=\log\lambda_{n-1}-q_{N+n-1}^{-1/10}\log\lambda_{n-1}
$$
with $\lambda_0=\lambda$. It is easy to see that for any $\varepsilon>0$, there exists a large integer $N$ such that $\lambda_n$ decreases to some $\lambda_\infty$ with $\lambda_\infty>\lambda^{1-\varepsilon/4}$.
It holds by Lemma \ref{lemma4} and the assumption on $v$ that $g_0$ is a $C^2$ cosine-type function.
Then it is not hard to verify that if $\lambda\gg 1$, then \eqref{norm-lower-bound}, \eqref{norm-derivative} and \eqref{maxmax}-\eqref{I-second-derivative} hold true if $\{c_{0,j},j=1,2\}$ locates far from the zeros of $\partial_x g_0(x,t)=0$, and \eqref{ga110}-\eqref{lm17-main1} hold true if $\{c_{0,j},j=1,2\}$ locates close to the zeros of $\partial_x g_0(x,t)=0$. For Case \textbf{3}, note
\begin{equation}\label{step11}
\min\limits_{x\in I_{0}}|g_{1}(x,t)(\text{mod}~\pi)|=\min\limits_{x\in I_{0}}|\arctan(t-v)+o(\lambda^{-1})(\text{mod}~\pi)|>\lambda^{-\frac{1}{3}r^{\frac{1}{50}}_{0}}\geq \lambda^{-\frac{1}{3}}
\end{equation}
and \begin{equation}\label{step11*}\min\limits_{x\in I_{0}}|g_{1}(x,t)(\text{mod}~\pi)|=\min\limits_{x\in I_{0}}|\arctan(t-v)+o(\lambda^{-1})(\text{mod}~\pi)|\leq C\lambda^{-1}\end{equation}
 $$\text{for}~t\in \mathcal{I}(=[-\frac{2}{\lambda}+\inf v, \sup v +\frac{2}{\lambda}]).$$ Then \eqref{step11}, \eqref{step11*} and the fact $\Sigma^{\lambda}\subset \mathcal{I}$ imply $t\notin \mathcal{I}.$ Then \eqref{Case3lem} holds true. \eqref{cn1cn2} and \eqref{gn-gn+1} are clear to hold true by the definition of $c_{0,j}$ and $g_{0,j},~j=1,2.$

Now we assume that the above estimates in the theorem hold for the $n$-th step ($n\geq 0$).
We aim to prove them for the $(n+1)$-th step.
Let $x\in I_{n+1}$.
We can rewrite $A_{r^+_{n+1}}(x,t)$ as
\begin{equation}\label{product3}
A_{r^+_{n+1}}(x,t)=A_{t_{s}-t_{s-1}}(T^{ t_{s-1}}x,t)\cdots A_{ t_{j}-t_{j-1}}(T^{ t_{j-1}}x,t)\cdots A_{t_{1}-t_0}(x,t),
\end{equation}
where $t_{j+1}-t_{j}=r^+_{n}(T^{t_j}x,t)$ ($0\leq j\leq s-1$), $t_0=0$ and
$r^+_{n+1}=t_{s}$. We are going to estimate the norm and the angle of $A_{r^+_{n+1}}(x,t)$ from the assumptions on each $A_{ t_{j}-t_{j-1}}(T^{ t_{j-1}}x,t)$.

To study (\ref{product3}), we first give estimates on the product of two matrices
\begin{equation}\label{product2} A_{n_1+n_2}(x,t)=A_{n_2}(T^{n_1}x,t)\cdot A_{n_1}(x,t).\end{equation}
Let $\theta=\frac\pi2-(s_{n_2}(T^{n_1}x,t)-u_{n_1}(T^{n_1}x,t))$, $l_{n_1+n_2}=\|A_{n_1+n_2}(x,t)\|$, $l_{n_1}=\|A_{n_1}(x,t)\|$ and $l_{n_2}=\|A_{n_2}(T^{n_1}x,t)\|$.

Then we the following four lemmas, for which  the proof  are given in the Appendix A.

\begin{Lemma}\label{arctan} Let $0<\eta\leq  10^{-2}$.
Then it holds that
$$l_{n_1+n_2}\geq \left\{
    \begin{array}{ll}
l_{n_2} l_{n_1}  |\cos\theta|, & \hbox{$|\theta-\frac\pi2|>\max \{l_{n_1}^{-\eta},l_{n_2}^{-\eta}\}$}; \\
\max\{l_{n_2}l_{n_1}^{-1},l_{n_2}^{-1}l_{n_1}\}, & \hbox{general case}.
                                                                             \end{array}
                                                                           \right.$$
Moreover,
$$s_{n_1+n_2}(x,t)=-\frac12{\rm arccot\ } f(l_{n_1}, l_{n_2}, \theta)+s_{n_1}(x,t),$$
$$u_{n_1+n_2}(T^{n_1+n_2}x,t)=-\frac12\arctan f(l_{n_2},l_{n_1},\theta)+u_{n_2}(T^{n_1+n_2}x,t),$$
where \begin{equation}\label{product-two-matrix}
f(l_{n_1}, l_{n_2}, \theta)
=-\frac{l_{n_1}^2-l_{n_1}^{-2}l_{n_2}^{-4}}{2(1-l_{n_2}^{-4})}\cot\theta -\frac{l_{n_1}^2l_{n_2}^{-4}-l_{n_1}^{-2}}{2(1-l_{n_2}^{-4})}\tan\theta.
\end{equation}
\end{Lemma}

\begin{Lemma}\label{norm_derivative}
Let $0<\eta\leq 10^{-2},\ 0<\mu\leq  10^{-2}$.
Assume that $|\partial_x^{m_1}\partial_t^{m_2} l_{n_i}|\leq |l_{n_i}|^{1+(m_1+m_2)\eta}$ and $|\partial_x^{m_1}\partial_t^{m_2}\theta|\leq \min\{l_{n_1}^{\eta},l_{n_2}^{\eta}\}$ for $i=1,2$ and $m_1+m_2=1,2$.
Then it holds that
$$\Big|\partial_x^{m_1}\partial_t^{m_2}l_{n_1+n_2}\Big|\leq \left\{
    \begin{array}{ll}
l_{n_1+n_2}^{1+(m_1+m_2)\eta}, & \hbox{$|\theta-\frac\pi2|>\max \{l_{n_1}^{-\eta},l_{n_2}^{-\eta}\}$} \\
l_{n_1+n_2}^{1+(m_1+m_2)\eta+2(m_1+m_2)\eta\mu}, & \hbox{$l_{n_1}^\mu\geq l_{n_2}$ or $l_{n_2}^\mu\geq l_{n_1}$}
                                                                             \end{array}.
                                                                           \right.$$
\end{Lemma}
\vskip 0.4cm

Without loss of generality, we only consider $s_{n_1+n_2}$. Similar properties hold for $u_{n_1+n_2}$.
Then we need to figure out the properties of
\begin{equation}\label{phi-theta}
 \tilde\phi(x,t):=-\frac12{\rm arccot\ }f(l_{n_1}, l_{n_2}, \theta).
\end{equation}

\begin{Lemma}\label{nonresonance_phi}
Assume that $|\partial_x^{m_1}\partial_t^{m_2} l_{n_i}|\leq |l_{n_i}|^{1+(m_1+m_2)\eta}$ and $|\partial_x^{m_1}\partial_t^{m_2} \theta|\leq \min\{l_{n_1}^{\eta},l_{n_2}^{\eta}\}$  for $i=1,2$ and $m_1+m_2=1,2$.
If $l_{n_2}<l_{n_1}^{\mu}$ or $|\theta-\frac\pi2|>\max \{l_{n_1}^{-\eta},l_{n_2}^{-\eta}\}$, then
$$\Big|\partial_x^{m_1}\partial_t^{m_2}\tilde \phi(x,t)\Big|\leq l_{n_1}^{-2+2(m_1+m_2)\eta},\quad m_1+m_2=0,1,2.$$
\end{Lemma}


\begin{Lemma}\label{shape-phi-n+1}
Assume $l_{n_2}^\mu> l_{n_1}$, $|\partial_x^{m_1}\partial_t^{m_2}\theta(x,t)|\leq l_{n_1}^{\eta}$, $x\in I_{n+1,2}$ \footnote{or $x\in I_{n+1,1}$} and $|\partial_x^{m_1}\partial_t^{m_2} l_{n_i}|\leq |l_{n_i}|^{1+(m_1+m_2)\eta}$ ($m_1+m_2=0,1,2$). 
Suppose that $\theta(x,t)$ satisfies (\ref{lm17-main}), (\ref{I-zero-derivative}),  (\ref{I-first-derivative}) and (\ref{I-second-derivative}).
Then $\tilde\phi(x,t)$ has the following properties.
\begin{enumerate}
\item[(i):] $\tilde\phi(x,t)$ increases from $0$ to $\pi$ as $\theta$ changes from $0$ to $\pi$. Moreover, if $|\frac{\pi}{2}-\theta|< l_{n_1}^{-\eta}$, then $\tilde\phi$ increases from $Cl_{n_1}^{-2+\eta}$ to  $\pi-Cl_{n_1}^{-2+\eta}$.
\item[(ii):] If  $|\frac{\pi}{2}-\theta|\geq l_{n_1}^{-1+\eta}$, then it holds that $|\partial^j_x\tilde\phi(x,t)|+|\partial^j_t\tilde\phi(x,t)|\leq  l_{n_1}^{-2+j\eta}\cdot|\frac{\pi}{2}-\theta|^{-(j+1)}$ $(j=0,1,2)$.  Moreover, if $|\frac{\pi}{2}-\theta|\geq l_{n_1}^{-\eta}$, then $|\partial_x^j \tilde \phi(x,t)|+|\partial_t^j \tilde \phi(x,t)|\leq l_{n_1}^{-2+5\eta}$ $(j=0,1,2)$.

\item[(iii):] If  $l_{n_1}^{-1-\eta}<|\frac{\pi}{2}-\theta|< l_{n_1}^{-1+\eta}$, then
$Cl_{n_1}^2\geq |\partial^2_X\tilde\phi(x,t)|\ge l_{n_1}^{1-6\eta}\gg 1,~X\in \{x,t\}.$
\item[(iv):] If $|\frac{\pi}{2}-\theta|\leq l_{n_1}^{-1-\eta}$, then $|\partial_X\tilde\phi(x,t)|\geq l_{n_1}^\eta\gg 1,~X\in\{x,t\}.$
\item[(v):] If $|\frac{\pi}{2}-\theta|\leq  l_{n_1}^{-2-\eta}$, then $|\partial_X\tilde\phi(x,t)|\geq l_{n_1}^{2-3\eta},~X\in \{x,t\}.$
At this time, $|\tilde \phi(x,t)-\frac\pi2|\leq c\pi$, where $0<c<\frac12$.
\item[(vi):] $|\partial_x^{m_1}\partial_t^{m_2} \tilde{\phi}(x,t)|\leq l_{n_1}^9$ with $m_1+m_2=2.$
\end{enumerate}
\end{Lemma}



\vskip 0.2cm

Now we go back to (\ref{product3}).
Since the step $n$ and step $n+1$ may be resonant or non-resonant, we divide the proof into four cases:
the case from Case $\bf 1_n$ to Case $\bf 1_{n+1}$,
the case from Case $\bf 1_n$ to Case $\bf 2_{n+1}$,
the case from Case $\bf 2_n$ to Case $\bf 2_{n+1}$,
and the case from Case $\bf 2_n$ to Case $\bf 1_{n+1}$.

\vskip 0.4cm

\noindent{\bf (1)\ The case from Case $\bf 1_n$ to Case $\bf 1_{n+1}$.}

\vskip 0.2cm

In this case, $I_n$ consists of two disjoint intervals $I_{n,i}$, $i=1, 2$ and  (\ref{I-zero-derivative}) - (\ref{I-second-derivative}) hold true for the $n$-th step by the inductive assumption.
We will prove that (\ref{I-zero-derivative}) - (\ref{I-second-derivative}) hold true by replacing $n$ by $n+1$.
Without loss of generality, we only consider $I_{n,1}$ and assume $\partial_x g_{n+1}(x,t)<0$ for $x\in I_{n,1}$. The other case follows similarly.

Let $x\in I_{n+1,1}$.
Note that by Lemma \ref{nonresonance_phi}, the angle between $A_{ t_2-t_1}(T^{ t_1}x,t)$ and $ A_{t_{1}}(x,t)$, that is, $s_{t_2-t_1}(T^{t_1}x,t)-u_{t_1}(T^{t_1}x,t)$, satisfies
$$\|s_{t_2-t_1}(T^{t_1}x,t)-u_{t_1}(T^{t_1}x,t)-g_{n+1}(T^{t_1}x,t)\|_{C^2}<\lambda^{-\frac32 r_{n}}.$$
Since $T^{t_1}x\in I_{n,1}$ and step $n+1$ belongs to Case $\textbf{1}_{n+1}$, we have
$T^{t_1}x\in I_{n,1}\backslash I_{n+1,1}$.
Then by \eqref{Calpha} it holds that
\begin{align*}
|s_{t_2-t_1}(T^{t_1}x,t)-u_{t_1}(T^{t_1}x,t)|\geq& |g_{n+1}(T^{t_1}x,t)|-\lambda^{-\frac32 r_{n}(t)}\geq cq^{-6000\tau}_{N+n}-\lambda^{-\frac32 r_{n}(t)}\\
\geq &cq^{-6000\tau}_{N+n}-\lambda^{-\frac32 q_{N+n-1}}\gg \lambda^{-\eta_n r_{n}},
\end{align*}
where $\eta_n:=q_{N+n}^{-1}$.
Note that in the above inequalities, we used the Diophantine condition.
We also remark that the Diophantine condition is needed to verify the condition $|\theta-\frac\pi2|\geq \max\{l_{n_1}^{-\eta},l_{n_2}^{-\eta}\}$ in Lemma \ref{arctan} - Lemma \ref{nonresonance_phi}.
Then by Lemma \ref{arctan} and Lemma \ref{norm_derivative}, we obtain
$$
\left\{\begin{array}{ll}
&\|A_{t_2}(x,t)\|\geq \lambda_{n}^{(1-C\eta_n)t_2},\\&\big|\partial_x^{m_1}\partial_t^{m_2}\|A_{t_2}(x,t)\|\big|\leq \|A_{t_2}(x,t)\|^{1+C(m_1+m_2)\eta_n},\quad m_1+m_2=1,2.\end{array}\right.$$
By Lemma \ref{nonresonance_phi}, we have
$$\|u_{t_2}(T^{t_2}x,t)-u_{t_2-t_1}(T^{t_2}x,t)\|_{C^2}<\lambda^{-\frac32 r_{n}},$$
which leads to
$$\|s_{t_3-t_2}(T^{t_2}x,t)-u_{t_2}(T^{t_2}x,t)-g_{n+1}(T^{t_2}x,t)\|_{C^2}<\lambda^{-\frac32 r_{n}}.$$
Hence  the estimates of $A_{t_3-t_2}(T^{t_2}x)\cdot A_{t_2}(x)$ is available.

Note that for $0<|k|<q_{N+n}$, $$\inf_{j\in \Z}|d_{n+1}-k\alpha-j|> |I_{n+1}|.$$
Then we have $T^kx\not\in I_{n+1}(t)$.
Hence repeated applications of above arguments lead to
$$\left\{
\begin{array}{ll}
&\|A_{ r_{n+1}^+}(x,t)\|>\lambda_n^{(1-C\eta_n)r_{n+1}^+}\geq \lambda_{n+1}^{r_{n+1}^+},
\\&\Big|\partial_x^{m_1}\partial_t^{m_2}\|A_{ r_{n+1}^{+}}(x,t)\|\Big|<\|A_{ r_{n+1}^+}(x,t)\|^{1+C(m_1+m_2)\eta_n},\quad m_1+m_2=1, 2,
\\&\|s_{r^+_{n+1}}-s_{r^+_{n}}\|_{C^2}<\lambda^{-\frac32r_{n}}.
\end{array}\right.
$$

Thus, \boxed{\eqref{norm-lower-bound}~\text{and}~\eqref{norm-derivative}} hold true.
Similarly, the same estimates hold for  $\|A_{ r_{n+1}^-}(x,t)\|$ and $u_{r^-_{n+1}}$. Hence
\begin{equation}\label{zhishux}\|g_{n+1}-g_{n+2}\|_{C^2}\leq \lambda^{-\frac32r_{n}},\end{equation}
which implies \boxed{\eqref{gn-gn+1}}.

Combining this with the assumption on \eqref{lm17-main} of $g_{n+1},$ we immediately obtain
$$\frac{1}{10}-\lambda^{-\frac32r_{n}} <  \frac{\partial g_{n+1}(x,t)}{\partial t}  < \lambda^{10q_{N+n-1}}+\lambda^{-\frac32r_{n}}.$$
This leads to \boxed{\eqref{lm17-main}} for $g_{n+2}$ as desired. \footnote{\eqref{zhishux} allows us to make an appropriate adjustments to \( g_0 \) to modify the lower bound and the upper bound, ensuring that \eqref{lm17-main} holds for $g_{n+2}$.}
Consequently, $g_{n+2}$ possesses only one zero in each $I_{n+1,j}$, denoted  by $c_{n+1,j},\ j=1, 2$.
Then the above equality implies $|c_{n+2,j}-c_{n+1,j}|<\lambda^{-\frac34r_n}$
and hence \boxed{\eqref{maxmax}} and \boxed{(\ref{I-zero-derivative})-(\ref{I-second-derivative})} hold true if we replace $n$ by $n+1$.


\vskip 0.4cm

\noindent{\bf (2)\ The case from Case $\textbf{1}_n$ to Case $\textbf{2}_{n+1}$.}
\vskip 0.2cm
For this case, we have  $\inf\limits_{j\in \Z}|d_n-i\alpha-j|> |I_{n}|$ for $0<|i|<q_{N+n-1}$ and there exists $0 \leq |k|<q_{N+n}$ such that
$$
\inf\limits_{j\in \Z}|d_{n+1}-k\alpha-j|< |I_{n+1}|.
$$
Without loss of generality, we assume $k\geq 0$, $\partial_xg_{n+1}(x,t)>0$ for $x\in I_{n,1}$ and $\partial_xg_{n+1}(x,t)<0$ for $x\in I_{n,2}$. We further assume that $x\in I_{n+1,2}$ and $T^kx\in I_{n+1,1}$. Notice that
$$
A_{r_{n+1}^+}(x,t)=A_{r_{n+1}^+-k}(T^kx,t)\cdot A_{k}(x,t),\ \ A_{r_{n+1}^-}(x,t)= A_{-r_{n}^-}(x,t)\cdot A_{r_{n+1}^--r_{n}^-}(T^{-r_{n}^-}x,t).
$$
Since $x\in I_{n+1,2}$ and $T^{r_{n+1}^{X}}x\in I_{n+1,2},~X\in~\{+,-\}$ with $|I_{n+1,2}|=2q^{-2000\tau}_{N+n},$
by the diophantine condition, we have
$$2q^{-2000\tau}_{N+n}=|I_{n+1,2}|>\inf\limits_{j\in \Z}|T^{r_{n+1}^{X}}x-x-j|=\inf\limits_{j\in \Z}|T^{r_{n+1}^{X}}\alpha-j|>\frac{\gamma}{(r_{n+1}^{X})^{\tau-1}}.$$
Hence (by taking large $N$) we obtain $$r_{n+1}^{X}>q^{2000}_{N+n},~X\in~\{+,-\},$$ which implies $\boxed{\eqref{ga110}}.$
Then we can obtain with a similar argument in the proof of  ``Case $\textbf{1}_n$ to Case $\textbf{1}_{n+1}$'' that
$$
\left\{\begin{array}{ll}&\|A_{r_{n+1}^+-k}(T^kx,t)\|>\lambda_n^{(1-C\eta_n)(r_{n+1}^+-k)},\quad \|A_{k}(x,t)\|>\lambda_n^{(1-C\eta_n)k};
\\&\Big|\partial_x^{m_1}\partial_t^{m_2}\|A_{ r_{n+1}^{+}-k}(T^kx,t)\|\Big|<\|A_{ r_{n+1}^+-k}(T^kx,t)\|^{1+C(m_1+m_2)\eta_n}; \\&\Big|\partial_x^{m_1}\partial_t^{m_2}\|A_{k}(x,t)\|\Big|<\|A_{ k}(x,t)\|^{1+C(m_1+m_2)\eta_n}, \quad m_1+m_2=1, 2;
\\&\|s_{r_{n+1}^+-k}(T^kx,t)-u_k(T^kx,t)-g_{n+1}(T^kx,t)\|_{C^2}, \quad \|u_{r_{n+1}^-}(x,t)-u_{r_{n}^-}(x,t)\|_{C^2}\leq \lambda^{-\frac32r_{n}}.
\end{array}\right.$$
By Lemma \ref{arctan} and Lemma \ref{norm_derivative}, we obtain
$$\left\{\begin{array}{ll}&\|A_{ r_{n+1}^+}(x,t)\|>\lambda_n^{(1-C\eta_n)r_{n+1}^+-3|k|}\geq \lambda_n^{(1-C\eta_n)q_{N+n}^{2000}-3q_{N+n}}\geq \lambda_{n+1}^{r_{n+1}^+},\\
&\Big|\partial_x^{m_1}\partial_t^{m_2}\|A_{ r_{n+1}^{+}}(x,t)\|\Big|<\|A_{ r_{n+1}^+}(x,t)\|^{1+2C(m_1+m_2)\eta_n},\quad m_1+m_2=1, 2.
\end{array}.\right.$$
Thus, \boxed{\eqref{norm-lower-bound}~\text{and}~\eqref{norm-derivative}} hold true.
Moreover, by Lemma \ref{arctan}, we have
\begin{equation}\label{defgn+2}
g_{n+2}(x,t)=\phi(x,t)+g_{n+1,2}(x,t),\quad x\in I_{n+1,2},\end{equation}
$$\phi(x,t)=-\frac12{\rm arccot\ }f\big(\|A_{r_{n+1}^+-k}(T^kx,t)\|,\|A_{k}(x,t)\|,\frac\pi2-g_{n+1,1}(T^kx,t)\big),$$ $$g_{n+1,1}=s_{r_{n+1}^+-k}-u_k,\quad g_{n+1,2}=s_k-u_{r_{n+1}^{-}}\ \ \text{and}\ \ \|g_{n+1,j}-g_{n+1}\|_{C^2}\leq \lambda^{-\frac32r_{n}}.$$
Note that
\begin{align*}
\|A_{r_{n+1}^+-k}(T^kx,t)\|^{2q_{N+n}^{-1}} &\geq \lambda^{(1-C\eta_n)(r_{n+1}^+-k)\cdot (2q_{N+n}^{-1})}\geq \lambda^{(1-C\eta_n)(q_{N+n}^{2000}-k)\cdot (2q_{N+n}^{-1})}\\
&\geq \lambda^{(1-C\eta_n)(q_{N+n}^{2000}-q_{N+n})\cdot (2q_{N+n}^{-1})}\gg \lambda^{|k|^{100}}\gg \|A_{k}(x,t)\|.
\end{align*}
Thus we can apply Lemma \ref{shape-phi-n+1} to study $g_{n+2}$.
For simplicity, we further assume $|d_{n+1}-k\alpha|\le 40^{-1}|I_{n+1}|$,
which implies
\begin{align}\nonumber
I_{n+1,1}+k\alpha&=[c_{n+1,2}-q_{N+n}^{-2000\tau} -d_{n+1}+k\alpha, c_{n+1,2}+q_{N+n}^{-2000\tau} -d_{n+1}+k\alpha]\\ &\hspace{2cm} \supset [c_{n+1,2}-20^{-1}q_{N+n}^{-2000\tau} , c_{n+1,2}+20^{-1}q_{N+n}^{-2000\tau} ].\label{distance-resonant}
\end{align}
Hence by (\ref{I-first-derivative}) (note that $c_{n+1,2}$ is the unique zero of $g_{n+1}$ in $I_{n,2}$),
\begin{equation}\label{range-resonance}
g_{n+1}(I_{n+1,1}+k\alpha)+\pi/2\supset [\pi/2-20^{-1}q_{N+n}^{-6000\tau},\ \pi/2+20^{-1}q_{N+n}^{-6000\tau}].
\end{equation}
This implies that the jump part in the middle of the graph of $\tilde\phi$ is contained in $I_{n+1,2}$.
Moreover, note $g_{n+2}(x,t)=\phi(x,t)+g_{n+1,2}(x,t)$.
By Lemma \ref{shape-phi-n+1}, it can be seen that $\partial_xg_{n+2}$ has exactly 2 zeros in $I_{n+1,j}$ denoted by $\{\tilde{c}_{n+2,j},\tilde{c}'_{n+2,j}\}$ with $$|\tilde{c}_{n+2,j}-\tilde{c}'_{n+2,j}|\leq Cl^{-\frac{1}{2}}_{k},$$ which implies \boxed{\eqref{tic}}. Moreover, we can prove that the graph of $g_{n+2}(x,t)$ possesses a shape satisfying (\ref{five-piece}) - (\ref{g_nmin}) as follows (see Figure 2).

\begin{equation}\label{five-piece}
\left\{\begin{array}{ll}
q_{N+n}^2\ge -\partial_x{g}_{n+2}\ge q_{N+n}^{-2}, \ {g}_{n+2}\in  [-\frac{\pi}{100},\ \frac{\pi}{100} ], & x \le C_1\\
 \\
 l_{k}^{2}\ge\partial^2_x{g}_{n+2}\geq l_{k}^{\mu},\ \partial_x{g}_{n+2} \ \text{has one zero},& x \in  [C_1,C_2] \\
\\
l^2_{k}\geq \partial_x{g}_{n+2}\geq l_{k}^{2-\eta},\ {g}_{n+2}\in  [\frac\pi{10}, \frac{9\pi}{10}], & x \in  [C_2,C_3] \\
\\
  l_{k}^{2}\ge-\partial^2_x{g}_{n+2}\geq l_{k}^{\mu},\ \partial_x{g}_{n+2} \ \text{has one zero},& x \in [C_3,C_4]\\
\\
 q^{2}_{N+n}\geq -\partial_x{g}_{n+2}\geq q^{-2}_{N+n}, {g}_{n+2}\in  [\frac{99\pi}{100},\frac{101\pi}{100} ], & x \ge C_4.\\
\\
~|\partial_x^{m_1}\partial_t^{m_2} g_{n+2}(x,t)|\leq l_{k}^5,\qquad~m_1+m_2=2,
 \end{array}\right.
\end{equation}

where $$C_1:=d_{n+1}-k\alpha +c_{n+1,1}-l_{k}^{-1-\eta};~C_2:=d_{n+1}-k\alpha +c_{n+1,1}-l_{k}^{-2-\eta};$$
$$C_3:= d_{n+1}-k\alpha +c_{n+1,1}+l_{k}^{-2-\eta};~C_4:=d_{n+1}-k\alpha +c_{n+1,1}+l_{k}^{-1+\eta}.$$

\begin{equation}\label{c-n+2-c-n+1}
|c_{n+2,j}-c_{n+1,j}|<C\lambda^{-\frac12r_{n}},\quad j=1,2.
\end{equation}
\begin{equation}\label{range-of-distance}
\pi-l_{k}^{-1+3\eta}\leq \max_{x\in I_{n+1,2}}\{g_{n+2}(x,t)\}-\min_{x\in I_{n+1,2}}\{g_{n+2}(x,t)\} \leq \pi-l_{k}^{-1-3\eta}.
\end{equation}
\begin{equation}\label{2ndderivative}
 |\partial^2_x g_{n+2}|\ge l_k^{1-10\mu}{\rm\ if \ } |\partial_x g_{n+2}|\le 100^{-1}q_{N+n}^{-2}.
\end{equation}
\begin{equation}\label{g-derivative-t}
 \frac{1}{10}\leq \partial_t g_{n+2}(x,t)\leq Cl_k^5.
\end{equation}
\begin{equation}\label{g_nondegenerate}
|g_{n+2}(x,t)|\geq C\cdot \text{dist}(x,X)^3,
\end{equation}where $X$ denotes the set of the zeros of $g_{n+2}$ in $I_{n+1}$ (this implies \boxed{\eqref{maxmax}}).

Besides, if $|c_{n+1,2}+k\alpha-c_{n+1,1}|=0,$ then we have
\begin{equation}\label{g_nmin}\min\limits_{x\in I_{n+1}} |g_{n+2}(x,t)(\rm{mod}~\pi)|>l_k^{-\frac{3}{2}}.
\end{equation}
The proof of (\ref{five-piece}) - (\ref{g_nmin}) will be given in Appendix A.

Then by (\ref{five-piece}) it follows that $\partial_x g_{n+1}$ has at most two zeros ($\tilde{c}_{n+1,j}$ and $\tilde{c}'_{n+1,j}$) in each $I_{n+1,j}$, which divides $I_{n+1,j}$ into at most three monotone intervals.   This together with (\ref{range-of-distance}) implies that $g_{n+1}$ may has at most 2 zeros in each $I_{n,j}$. Hence \boxed{\eqref{either12}} holds true. On the other hand, \boxed{\eqref{ga2}} holds directly from \eqref{c-n+2-c-n+1} and \boxed{\eqref{maxmax1}} follows from the last term of \eqref{five-piece}. And \eqref{2ndderivative} implies \boxed{\eqref{range-g-n}} and \boxed{\eqref{range-g-n-lower-bound}}. \boxed{\eqref{lm17-main1}}follows from \eqref{g-derivative-t}. \boxed{\eqref{lm17-main}} of $g_{n+2}$ directly follows from \eqref{lm17-main1} and the fact $|k|<q_{N+n}.$ \eqref{g_nmin} implies the condition of Case {\bf 3}. We will discuss \eqref{cn1cn2} and \eqref{Case3lem}  in the proof of  Case $\textbf{3}_{n+1}$ later (note \eqref{Case3lem} is a special case of (\ref{g_nmin1})).

 For \eqref{cc}, we need the following technical lemma, which is proved in the Appendix A.

\begin{Lemma}\label{technical_argument}
For $E=E_2E_1\in {\rm SL}(2,\mathbb R)$ satisfying $\|E_2\|\gg\|E_1\|\gg1$, it holds that $$|s(E)-E_1^{-1}\cdot s(E_2)|<C\|E\|^{-2},\quad |s(E_2)-E_1\cdot s(E)|<C\|E_2\|^{-2}.$$
Similarly, for $E=E_2E_1\in {\rm SL}(2,\mathbb R)$ satisfying $\|E_1\|\gg\|E_2\|\gg1$, it holds that
$$|u(E)-E_2\cdot u(E_1)|<C\|E\|^{-2},\quad |u(E_1)-E_2^{-1}\cdot u(E)|<C\|E_1\|^{-2}.$$
\end{Lemma}

\noindent \textbf{The proof of \eqref{cc}:}
\noindent Note that for $A\in \text{SL}(2,\mathbb R)$, the induced map $A:\mathbb {RP}^1\rightarrow \mathbb {RP}^1$ satisfies $\frac{dA}{d\theta}(\theta)=\|A\hat\theta\|^{-2}$, where $\hat\theta\in \theta$ is a unit vector.
By applying lemma \ref{technical_argument} to $A_{r^+_{n+1}-k}(T^kx,t)\cdot A_k(x)\cdot A_{r^{-}_{n+1}}(T^{-r^{-}_{n+1}}x,t)$, one can obtain
$$|g_{n+2}(x,t)-M_1\cdot g_{n+2}(T^{ k} x,t)|<\max\{\|A_{r^+_{n+1}}(x,t)\|^{-\frac32},\|A_{-r^{-}_{n+1}}(x,t)\|^{-\frac32}\},\quad x\in I_{n+1,1},$$
$$|g_{n+2}(x,t)-M_2\cdot g_{n+2}(T^{-k} x,t)|<\max\{\|A_{r^+_{n+1}}(x,t)\|^{-\frac32},\|A_{-r^{-}_{n+1}}(x,t)\|^{-\frac32}\},\quad x\in I_{n+1,2},$$
$$l_k^{-2}<M_j<l_k^{2},~j=1,2,$$
which imply \boxed{\eqref{cc}} for $g_{n+2}$ by (\ref{g_nondegenerate}).


\vskip .5cm

\noindent{\bf (3)\ The case from Case $\textbf{2}_n$ to Case $\textbf{2}_{n+1}$.}


 \vskip .2cm

In this case, we have both Case $\textbf{2}_n$ and Case $\textbf{2}_{n+1}$, which means that there exist $k_{n}$ and $k_{n+1}$ satisfying $0\leq |k_{n}|<q_{N+n-1}$ and $0\leq |k_{n+1}|<q_{N+n}$ such that
$$
\inf\limits_{j\in \Z}|d_{n}-k\alpha-j|<|I_n|,\ \ \inf\limits_{j\in \Z}|d_{n+1}-k\alpha-j|<|I_{n+1}|.
$$

First we assume that $k_{n+1}=k_{n}$. Note that $k_{n+1}$ is the unique integer such that $T^jx\in I_{n+1,1}$ for $0\leq |j|<q_{N+n}$ and $x\in I_{n+1,2}$. Then by the diophantine condition, we have $r^{X}_{n+1} >q_{N+n}^{2000}>k_{n+1}$ and $T^{r^{X}_n}x\notin I_{n+1},~X\in \{+,-\}.$
Without loss of generality, assume $k_{n+1}>0$.
We rewrite $A_{r^+_{n+1}}(x,t)=A_{r^+_{n+1}-r^+_{n}}(T^{r^+_n}x,t)\cdot A_{r^+_{n}}(x,t)$.
Similar to the case from Case $\textbf{1}_n$ to Case $\textbf{2}_{n+1}$, we obtain
$$\|A_{r^+_{n+1}-r^+_{n}}(T^{r^+_n}x,t)\|>\lambda_n^{(1-C\eta_n)(r_{n+1}^+-r^+_n)},$$
$$\Big|\partial_x^{m_1}\partial_t^{m_2}\|A_{r^+_{n+1}-r^+_{n}}(T^{r^+_n}x,t)\|\Big| <\|A_{r^+_{n+1}-r^+_{n}}(T^{r^+_n}x,t)\|^{1+C(m_1+m_2)\eta_n}.$$
We also have the same estimates for $A_{r^+_{n}}(x,t)$ by the inductive assumption.
Note that $T^{r^+_{n}}x\in I_{n}\backslash I_{n+1}$, which implies ${\rm dist} (T^{r^+_{n}}x, C^{(n+1)})\geq cq_{N+n}^{-2000\tau}$.
Then the angle between $A_{r^+_{n+1}-r^+_{n}}(T^{r^+_n}x,t)$ and $A_{r^+_{n}}(x,t)$ is measured by $|g_{n+1}(T^{r^+_n}x,t)|$, which is larger than $cq_{N+n}^{-6000\tau}$.
Then by Lemma \ref{arctan} and Lemma \ref{norm_derivative}, we obtain
$$\|A_{ r_{n+1}^+}(x,t)\|>\lambda_n^{(1-C\eta_n)q_{N+n}^{2000}-3|k|}\geq \lambda_{n+1}^{r_{n+1}^+},$$
$$\Big|\partial_x^{m_1}\partial_t^{m_2}\|A_{ r_{n+1}^{+}}(x,t)\|\Big|<\|A_{ r_{n+1}^+}(x,t)\|^{1+2C(m_1+m_2)\eta_n},\quad m_1+m_2=1, 2.$$
Thus \boxed{\eqref{norm-lower-bound}~\text{and}~\eqref{norm-derivative}} hold true.
By Lemma \ref{nonresonance_phi}, we have
\begin{equation}\label{g-II-II-characteristic}
\|g_{n+1}- g_{n+2}\|_{C^2}\leq l_{k_{n+1}}^{-2+3\eta}<\lambda^{-\frac32r_n},
\end{equation}
which implies \boxed{\eqref{gn-gn+1}}.
Note that all the estimates in {\eqref{maxmax}-\eqref{I-second-derivative}} and {\eqref{cc}-\eqref{lm17-main1} remain true under a small $C^2$ perturbation. Thus we have proved \boxed{\eqref{maxmax}}, \boxed{\eqref{lm17-main}} and \boxed{\eqref{ga110}-\eqref{lm17-main1}} for $(n+1)-$th step.

Now we assume that $k_{n+1}\neq k_{n}$. By the Diophantine condition, we have $0\leq |k_n|<q_{N+n-1}\leq |k_{n+1}|<q_{N+n}$.
This implies $$|d_{n+1}-k_n\alpha|\geq |(k_{n+1}-k_n)\alpha|-|d_{n+1}-k_{n+1}\alpha|>cq^{-\tau}_{N+n}-q^{-1500\tau}_{N+n}\gg |I_{n+1}|.$$
Then it follows that the jump part of $g_{n+1}$ is outside $I_{n+1,2}$. Hence \boxed{(\ref{I-zero-derivative}) - \eqref{I-second-derivative}} hold for $g_{n+1}$ restricting in $I_{n+1,2}$. Then by the discussion of the case from Case $\textbf{1}_n$ to Case $\textbf{2}_{n+1}$, \boxed{\eqref{norm-lower-bound} - \eqref{norm-derivative}}, \boxed{\eqref{maxmax}-\eqref{lm17-main}} and \boxed{\eqref{ga110}-\eqref{range-g-n-lower-bound}} for $(n+1)-$th step. \boxed{\eqref{lm17-main1}} can be obtained by \eqref{g-II-II-characteristic} and the fact $\frac{1}{10}<|\partial_t g_{n+1}|<C\lambda^{5|k|}.$

\vskip 0.4cm
\noindent {\bf (4)\ The case from Case $\textbf{2}_n$ to Case $\textbf{1}_{n+1}$.} This is similar to the case from Case $\textbf{2}_n$ to Case $\textbf{2}_{n+1}$ with $k_{n+1}\neq k_{n}$.

Now, it remains to show \eqref{cn1cn2} and \eqref{Case3lem} of Case $\textbf{3}_{n+1}.$

\noindent \textbf{The proof of \eqref{cn1cn2} and \eqref{Case3lem} of Case $\textbf{3}_{n+1}:$}

\boxed{\eqref{cn1cn2}} follows from \eqref{maxmax} and \eqref{cc}. For \eqref{Case3lem}, one notes $\min\limits_{x\in I_{n+1}}\left\vert g_{n+2}(x,t)(\text{mod}~\pi)\right\vert> \lambda^{-r^{\frac{1}{100}}_{n+1}}.$ Thus by the definition of $I_{j},~j\geq n+1$, we can guarantee that $for~any~j\geq n+1$,
\begin{equation}\label{54}\left\{\begin{array}{ll}&step~j~belongs~to~Case~\textbf{2},\\&~\min\limits_{x\in I_{j}}\left\vert g_{j+1}(x,t)(\rm{mod}~\pi)\right\vert>c\lambda^{-r^{\frac{1}{100}}_{n+1}}.\end{array}\right.\end{equation}

By the diophantine condition,  for each $x\in \R/\Z,$ it holds that $x+m\alpha \in I_{n+1}$ for some $m\leq |I_{n+1}|^{-C}:=M_1.$ Set $M_2=\max\{\max\limits_{x\in I_{n+1}}[r^+_{n+1}(x,t)]^2,\max\limits_{x\in I_{n+1}}[r^-_{n+1}(x,t)]^2\}$. For each $x\in I_{n+1}$ and  each $M\geq \max\{M_2,M_1\},$ let $ 1\leq j_p\leq M,$ $1\leq p\leq m$ be all the times such that
$$j_p-j_{p-1}>|k_n|,~x+j_p\alpha \in I_{n+1},$$ where $j_0=0.$ Then by Proposition \ref{uh} and \eqref{54}, we have
$$\begin{array}{ll}&\| A_{M}(x,t)\|=\|A_{M-j_m}(x+j_m\alpha,t) \prod\limits_{p=1}^m A_{j_p-j_{p-1}}(x+j_{p-1}\alpha,t)\|\\&\geq \left\|A_{M-j_m}(x+j_m\alpha,t)\right\|^{-1}\prod\limits_{p=1}^m \left\|A_{j_p-j_{p-1}}(x+j_{p-1}\alpha,t)\right\|^{\frac{3}{5}}\\&\geq \lambda^{-M^{\frac{1}{2}}}\lambda^{\frac{1}{2}j_m}\geq \lambda^{-M^{\frac{1}{2}}}\lambda^{\frac{1}{2}M}\geq \lambda^{\frac{1}{3}M}.\end{array}$$

Now for any $M\geq \max\{M_2,M_1\}$ and $x\in \R/\Z,$ let $j_1$ be the first time such that $x+j_1\alpha\in I_{n+1},$ then we have
$$\|A_{M}(x,t)\|\geq \|A_{j_1}(x,t)\|^{-1}\|A_{M-j_1}(x+j_1\alpha,t)\|\geq \lambda^{-M}\lambda^{\frac{1}{3}(M-j_1)}\geq \lambda^{\frac{1}{4}M}.$$

Similarly,  for any $M\geq \max\{M_2,M_1\}$ and $x\in \R/\Z,$ we have $\|A_{-M}(x,t)\|\geq \lambda^{\frac{1}{4}M}.$ Therefore, by taking $c=\frac{1}{\lambda^{\frac{1}{4}\max\{M_2,M_1\}}}$ and $\rho=\lambda^{\frac{1}{4}}$,  for each $n\in \Z$ and for each $x\in \R/\Z$, we have
$$\|A_{n}(x,t)\|\geq c\rho^n.$$ Then Proposition \ref{johnson} and Proposition \ref{yocc} complete the proof of \boxed{\eqref{Case3lem}}.

\section*{Appendix A: The proofs of some claims in Section \ref{proofofinduction}}\label{appdb}
\setcounter{equation}{0}
\renewcommand{\theequation}{A.\arabic{equation}}
\setcounter{theorem}{0}
\renewcommand{\thetheorem}{A.\arabic{theorem}}

In this section, we prove some claims in Section \ref{proofofinduction}. In the following proofs, we let $C$, $C_1$ and $C_2$ be positive universal constants.

\noindent{\bf Proof of Lemma \ref{arctan}.}
Note that
$$A_{n_1+n_2}=R_{u_{n_2}}\left(
      \begin{array}{cc}
        l_{n_2} & 0 \\
        0 & l_{n_2}^{-1} \\
      \end{array}
    \right)R_{\theta}\left(
      \begin{array}{cc}
        l_{n_1} & 0 \\
        0 & l_{n_1}^{-1} \\
      \end{array}
    \right)R_{\frac\pi2-s_{n_1}}.
$$
Then it is easy to see
\begin{equation}\label{l-n-1+n-2}
l_{n_1+n_2}^2+l_{n_1+n_2}^{-2}
= l_{n_2}^2l_{n_1}^2\cos^2\theta+l_{n_2}^{-2}l_{n_1}^2\sin^2\theta +l_{n_2}^2l_{n_1}^{-2}\sin^2\theta+l_{n_2}^{-2}l_{n_1}^{-2}\cos^2\theta.
\end{equation}
Then the norm estimate for $l_{n_1+n_2}$ follows.
The explicit formula of $s_{n_1+n_2}(x)$ and $u_{n_1+n_2}(T^{n_1+n_2}x)$ also can be given by a straightforward computation.

\

\noindent{\bf Proof of Lemma \ref{norm_derivative}.}
Assume $|\theta-\frac\pi2|>\max\{l_{n_1}^{-\eta},l_{n_2}^{-\eta}\}$.
Then $|\cos\theta|>\max\{l_{n_1}^{-\eta},l_{n_2}^{-\eta}\}$.
And recall we have
$$\Big|\partial_x^{m_1}\partial_t^{m_2}l_{n_j}\Big|\leq l_{n_j}^{1+(m_1+m_2)\eta},\quad \Big|\partial_x^{m_1}\partial_t^{m_2}\theta\Big|\leq l_{n_j}^{(m_1+m_2)\eta}, \quad m_1+m_2=1,2,\ j=1,2.$$
Take the derivative on both sides of (\ref{l-n-1+n-2}) and reserve the main term. Then a direct computation shows that
$$\Big|\partial_x^{m_1}\partial_t^{m_2}l_{n_1+n_2}\Big|\leq C \Big|\partial_x^{m_1}\partial_t^{m_2}(l_{n_2}l_{n_1}\cos\theta)\Big|\leq  l_{n_1+n_2}^{1+(m_1+m_2)\eta}, \quad m_1+m_2=1,2.$$
Similarly, if $l_{n_1}^\mu\geq l_{n_2}$ (or $l_{n_2}^\mu\geq l_{n_1}$), then for $m_1+m_2=1,2$, it holds that
$$\Big|\partial_x^{m_1}\partial_t^{m_2}l_{n_1+n_2}\Big|\leq C\Big|\partial_x^{m_1}\partial_t^{m_2} \sqrt{l_{n_2}^2l_{n_1}^2\cos^2\theta+l_{n_2}^{-2}l_{n_1}^2\sin^2\theta}\Big|\leq l_{n_1+n_2}^{1+(m_1+m_2)(1+2\mu)\eta}.$$

\

\noindent{\bf Proof of Lemma \ref{nonresonance_phi}.}
Recall
$$\tilde\phi=-\frac12{\rm arccot\ }f(l_{n_1}, l_{n_2}, \theta),$$
$$f(l_{n_1}, l_{n_2}, \theta)
=-\frac{l_{n_1}^2-l_{n_1}^{-2}l_{n_2}^{-4}}{2(1-l_{n_2}^{-4})}\cot\theta -\frac{l_{n_1}^2l_{n_2}^{-4}-l_{n_1}^{-2}}{2(1-l_{n_2}^{-4})}\tan\theta.$$
And
$$\partial_x\tilde \phi=\frac{ \partial_xf}{2(1+f^2)},\quad \partial_t\tilde \phi=\frac{ \partial_tf}{2(1+f^2)},$$
\begin{equation}\label{2jieddd}\partial_x^{m_1}\partial_t^{m_2}\tilde \phi =\frac{-f\partial_x^{m_1}f\partial_t^{m_2}f}{(1+f^2)^2} +\frac{\partial_x^{m_1}\partial_t^{m_2}f}{2(1+f^2)},\quad m_1+m_2=2.\end{equation}

Assume $|\theta-\frac\pi2|>\max\{l_{n_1}^{-\eta},l_{n_2}^{-\eta}\}$.
Then $|\cot\theta|>\max\{l_{n_1}^{-\eta},l_{n_2}^{-\eta}\}$.
And hence $|f(l_{n_1}, l_{n_2}, \theta)|\geq C l_{n_1}^{2-\eta}$, which leads to  $|\tilde \phi|\leq C l_{n_1}^{-2+\eta}$.
Moreover, a direct computation gives
$$|\partial_x^{m_1}\partial_t^{m_2} f|\leq |f|^{1+(m_1+m_2)\eta},\quad m_1+m_2=1,2$$
and then $$|\partial_x^{m_1}\partial_t^{m_2}\tilde \phi|\leq Cl_{n_1}^{-2+(m_1+m_2+1)\eta},\quad m=0,1,2.$$

Now we assume $l_{n_1}^\mu\geq l_{n_2}$.
Since $l_{n_1}\gg l_{n_2}$, it is not hard to see
$$\partial_x^{m_1}\partial_t^{m_2}f\approx \partial_x^{m_1}\partial_t^{m_2}(l_{n_1}^2\cot\theta) +\partial_x^{m_1}\partial_t^{m_2}(l_{n_1}^2l_{n_2}^{-4}\tan\theta),\quad m_1+m_2=0,1,2.$$
Moreover, for $m_1+m_2=0$, it holds that
$$|f|\geq C (l_{n_1}^2\cot\theta+l_{n_1}^2l_{n_2}^{-4}\tan\theta)\geq Cl_{n_1}^2l_{n_2}^{-2}\geq C l_{n_1}^{2-2\mu}.$$
Hence by the assumptions on the estimates of $\partial_x^{m_1}\partial_t^{m_2}l_{n_1}$, $\partial_x^{m_1}\partial_t^{m_2}l_{n_2}$ and $\partial_x^{m_1}\partial_t^{m_2}\theta$, we have
$$|\tilde \phi|=\Big|\frac12{\rm arccot\ } f\Big|\leq Cl_{n_1}^{-2+2\mu},$$
$$|\partial_x\tilde \phi|=\Big|\frac{ \partial_xf}{2(1+f^2)}\Big|\leq Cl_{n_1}^{-2+(1+2\mu)\eta},\quad |\partial_t\tilde \phi|=\Big|\frac{ \partial_tf}{2(1+f^2)}\Big|\leq Cl_{n_1}^{-2+(1+2\mu)\eta},$$
$$|\partial_x^{m_1}\partial_t^{m_2}\tilde \phi|\leq \Big|\frac{f\partial_x^{m_1}f\partial_t^{m_2}f}{(1+f^2)^2} \Big| +\Big|\frac{\partial_x^{m_1}\partial_t^{m_2}f}{2(1+f^2)}\Big|\leq Cl_{n_1}^{-2+2(1+2\mu)\eta},\quad m_1+m_2=2.$$

\

\noindent{\bf Proof of Lemma \ref{shape-phi-n+1}.}
Assume $l_{n_2}^\mu\geq l_{n_1}$.
Since $l_{n_2}\gg l_{n_1}$, it is easy to see
\begin{equation}\label{partial_f}
{\partial^m_X f}\approx {-\frac12\partial^m_X(l_{n_1}^2\cot\theta-l_{n_1}^{-2}\tan\theta)},\quad m=0,1,2,~X\in \{x,t\}.
\end{equation}
If $\theta$ varies from $0$ to $\frac\pi2$, then $f$ varies from $-\infty$ to $+\infty$, and then we can choose one of single-valued branches of ${\rm arccot}$ such that $\tilde \phi$ varies from $0$ to $\frac\pi2$.
If $\theta$ varies from $\frac\pi2$ to $\pi$, then $f$ varies from $-\infty$ to $+\infty$, and then we can choose one of single-valued branches of ${\rm arccot}$ such that $\tilde \phi$ varies from $\frac\pi2$ to $\pi$.
Hence, it is easy to see that if $\theta$ changes from $\frac\pi2-l_{n_1}^{-\eta}$ to $\frac\pi2+l_{n_1}^{-\eta}$, then $\tilde \phi$ increases from $l_{n_1}^{-2+\eta}$ to $\pi-l_{n_1}^{-2+\eta}$. Thus we have proved \boxed{(i)}.

If $|\frac{\pi}{2}-\theta|\geq l_{n_1}^{-1+\eta}$, then $|\cot\theta|>l_{n_1}^{-\eta}$ and $|\tan\theta|<l_{n_1}^{-\eta}$, which lead to
$$|f|\geq Cl_{n_1}^{2}|\cot\theta|\geq Cl_{n_1}^{2-\eta}.$$
Then similar to the case $l_{n_1}^\mu\geq l_{n_2}$ in the proof of Lemma \ref{nonresonance_phi}, we have
$$|\partial^m_x\tilde\phi(x,t)|+|\partial^m_t\tilde\phi(x,t)|+|\partial^2_{xt}\tilde\phi(x,t)|\leq  l_{n_1}^{-2+2m\eta}\cdot|\cot\theta|^{-(m+1)},\quad m=0,1,2.$$
Hence we have proved \boxed{(ii)} and \boxed{(vi)} under the condition $|\frac{\pi}{2}-\theta|\geq l_{n_1}^{-1+\eta}$.

For (iii), we assume $l_{n_1}^{-1-\eta}<|\frac{\pi}{2}-\theta|< l_{n_1}^{-1+\eta}$.
By omitting all the small terms of $\partial^2_x\tilde\phi(x,t)$, we have
$$\partial^2_x\tilde\phi(x,t)\approx\frac{l_{n_1}^6\cdot \cot \theta}{1+l_{n_1}^8\cdot\cot^4 \theta}\cdot \left(\partial_x\theta\right)^2.$$
Then the assumption on $\theta$ gives $l_{n_1}^{-1-\eta}<|\cot\theta|< l_{n_1}^{-1+\eta}$. And note that $l_{n_1}^{\eta}\geq |\partial_x\theta|\geq q_n^{-600\tau}\geq l_{n_1}^{-\eta}$. In the last inequality, we have used $l_{n_1}^\eta\geq \lambda^{\eta r_{n}}\geq \lambda^{10^{-3}\varepsilon\eta q_{n}}\geq q_n^{-600\tau}$. Combining all these estimates, we obtain
$$l_{n_1}^7\geq |\partial^2_X\tilde\phi(x,t)|\geq l_{n_1}^{1-7\eta}\gg 1,~X\in\{x,t\},$$ which yields \boxed{(iii)} and \boxed{(vi)} under the condition $l_{n_1}^{-1-\eta}<|\frac{\pi}{2}-\theta|< l_{n_1}^{-1+\eta}$.

For (iv), we assume $|\frac{\pi}{2}-\theta|\leq l_{n_1}^{-1-\eta}$. Then $|\cot\theta|\leq l_{n_1}^{-1-\eta}$, and hence $|\tan\theta|\geq l_{n_1}^{1+\eta}$.
Also note that we have $l_{n_1}^{\eta}\geq |\partial_x\theta|\geq l_{n_1}^{-\eta}$.
By omitting the small terms of $\partial_xf$, we have
$$|\partial_xf|\geq C(l_{n_1}^2+l_{n_1}^{-2}\tan^2\theta)|\partial_x\theta|.$$
Recall (\ref{partial_f}). If $c\leq l_{n_1}^2|\cot\theta|\leq C$, then $|\cot\theta|\approx l_{n_1}^{-2}$ and $|f|\leq C$, which leads to
$$|\partial_X\tilde\phi(x,t)|\geq \frac{|\partial_xf|}{1+C^2}\geq l_{n_1}^{2-\eta},~X\in\{x,t\}.$$
If $l_{n_1}^2|\cot\theta| \leq c$, then $|f|\approx \frac{1}{2}l_{n_1}^{-2}|\tan\theta|$. And hence
$$|\partial_X\tilde\phi(x,t)|\geq \frac{Cl_{n_1}^{-2}\tan^2\theta\cdot|\partial_x\theta|}{1+Cl_{n_1}^{-4}\tan^2\theta} \geq l_{n_1}^{2-\eta},~X\in\{x,t\}.$$
If $l_{n_1}^2|\cot\theta| \geq C$, then $|f|\approx \frac{1}{2}l_{n_1}^{2}|\cot\theta|$. Then we have
$$|\partial_X\tilde\phi(x,t)|\geq \frac{Cl_{n_1}^{2}\cdot|\partial_x\theta|}{1+Cl_{n_1}^{4}\cot^2\theta} \geq l_{n_1}^{-2-\eta}\tan^2\theta\geq l_{n_1}^\eta\gg 1,~X\in \{x,t\}.$$
Combining the above three inequalities together, we have proved \boxed{(iv)}.

For (v), we assume $|\frac{\pi}{2}-\theta|\leq  l_{n_1}^{-2-\eta}$.
It is easy to see $c\leq \frac{|f|}{l_{n_1}^{-2}|\tan\theta|}\leq C$ and then $|f|\geq l_{n_1}^{\eta}\gg 1$. And hence
$$
|\partial_X\tilde\phi(x,t)|\geq \frac{Cl_{n_1}^{-2}\tan^2\theta\cdot|\partial_x\theta|}{1+Cl_{n_1}^{-4}\tan^2\theta}\geq l_{n_1}^{2-\eta},~X\in\{x,t\}.
$$
At this time, we have
$$|\tilde\phi(x,t)-\frac\pi2|<C|f|^{-1}\leq Cl_{n_1}^{-\eta}.$$ Then we obtain \boxed{(v)}.

For (vi), it remains to prove the case $|\frac{\pi}{2}-\theta|\leq l_{n_1}^{-1-\eta},$ which implies $|\cot\theta|\leq l_{n_1}^{-1-\eta}$ and hence $|\tan\theta|\geq l_{n_1}^{1+\eta}$. Also note that we have $l_{n_1}^{\eta}\geq |\partial_x\theta|\geq l_{n_1}^{-\eta}$. We set $l_{n_1}^2\cot\theta=U.$ Clearly,
\begin{equation}\label{11223}|U|\leq l_{n_1}^{\frac{10}{9}}.\end{equation}
Note that for $X,Y \in \{x,t\}$ it holds that $$|\partial_X U|=|\partial_X (l_{n_1}^2\cot\theta)|=2l_{n_1}(\partial_X l_{n_1})\cdot \cot\theta-l_{n_1}^2\cdot (1+\cot^2\theta)(\partial_X \theta)|\leq Cl_{n_1}^3$$ and $$\begin{array}{ll}|\partial_Y\partial_X U|&=|\partial_Y\partial_X (l_{n_1}^2\cot\theta)|\\&=\left|2(\partial_Y l_{n_1})(\partial_X l_{n_1})\cdot \cot\theta+2l_{n_1}(\partial_Y\partial_X l_{n_1})\cdot \cot\theta+\right.\\& 2l_{n_1}(\partial_X l_{n_1})\cdot (-1-\cot^2\theta)(\partial_Y\theta)-2l_{n_1}(\partial_Y l_{n_1})\cdot (1+\cot^2\theta)(\partial_X \theta)\\&\left.-l_{n_1}^2\cdot (2\cot\theta)(-1-\cot^2 \theta)(\partial_X \theta \partial_Y \theta)+l_{n_1}^2\cdot (1+\cot^2\theta)(\partial_Y\partial_X \theta)\right|\\&\leq Cl_{n_1}^3.\end{array}$$ Then
\begin{equation}\label{lem12-1}|\partial _X f|= |-\frac{1}{2}(1+\frac{1}{U^2})\partial_X U|\leq Cl_{n_1}^3(1+\frac{1}{U^2})\end{equation} and
\begin{equation}\label{lem12-2}|\partial_Y\partial _X f|= \big|\frac{1}{U^3}(\partial_X U)(\partial_Y U)-\frac{1}{2}(1+\frac{1}{U^2})\partial_Y\partial_X U\big|\leq C l_{n_1}^6(\frac{1}{U^2}+\left\vert\frac{1}{U^3}\right\vert).\end{equation}
Recall \eqref{2jieddd}:
$$|\partial_X\partial_Y\tilde \phi|= \Big|\frac{-f\partial_Xf\partial_Yf}{(1+f^2)^2} +\frac{\partial_X\partial_Yf}{2(1+f^2)}\Big|.$$

Now we  consider separately the following 3 cases:

\begin{enumerate}
\item If $c\leq |U|\leq C,$ then $|f|\leq C_1.$ Then \eqref{2jieddd}, \eqref{lem12-1} and \eqref{lem12-2} imply that $|\partial_X\partial_Y\tilde \phi|\leq C_1 l_{n_1}^6.$
\item If $|U| \leq c,$ then $f\approx \frac{1}{2}U^{-1}>1.$ Then the equalities of \eqref{2jieddd} and \eqref{lem12-1} implies $\partial _X f$ is dominated by $-\frac{1}{2}\frac{\partial_X}{U^2}$ and $\partial _Y\partial _X f$ by $\frac{\partial _Y U\partial _X U}{U^3}-\frac{1}{2}\frac{\partial_X\partial_Y U}{U^2}.$

Then \eqref{lem12-2} imply that
$$\begin{array}{ll}|\partial_X\partial_Y\tilde \phi|= \Big|\frac{-2f\partial_Xf \partial_Yf+\partial_X\partial_Yf(1+f^2)}{2(1+f^2)^2} \Big|&\leq C\frac{U^{-4}}{(1+U^{-2})^2}l_{n_1}^6\leq Cl_{n_1}^6.
\end{array}$$
\item If $|U| \geq C,$ then $1\leq |f|\approx |\frac{1}{2}U|\leq Cl_{n_1}^{\frac{10}{9}}$ (by \eqref{11223}). Then \eqref{2jieddd}, \eqref{lem12-1} and \eqref{lem12-2} directly imply $|\partial_X\partial_Y\tilde \phi|\leq C l_{n_1}^7.$
\end{enumerate}
Then we finish the proof of \boxed{(vi)}.

\noindent{\bf Proof of (\ref{five-piece}).}
Recall we have defined
$$g_{n+2}(x,t)=\phi(x,t)+g_{n+1,2}(x,t),$$
with $\phi(x,t)=-\frac12{\rm arccot}f(\|A_{r_{n+1}^+-k}(T^kx,t)\|,\|A_{k}(x,t)\|,\frac\pi2-g_{n+1,1}(T^kx,t))$, $g_{n+1,1}=s_{r_{n+1}^+-k}-u_k$, $g_{n+1,2}=s_k-u_{r_{n+1}^{-}}$ and $\|g_{n+1,j}-g_{n,j}\|_{C^2}\leq \lambda^{-\frac32r_{n}}$.
Then $g_{n+1,1}(T^kx,t)$ satisfies (\ref{I-zero-derivative})-(\ref{I-second-derivative}). Recall $d_{n+1}(k)=c_{n+1,1}+k\alpha-c_{n+1,2}$.
And the zeros of $g_{n+1,1}(T^kx,t)$ and $g_{n+1,2}(x,t)$ approximate to $c_{n+1,1}+k\alpha=c_{n+1,2}+d_n$ and $c_{n+1,2}$ with an error of $O(\lambda^{-\frac32r_{n}})$, respectively.
Hence we can apply Lemma \ref{shape-phi-n+1} to $g_{n+2}$ and then (\ref{five-piece}) follows.

\

\noindent{\bf Proof of (\ref{c-n+2-c-n+1}) and (\ref{g_nmin}).}
We have (\ref{defgn+2}) and (\ref{partial_f}). Note that $g_{n+1,1}(T^kx,t)$ satisfies (\ref{I-zero-derivative})-(\ref{I-second-derivative}).
Then to solve $g_{n+2}(x,t)=0$, we can rewrite $g_{n+2}$ as
\begin{equation}\label{gn+21}g_{n+2}=-a_1l_{k}^{-2}(a_2(x+k\alpha-e_1))^{-1}+a_3(e_2-x),\end{equation}
where $q_{n}^{-3}\leq a_1,a_2,a_3\leq q_{n}^{3}$, $l_{k}\geq \lambda^{r_{n}}$, and $e_j$ ($j=1,2$) denotes the zero of $g_{n+1,j}$ satisfying $|e_j-c_{n+1,j}|<\lambda^{-\frac32 r_{n}}$.
Then we may have either no solution or two solutions as follows
$$c_{n+2,2},c'_{n+2,2}=e_1-k\alpha+\frac12(e_2-e_1+k\alpha\pm \sqrt{(e_2-e_1+k\alpha)^2-4a_1a_2^{-1}a_3^{-1}l_k^{-2}}).$$
For the no solution case, $g_{n+1}$ reaches its extremum at $$c_{n+2,2}=c'_{n+2,2}=e_2-a_1^\frac12a_2^{-\frac12}a_3^{-\frac12}l_k^{-1} =e_1-k\alpha+a_1^\frac12a_2^{-\frac12}a_3^{-\frac12}l_k^{-1}.$$
In either case, one can easily verify $$|c_{n+2,2}-c_{n+1,2}|<l_k^{-\frac34}\leq \lambda^{-\frac34 r_n},\quad |c'_{n+2,2}-c_{n+1,1}+k\alpha|<l_k^{-\frac34}\leq \lambda^{-\frac34 r_n}.$$
Similarly, we can also obtain $|c_{n+2,1}-c_{n+1,1}|< \lambda^{-\frac34 r_n}$.

Moreover, one notes that if $|e_1-e_2-k\alpha|=0,$ then \eqref{gn+21} implies
$$\min\limits_{x\in I_{n+1}}|g_{n+2}(x,t)(\rm{mod}~\pi)|\geq 2\sqrt{a_1a_3}l_k^{-1},$$ which implies \eqref{g_nmin}.
\

\noindent{\bf Proof of (\ref{range-of-distance}).}
From (\ref{five-piece}), We can find $b_{n+2,1}$ and $b_{n+2,2}$ such that $\partial_xg_{n+2}(b_{n+2,j},t)=0$.
As in the proof of Lemma \ref{shape-phi-n+1}, by omitting the small terms, $\partial_x f$ can be rewritten as
$$\partial_x f\approx \frac12(l_{k}^2+l_k^{-2}\tan^2\theta')\partial_x\theta',$$
with $l_k=\|A_k(x,t)\|$ and $\theta'=\frac\pi2 -g_{n+1,1}(T^kx,t)$.
Then $\partial_xg_{n+2}(x,t)=0$ is equivalent to
\begin{equation}\label{pag0}
\frac{(l_{k}^2+l_k^{-2}\tan^2\theta')\partial_x\theta'} {2+2(l_{k}^2\cot\theta'+l_{k}^{-2}\tan\theta')^2}=-\partial_xg_{n+1,2}.
\end{equation}
Note that $q_{n+1}^{-4}\leq |\frac{\partial_xg_{n+1,2}}{\partial_x\theta'}|\leq q_{n+1}^{4}$. By the same discussion on the comparison between $l_k^2$ and $l_k^{-2}\tan\theta'$, we have $s_1l_k^{-1}<|\cot\theta'| <s_2l_k^{-1}$ with $q_{n+1}^{-4}<s_1,s_2<q_{n+1}^{4}$.
Hence it holds that
$$C_1s_1l_k^{-1}<|b_{n+2,j}-c_{n+1,2}-d_{n+1}|<C_2s_1l_k^{-1}.$$
Let $b_{n+2,1}$ be the one with $b_{n+2,j}-c_{n+1,2}-d_{n+1}>0$ and $b_{n+2,2}$ be the one with $b_{n+2,j}-c_{n+1,2}-d_{n+1}<0$.
Then by Lemma \ref{shape-phi-n+1},
$$g_{n+2}(b_{n+2,1})=g_{n+1,2}(b_{n+2,1})+Cs_1^{-1}l_{k}^{-1},\quad g_{n+2}(b_{n+2,2})=g_{n+1,2}(b_{n+2,2})+\pi-Cs_1^{-1}l_{k}^{-1}.$$
Note $g_{n+1,2}(b_{n+2,2})-g_{n+1,2}(b_{n+2,1})\leq Cs_2(b_{n+2,2}-b_{n+2,1})\leq Cs_2s_1l_k^{-1}$.
Hence
$$ \pi -C_1s_1l_k^{-1} \leq g_{n+2}(b_{n+2,2})-g_{n+2}(b_{n+2,1})\leq \pi -C_2s_1l_k^{-1}.$$

\

\noindent{\bf Proof of (\ref{2ndderivative}).}
We assume $|\partial_x g_{n+2}|\leq 100^{-1}q_{n+1}^{-2}$.
Similar to (\ref{pag0}), by omitting the small terms, we have
$$q_{n+1}^{-6}\leq \frac{l_{k}^2+l_k^{-2}\tan^2\theta'} {2+2(l_{k}^2\cot\theta'+l_{k}^{-2}\tan\theta')^2}\leq q_{n+1}^6,$$
which implies $q_{n+1}^{-6}l_k^{-1}<|\cot\theta'|<q_{n+1}^6l_{k}^{-1}$.
Then by (iii) of Lemma \ref{shape-phi-n+1}, we have $|\partial^2_x g_{n+2}|\geq l_{k}^{1-10\mu}.$

\

\noindent{\bf Proof of (\ref{g-derivative-t}).}
Recall
$g_{n+2}(x,t)=\phi(x,t)+g_{n+1,2}(x,t),$
with

$\phi(x,t)=-\frac12{\rm arccot}f(\|A_{r_{n+1}^+-k}(T^kx,t)\|$, $\|A_{k}(x,t)\|,\frac\pi2-g_{n+1,1}(T^kx,t))$, $g_{n+1,1}=s_{r_{n+1}^+-k}-u_k$, $g_{n+1,2}=s_k-u_{r_{n+1}^{-}}$ and $\|g_{n+1,j}-g_{n+1}\|_{C^2}\leq \lambda^{-\frac32r_{n}}$.

It is not hard to see that $\partial_t\phi(x,t)$ and $\partial_t g_{n+1,2}(x,t)$ have the same positive sign. Hence we have
$$\partial_t g_{n+2}(x,t)\geq \partial_t\phi(x,t)+\partial_t g_{n+1,2}(x,t)>\partial_t g_{n+1,2}(x,t).$$ Thus, during the process from some Case to Case $\textbf{2}$, \( \partial_t g_{\cdot}\) increases, which implies the only possible loss from the first step occurs in the transition from Case ${\bf X}_p$  to Case ${\bf 1}_{p+1}$. By the fact $|g_{p}-g_{p+1}|_{C^2}\leq \lambda^{-r_{p-1}},$ we obtain $$\partial_t g_{n+2}(x,t)>\partial_t g_{0,2}(x,t)-\sum\limits_{p=1}^{\infty}\lambda^{-r_{p-1}}>1-\frac{1}{10}>\frac{1}{10}.$$ Then we obtain the lower bound. For the upper bound, one can obtain it in a similar way as to obtain \eqref{five-piece}.

\

\noindent{\bf Proof of (\ref{g_nondegenerate}).}
Note that at this case, there exist $b_{n+2,1}\neq b_{n+2,2}$ such that $\partial_xg_{n+2}(b_{n+2,j})=0$.
Without loss of generality, we only consider the case with $x\in I_{n+1,2}$ and  $b_{n+2,1}<c'_{n+2,2}<b_{n+2,2}<c_{n+2,2}$. Other cases can be dealt with similarly.

Recall $X=\{c_{n+2,2},c'_{n+2,2}\}$.
Let $r={\rm dist}(x,X)$. 
We divide $B(X,r)$ into three parts: $[c'_{n+2,2}-r,b_{n+2,1}]\cup [b_{n+2,1},c'_{n+2,2}]\cup [c'_{n+2,2},c_{n+2,2}+r]$.
If $x\in (c'_{n+2,2}-r,b_{n+2,1}]$, then $$g_{n+2}(x)=(g_{n+2}(x)-g_{n+2}(b_{n+2,1}))+g_{n+2}(b_{n+2,1}).$$
Note that $g_{n+2}(b_{n+2,1})=g_{n+1,2}(b_{n+2,1})+Cs_1^{-1}l_{k}^{-1}$ and $|b_{n+2,1}-c'_{n+2,2}|\leq |b_{n+2,1}-b_{n+2,2}|<Cs_1l_{k}^{-1}$.
And then
\begin{equation}\label{lower1}|g_{n+2}(b_{n+2,1})|\geq |g_{n+1,2}(b_{n+2,1})|+Cs_1^{-1}l_{k}^{-1}\geq C|b_{n+2,1}-c'_{n+2,2}|^3.\end{equation}
On the other hand, from the proof of (\ref{range-of-distance}),  if $|\partial_xg_{n+2}|<q_{n+1}^{-2}$, then we have $s_4l_k^{-1}<|\cot\theta|<s_5l_k^{-1}$ with $q_{n+1}^{-4}<s_4,s_5<q_{n+1}^4$, which implies $|\partial_x^2g_{n+2}|\geq l_{k}^{1-7\eta}$ due to Lemma \ref{shape-phi-n+1} (iii).
Hence for $x\in (c'_{n+2,2}-r,b_{n+2,1}]$, there exist $d'>0$ and $d''>0$ such that $|x-b_{n+2,1}|=d'+d''$ and
$$|g_{n+2}(x)|\geq q_{n+1}^{-2}d'+Cd''^2\geq C|x-b_{n+2,1}|^3.$$
And then
$$|g_{n+2}(x)|\geq C|x-c'_{n+2,2}|^3.$$

If $x\in [c'_{n+2,2},c_{n+2,2}+r]$, then similarly as above there exist $d'>0$ and $d''>0$ such that ${\rm dist} (x,X)=d'+d''$ and
$$|g_{n+2}(x)|\geq q_{n+1}^{-2}d'+Cd''^2\geq C{\rm dist} (x,X)^3.$$

If $x\in [b_{n+2,1},c'_{n+2,2}]$, then $g_{n+2}$ increases very fast from near $0$ to near $\pi$. Moreover, we have (\ref{lower1}). Then the result follows.

\

\noindent{\bf Proof of Lemma \ref{technical_argument}.}
By singular value decomposition, it suffices to consider the case with
$$E_2=\left(
        \begin{array}{cc}
          e_2 & 0 \\
          0 & e_2^{-1} \\
        \end{array}
      \right)R_\theta,\quad E_1=\left(
        \begin{array}{cc}
          e_1 & 0 \\
          0 & e_1^{-1} \\
        \end{array}
      \right).
$$
Then $s(E_2)=\frac\pi2-\theta$ and $\tan(E_1^{-1}\cdot s(E_2))=e_1^2\cot\theta$.
Assume $\theta\neq 0$, otherwise it is trivial.
Let $w\in E_1^{-1}\cdot s(E_2)$ be a unit vector.
Then we have
$$
\|Ew\|=(1+e_1^4\cot^2\theta)^{-\frac12}\|E\left(
                                            \begin{array}{c}
                                              1 \\
                                              e_1^2\cot\theta \\
                                            \end{array}
                                          \right)\|
                                          =(e_1^{-2}e_2^2\sin^2\theta+e_1^2e_2^2\cot^2\theta)^{-\frac12}=C\|E\|^{-1}.
$$
Then we have $$|E_1^{-1}\cdot s(E_2)-s(E)|<C\|E\|^{-2}.$$

Now we let $w'\in E_1\cdot s(E)$ be a unit vector.
Then we have
$$\|E_2w'\|=\frac1{\|E_1\cdot s(E)\|}\|E_2E_1\cdot s(E)\|=\frac1{\|E\|\cdot\|E_1\cdot s(E)\|}.$$
 Let $w''\in s(E)$ be a unit vector.
Since $\|E\|\gg \|E_1\|$, it holds that
$$
\|E_2w'\|=\frac{C}{\|E\|\cdot\|E_1w''\|}=\frac{C(1+e_1^4\cot^2\theta)^{-\frac12}}{\|E\|}\cdot\|E_1\left(
                                            \begin{array}{c}
                                              1 \\
                                              e_1^2\cot\theta \\
                                            \end{array}
                                          \right)\|^{-1}\le Ce_2^{-1}=C\|E_2\|^{-1},
$$
which implies the lemma.

\

\noindent\bf{\footnotesize Acknowledgements}\quad\rm
{\footnotesize L. Ge was partially supported by NSF of China (Grants 12371185) and the Fundamental Research Funds for the Central Universities (the start-up fund), Peking University. Y. Wang and J. Xu were supported by the NSF of China (Grants 12271245).}\\[4mm]

\end{document}